\documentclass{lmcs}
\pdfoutput=1

\usepackage{lastpage}
\lmcsdoi{15}{3}{5}
\lmcsheading{}{\pageref{LastPage}}{}{}%
{Jul.~18,~2016}{Jul.~26,~2019}{}

\usepackage[font=small,labelfont=it]{caption}
\input{header.extra}
\title[Coherence for Frobenius pseudomonoids]{Coherence for Frobenius pseudomonoids and the geometry of linear proofs}

\author{Lawrence Dunn}
\address{North Florida Community College}
\email{dunnl@nfcc.edu}

\author{Jamie Vicary}
\address{University of Oxford}
\email{jamie.vicary@cs.ox.ac.uk}

\begin{document}

\begin{abstract}
We present a 3\-dimensional notation for proofs in nonsymmetric multiplicative linear logic with units, with a geometrical notion of equivalence, and without the need for a global correctness criterion or thinning links. We argue that traditional proof nets are the 2\-dimensional projections of these 3\-dimensional diagrams. These results rely on a coherence result for Frobenius pseudomonoids, for which we give a direct combinatorial proof.
\end{abstract}
\maketitle

\section{Introduction}

Frobenius pseudomonoids are higher-dimensional algebraic structures, first studied by Street~\cite{Street_2004}, which categorify the classical algebraic notion of Frobenius algebra~\cite{Kock_2003}. These higher algebraic structures have an important application to logic, since Frobenius pseudomonoids in the bicategory of categories, profunctors and natural transformations, for which the multiplication and unit have right adjoints, correspond to {$*$-autonomous categories}~\cite{Barr_1991, Barr_1995}, the standard categorical semantics for multiplicative linear logic. They also play a central role in topological quantum field theory~\cite{BDSV2, BDSV4, Kock_2003, Tillmann_1998}. Our main result is a combinatorial proof of a coherence theorem for Frobenius pseudomonoids, which does not rely on existing approaches in terms of Morse theory~\cite{KerlerLyubashenko, Lauda:2005}. In the second part of the paper, we apply this coherence theorem to the problem of geometrical proof representation in linear logic, giving a 3\-dimensional notation for proofs with a geometrical notion of equivalence. Our major proofs are formalized using the proof assistant \textit{Globular}~\cite{globular_tool, globular}, with online formal proofs hyperlinked directly from the text.

\subsection{Coherence for Frobenius pseudomonoids}

In category theory, a \textit{coherence result} is a method of identifying, in a uniform way, a broad class of identities which hold with respect to some algebraic object. They are the power-tools of higher-dimensional algebra; difficult to prove, but extremely useful. The goal of Section~\ref{sec:coherence} is to prove a variety of coherence theorems  concerning Frobenius pseudomonoids. Our main result follows in principle from work of Kerler and Lyubashenko~\cite{KerlerLyubashenko} as developed by Lauda~\cite{Lauda:2005}. However, that work has a heavily topological style, employing Morse-theoretical techniques that rely on the relationship between Frobenius pseudomonoids and manifold structures.  Our new proof is fundamentally different, with a purely combinatorial style that makes it more accessible to the logic community. It is also more adaptable, with the methods likely to be extendable to generalized structures that correspond to richer logical frameworks, but which no longer have a topological interpretation. (Indeed, important sub-results such as \autoref{prop:coherencesnake} do not follow immediately from \cite{KerlerLyubashenko, Lauda:2005}.)

Here we give the basic definitions, and the statement of the main coherence theorem. We begin with the definition of pseudomonoid and Frobenius pseudomonoid. These higher algebraic structures are defined in terms of \textit{presentations}, and thanks to the graphical calculus, these presentations can be understood informally: \textit{morphisms} are 2\-dimensional tiles from which we can build composite pictures; \textit{2\-morphisms} are rewrites that allow parts of pictures  to be transformed locally; and \textit{equations} tell us that certain sequences of rewrites are equal to certain other sequences. Formally, these presentations give the generating data for finitely-presented monoidal bicategories.

\begin{definition}
The \textit{pseudomonoid presentation} \P contains the following data:
\begin{itemize}
\item An object \cat C.
\item Morphisms $m: \cat C \boxtimes \cat C \to \cat C$ and $u: \cat 1 \to \cat C$:
\begin{calign}
\tikzpng[scale=3]{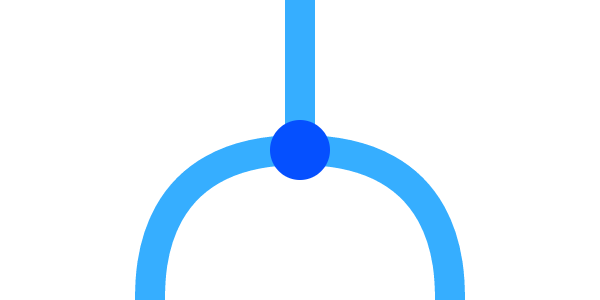}
&
\tikzpng[scale=3]{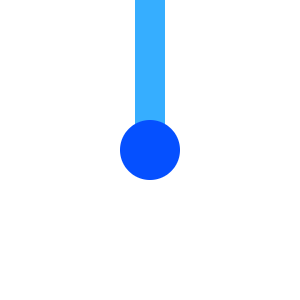}
\end{calign}
\item Invertible 2\-morphisms $\alpha$, $\lambda$, $\rho$:
\begin{calign}
\nonumber
\tikzpng[scale=2]{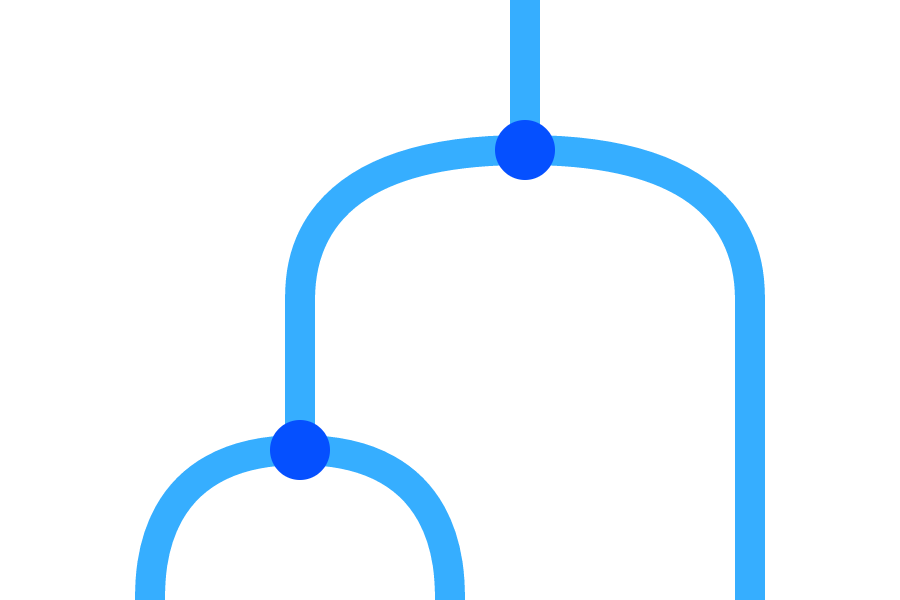}{\xto \alpha}\tikzpng[scale=2]{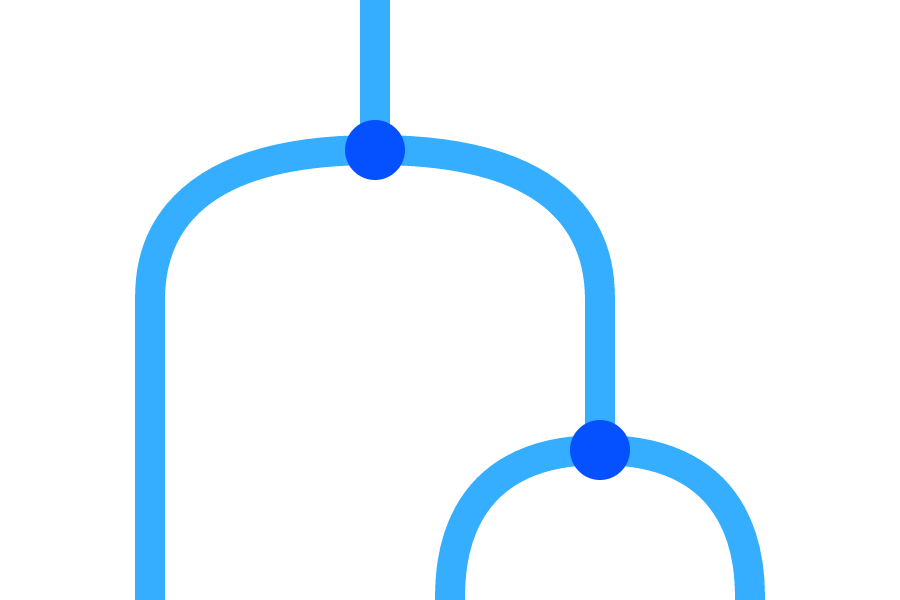}
&
\tikzpng[scale=2]{lambda_source} {\xto \lambda} \tikzpng[scale=2]{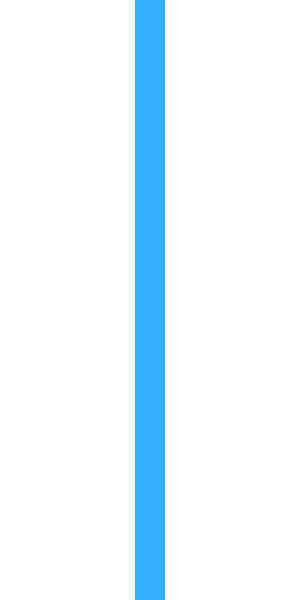}
&
\tikzpng[scale=2]{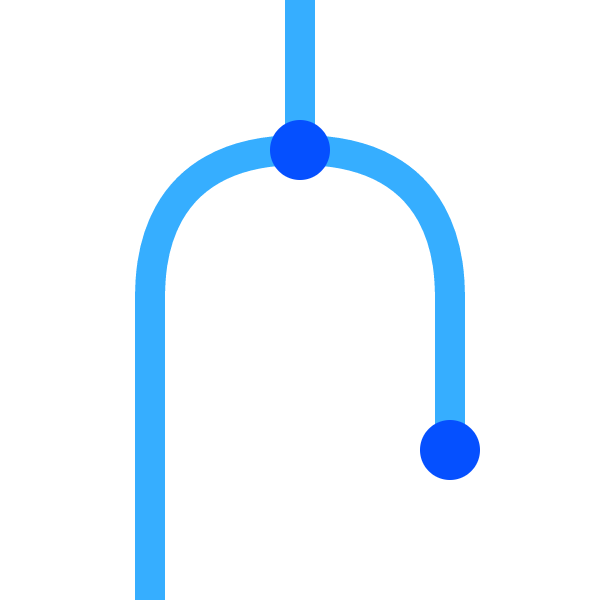} {\xto \rho} \tikzpng[scale=2]{bigidentity}
\end{calign}
\item Pentagon and unit equations:
\begin{calign}
\begin{tz}[xscale=1.7, yscale=1.4, scale=1.5, font=\scriptsize]
\node [inner sep=0pt] (1) at (0,0) {\tikzpng[scale=2]{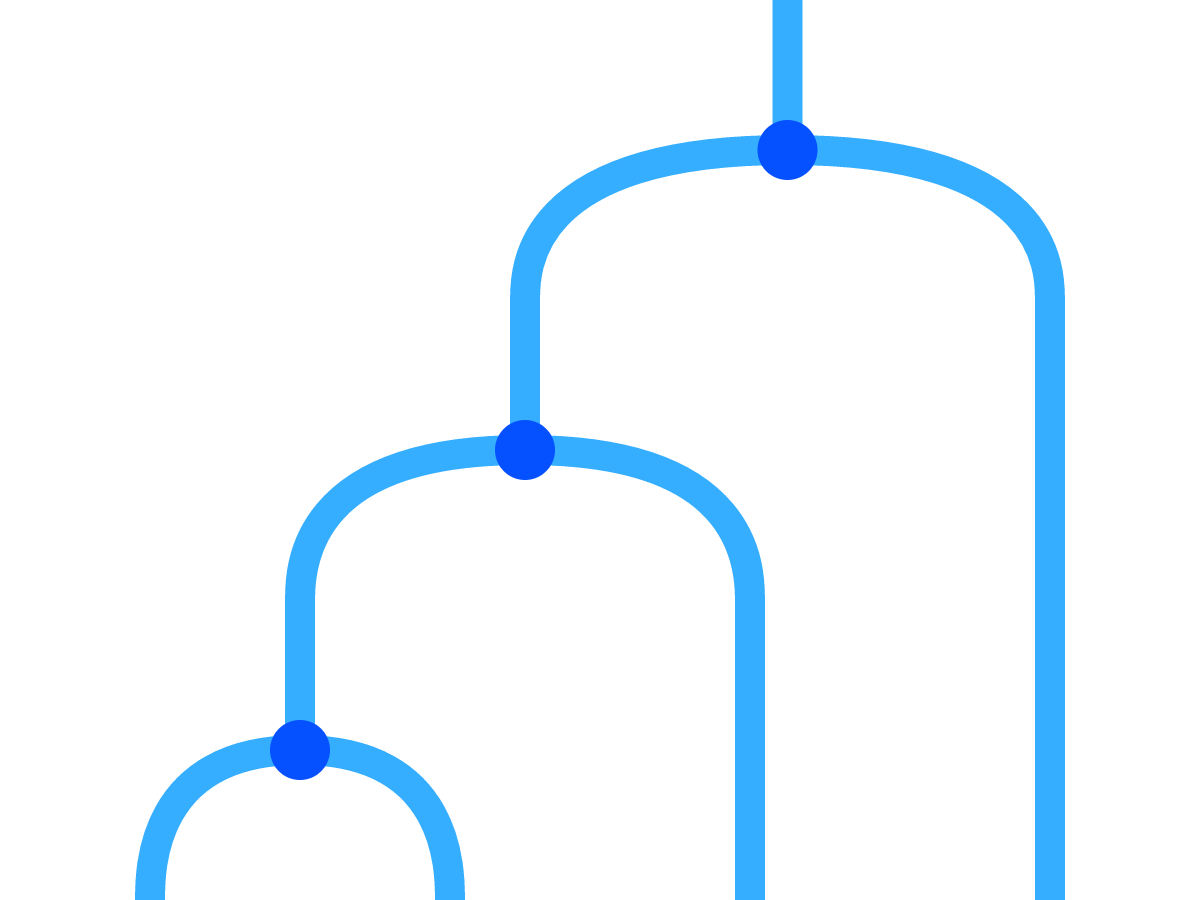}};
\node [inner sep=0pt] (2) at (1,0) {\tikzpng[scale=2]{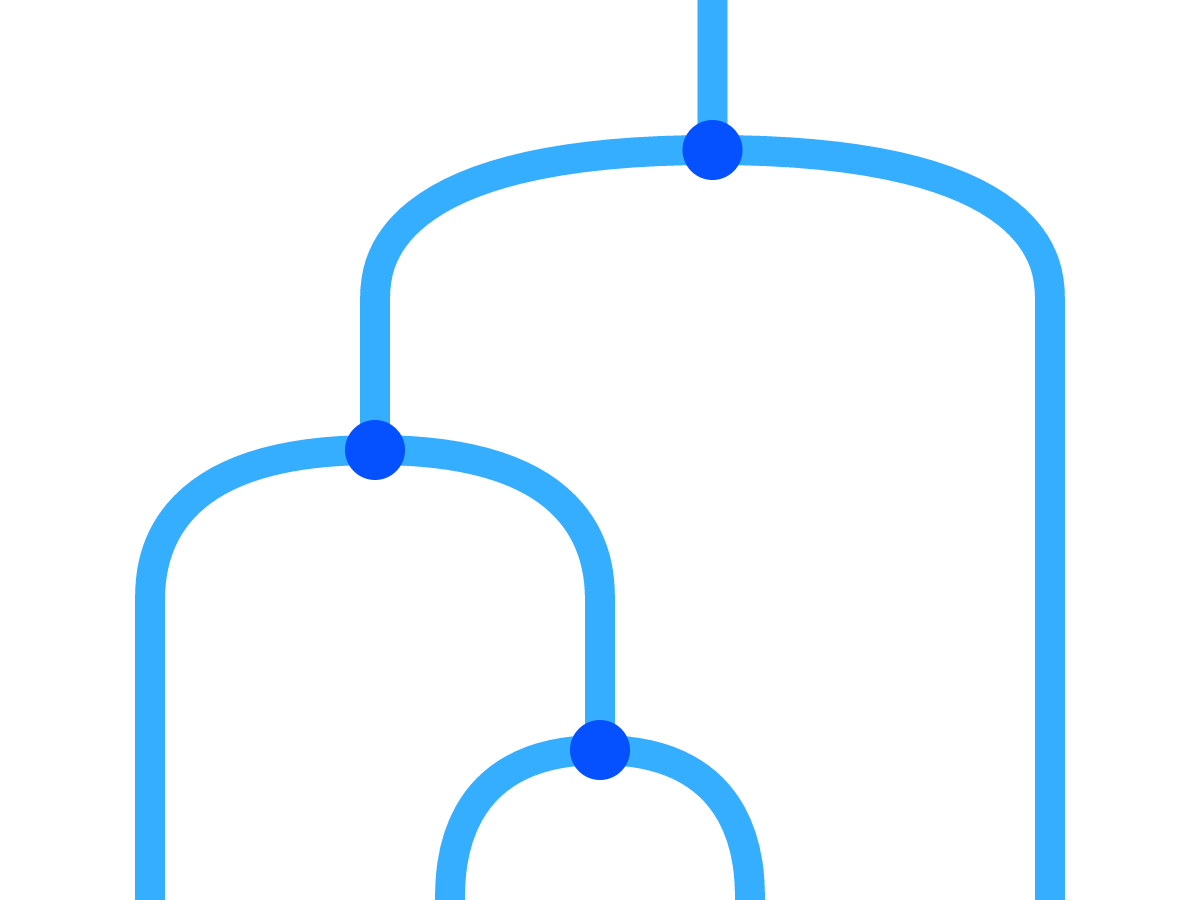}};
\node [inner sep=0pt] (3) at (2,0) {\tikzpng[scale=2]{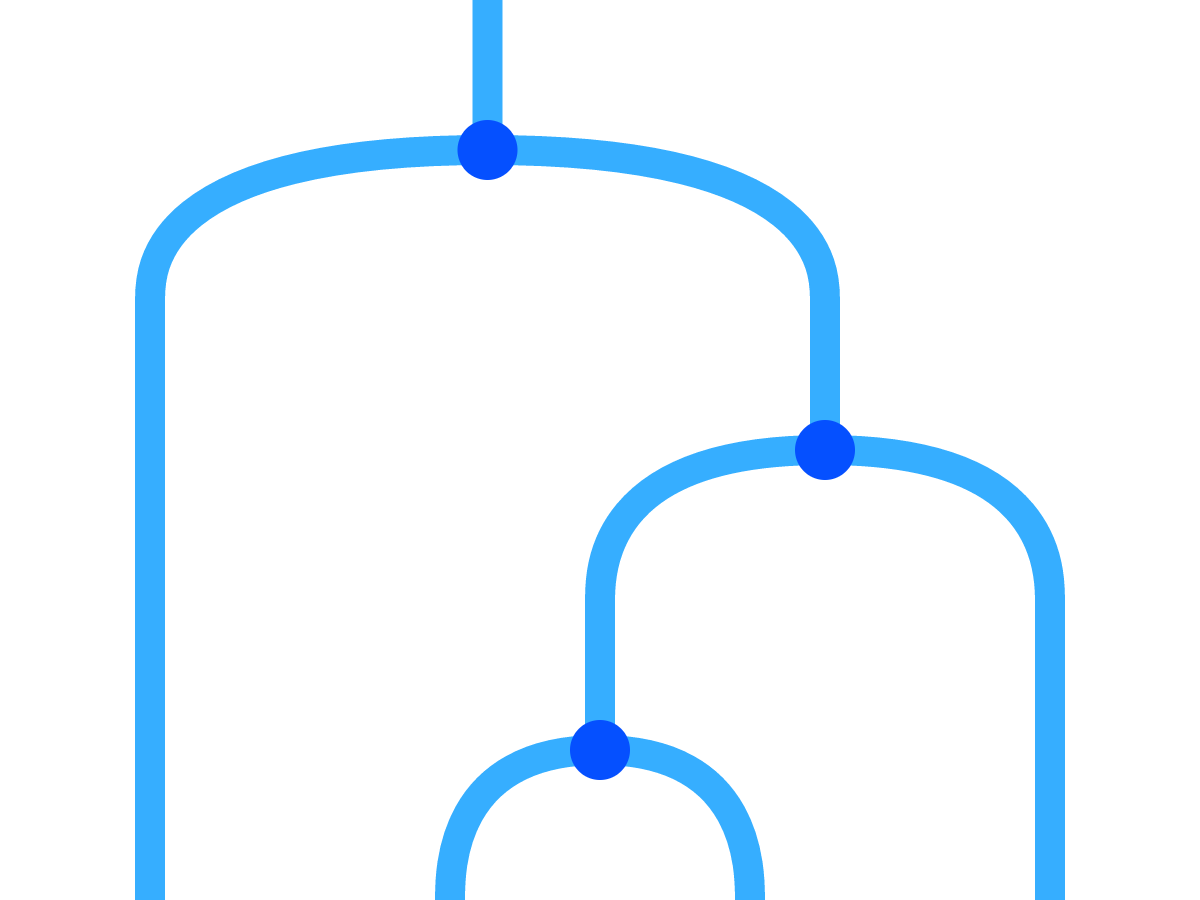}};
\node [inner sep=0pt] (4) at (2,-1) {\tikzpng[scale=2]{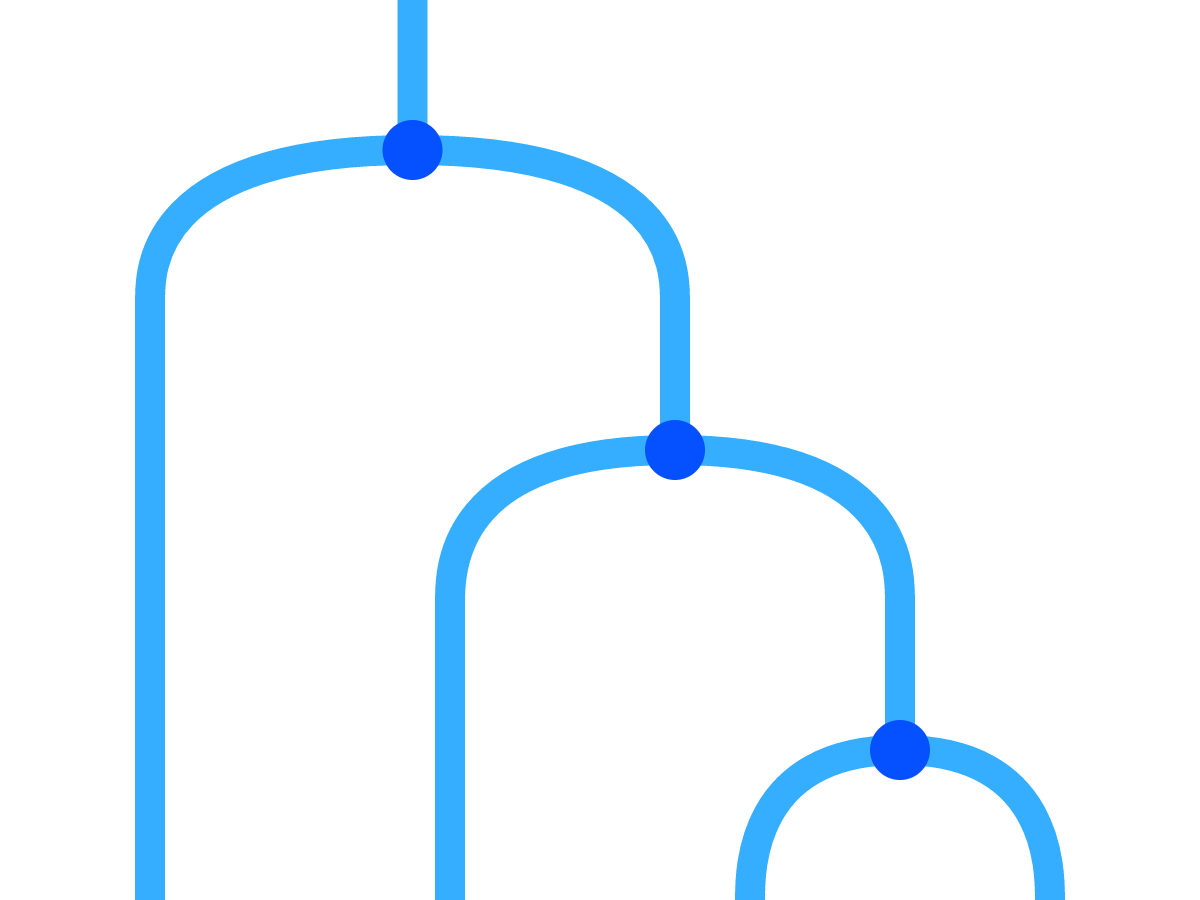}};
\node [inner sep=0pt] (5) at (0,-1) {\tikzpng[scale=2]{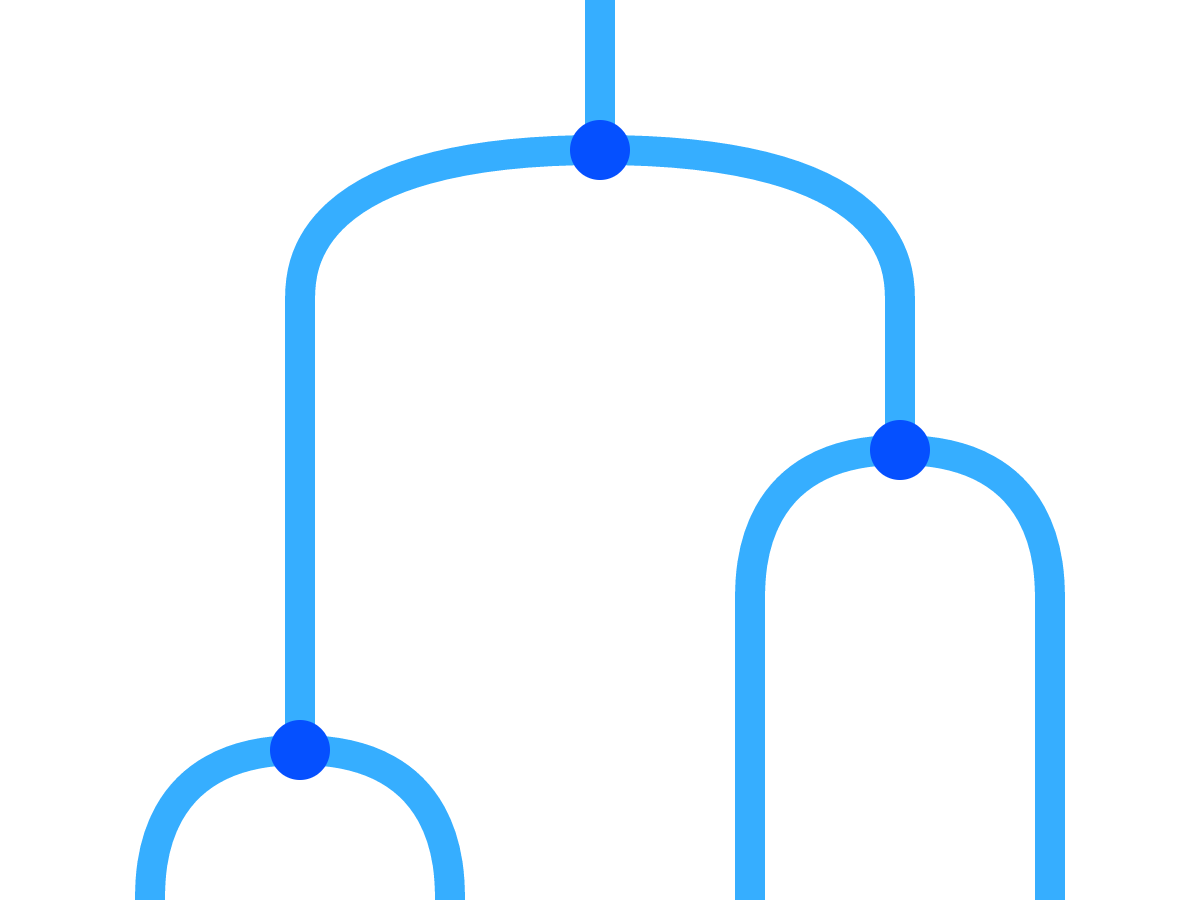}};
\node [inner sep=0pt] (6) at (1,-1) {\tikzpng[scale=2]{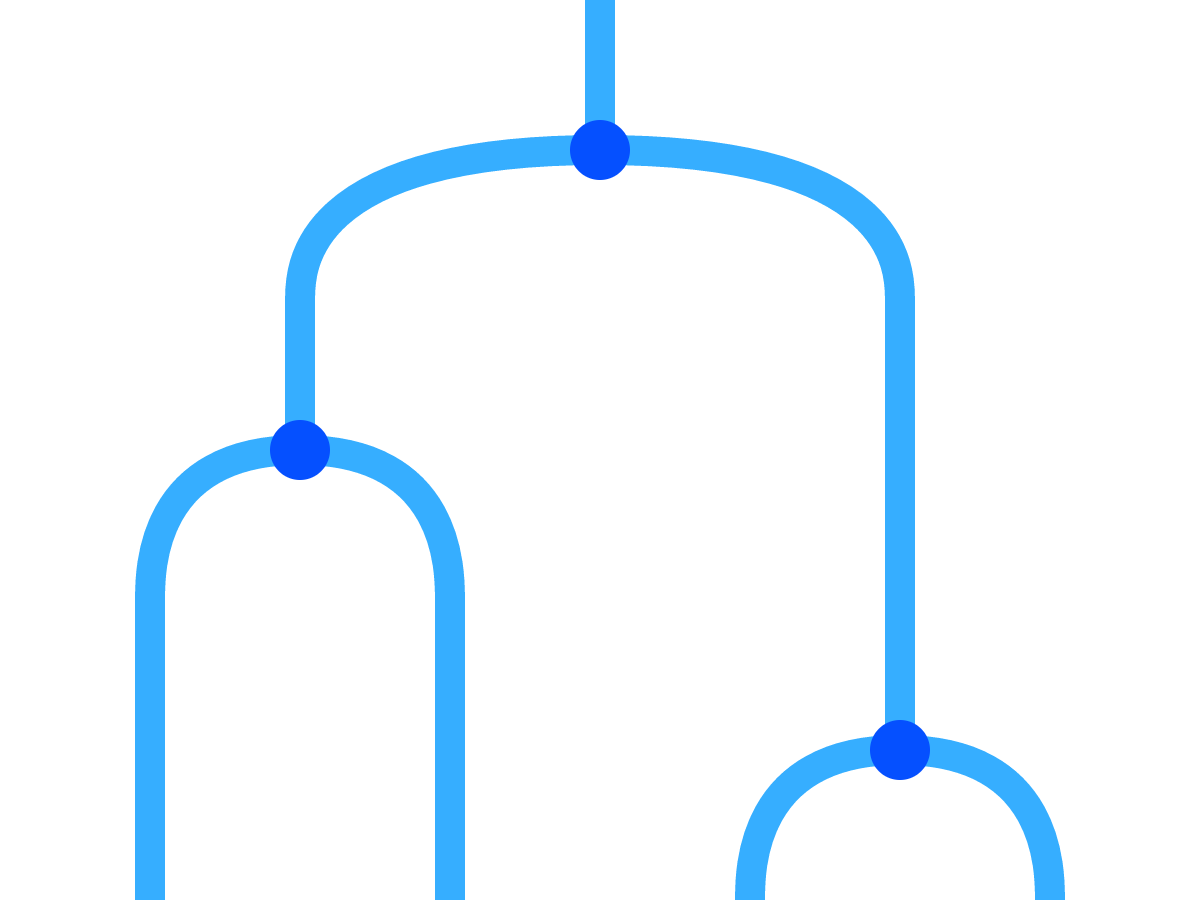}};
\draw [->, shorten <=-2pt, shorten >=-2pt] (1) to node [above] {$\alpha$} (2);
\draw [->, shorten <=-2pt, shorten >=-2pt] (2) to node [above] {$\alpha$} (3);
\draw [->, shorten <=-0pt, shorten >=1pt] (3) to node [right] {$\alpha$} (4);
\draw [->, shorten <=-0pt, shorten >=1pt] (1) to node [left] {$\alpha$} (5);
\draw [->, shorten <=-2pt, shorten >=-2pt] (5) to node [below] {$\sim$} (6);
\draw [->, shorten <=-2pt, shorten >=-2pt] (6) to node [below] {$\alpha$} (4);
\end{tz}
&
\begin{tz}[xscale=1.5, yscale=1.4, scale=1.5, font=\scriptsize]
\node [inner sep=0pt] (1) at (0,0) {\tikzpng[scale=2]{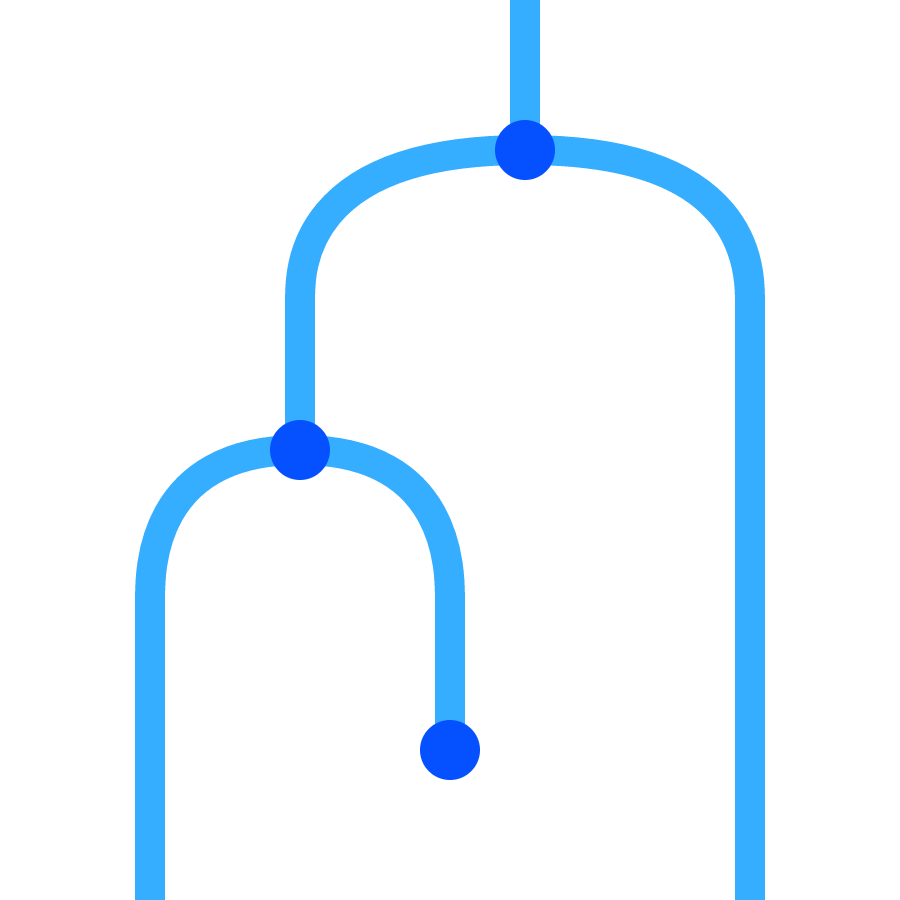}};
\node [inner sep=0pt] (2) at (1,0) {\tikzpng[scale=2]{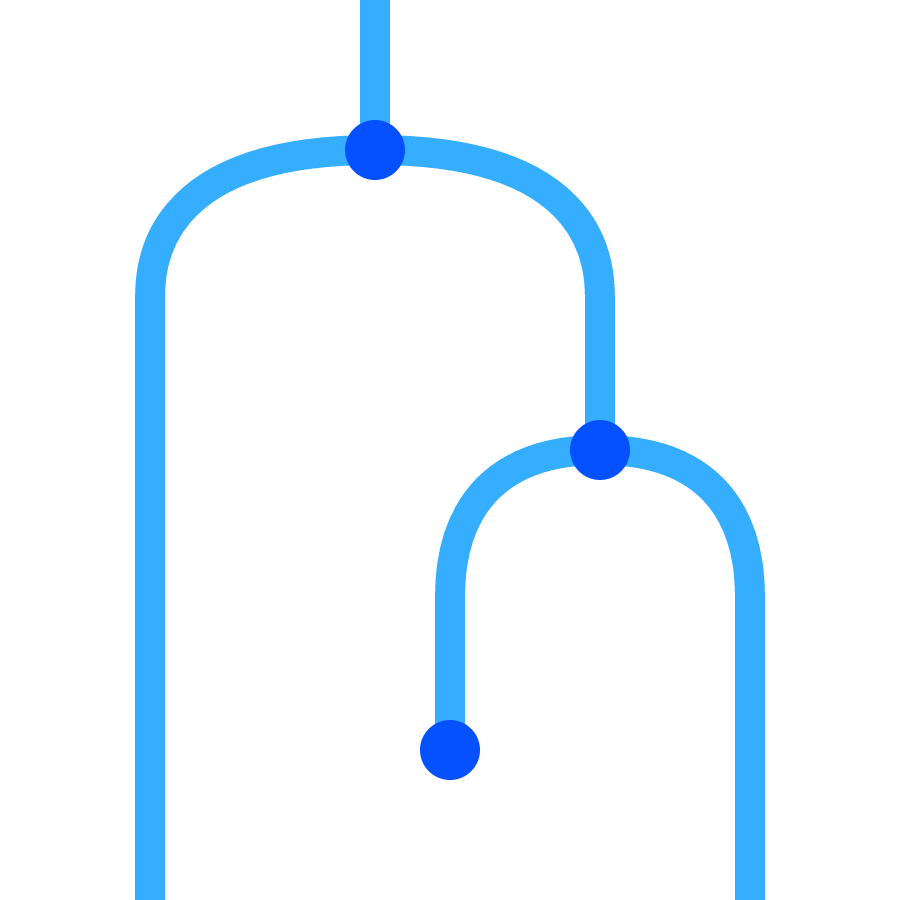}};
\node [inner sep=0pt, minimum height=0cm, inner sep=0pt] (3) at (0.5,1) {\tikzpng[scale=2]{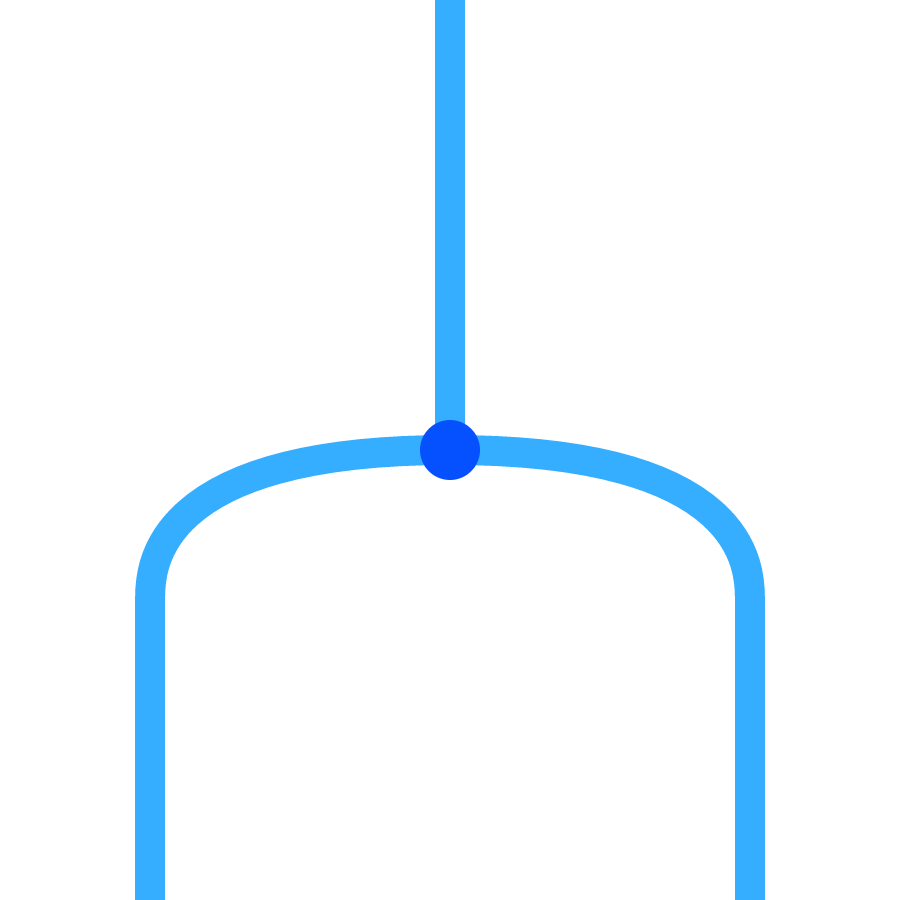}};
\draw [->] (1) to node [below] {$\alpha$} (2) {};
\draw [->] (1) to node [left] {$\rho$} (3);
\draw [->] (2) to node [right] {$\lambda$} (3);
\end{tz}
\end{calign}
\end{itemize}
\end{definition}

\noindent
Note that our pentagon equation has six sides, since we include the interchanger step explicitly. Also, when we say that the 2\-morphisms are \textit{invertible}, that implies the existence of additional generating 2\-morphisms $\alpha ^\inv$, $\lambda^\inv$ and $\rho^\inv$, along with the obvious additional equations.

\begin{definition}
\label{def:sdfrobenius}
The \textit{Frobenius pseudomonoid presentation} \F is the same as the pseudomonoid presentation \P, with the following additional data:
\begin{itemize}
\item Morphisms $\cup : \cat 1 \to \cat C \boxtimes \cat C$ and $f: \cat C \to \cat 1$:
\begin{calign}
\nonumber
\tikzpng[scale=3]{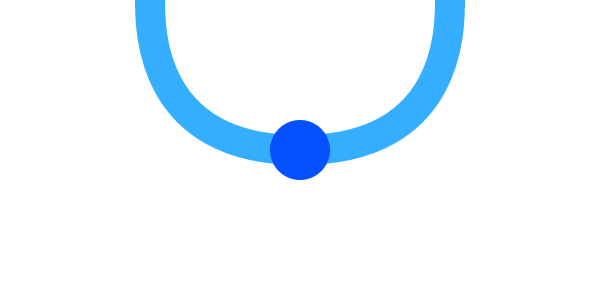}
&
\tikzpng[scale=3,yscale=-1]{unit}
\end{calign}
\item Invertible 2\-morphisms $\mu$ and $\nu$:
\begin{align*}
\tikzpng[scale=2]{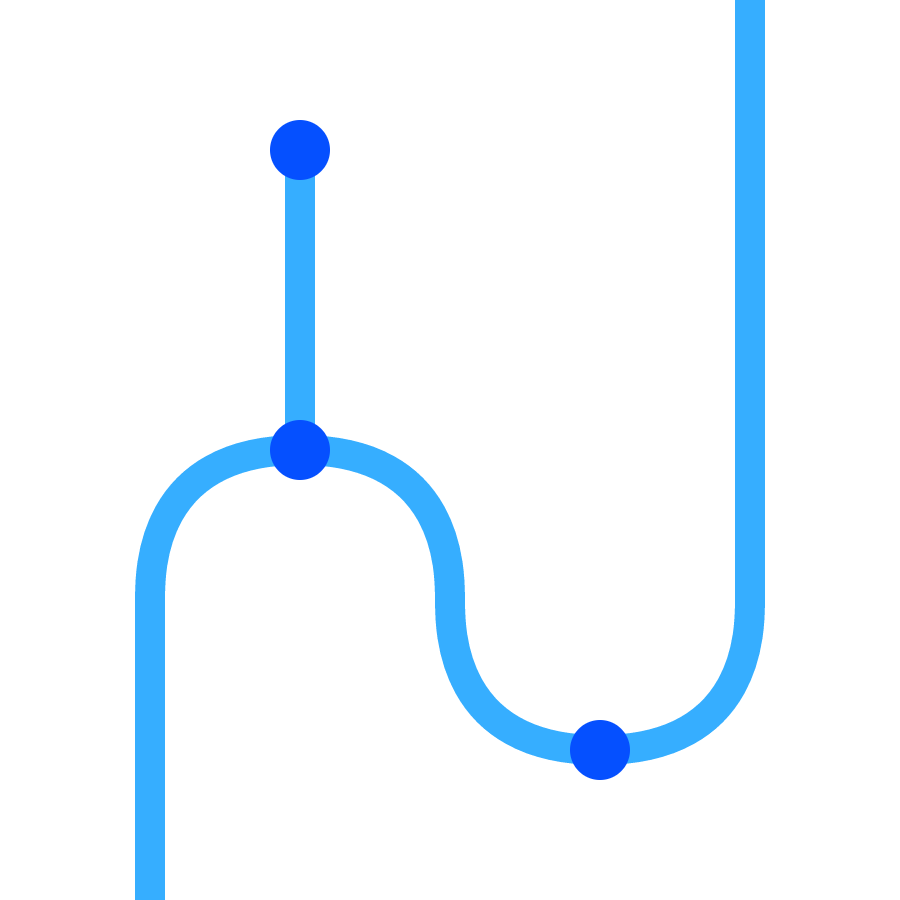} &\xto \mu \tikzpng[scale=2,yscale=1.5]{bigidentity}
&
\tikzpng[scale=2]{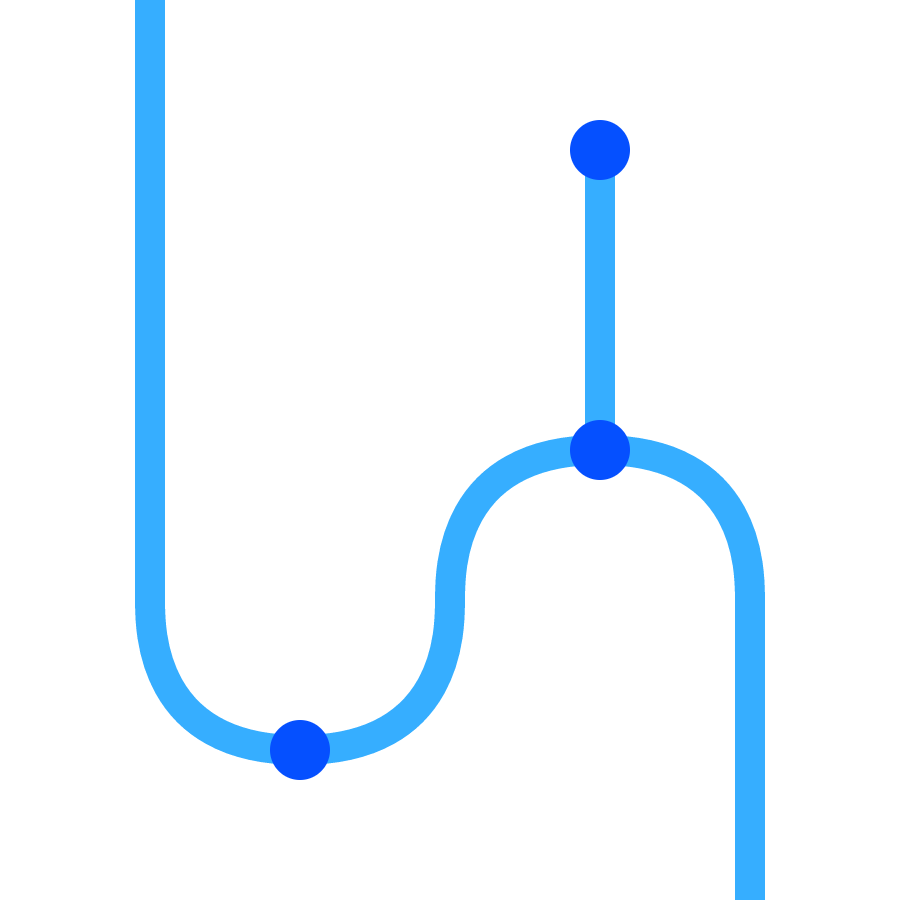} &\xto \nu \tikzpng[scale=2,yscale=1.5]{bigidentity}
\end{align*}
\item Swallowtail equations:
\begin{align*}
\id\quad&=\quad
\tikzpng[scale=2]{cup}
{\xto {\mu ^\inv}}
\tikzpng[scale=2]{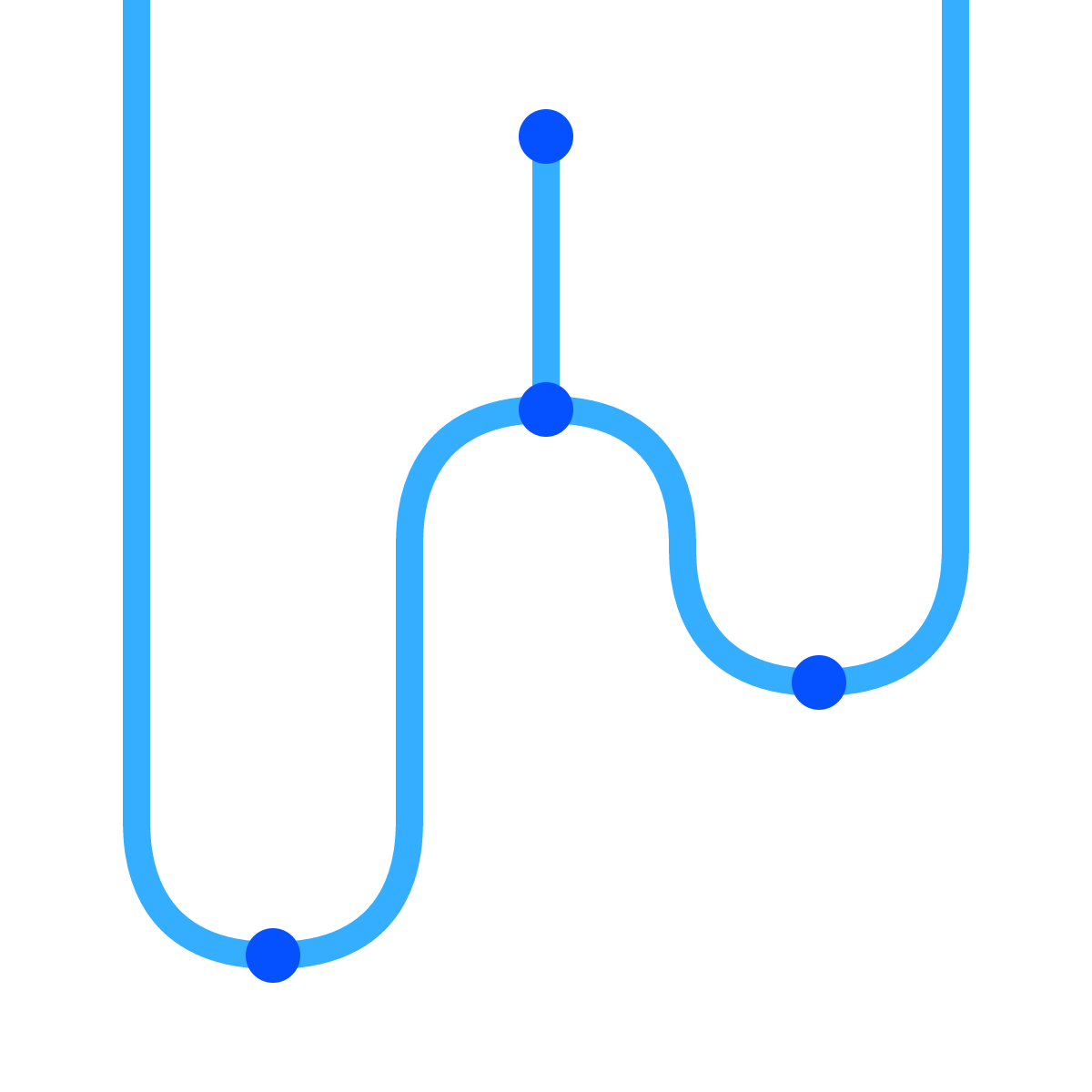}
{\xto \sim}
\tikzpng[scale=2,xscale=-1]{swallowtail1}
{\xto \nu}
\tikzpng[scale=2]{cup}
\\
\id\quad&=\quad
\tikzpng[scale=2]{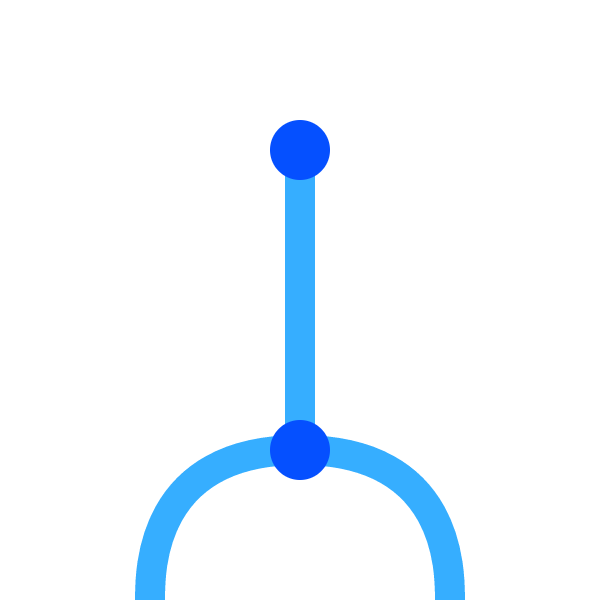}
{\xto {\mu ^\inv}}
\tikzpng[scale=2]{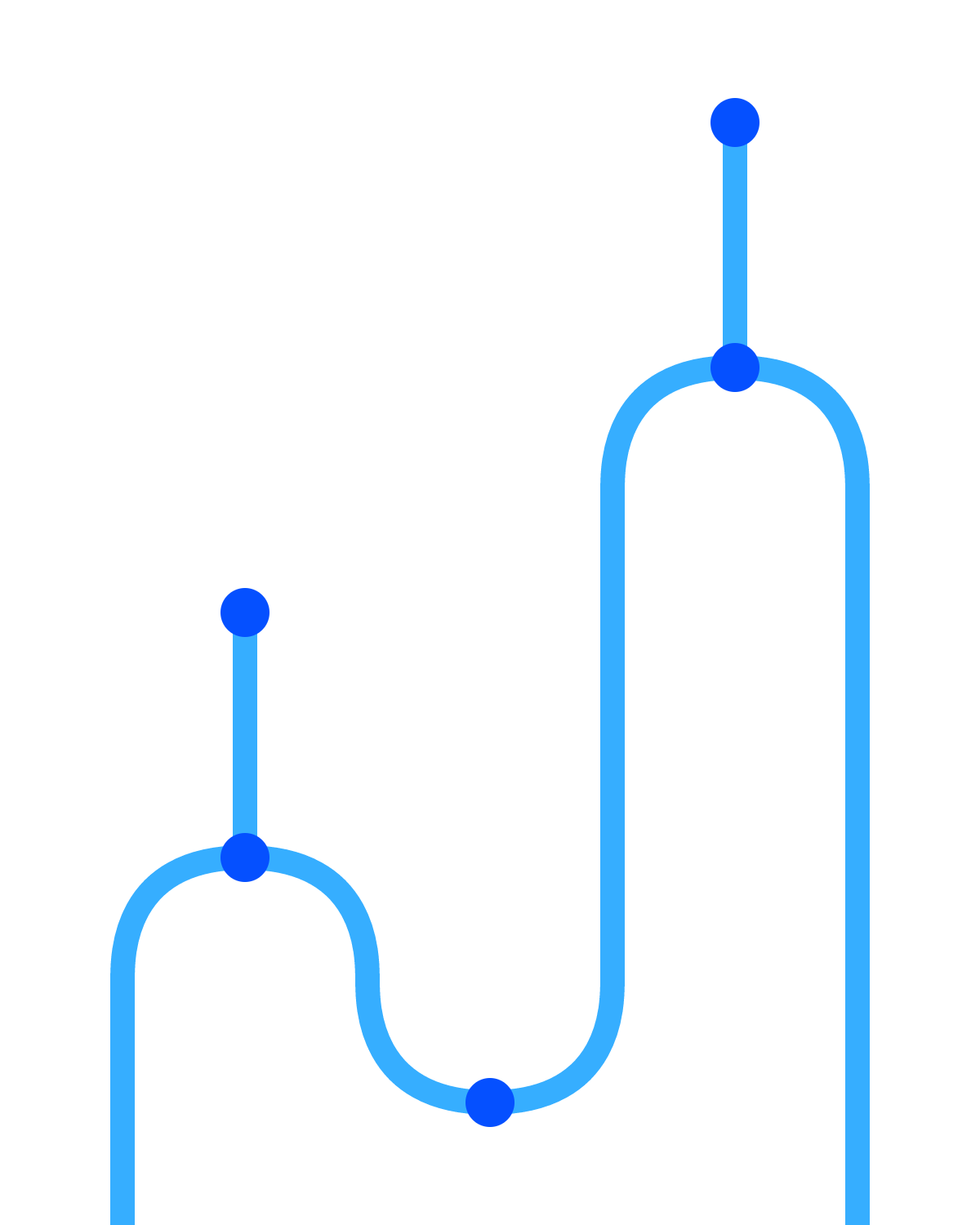}
{\xto \sim}
\tikzpng[scale=2,xscale=-1]{swallowtail2}
{\xto \nu}
\tikzpng[scale=2]{multform}
\end{align*}
\end{itemize}
\end{definition}

\noindent
The swallowtail equations are actually redundant, in the sense that any Frobenius pseudomonoid that does not satisfy them gives rise to one that does~\cite{GlobularSwallowtailCoherence, piotr-thesis}, but for simplicity we include them. This is equivalent to a presentation of Frobenius pseudomonoids in terms of an interacting pseudomonoid and pseudocomonoid; the presentation we have chosen is simpler, and easier to work with.

We now give the statement of the coherence theorem. Given a presentation $\S$, we write $\free \S$ for the free monoidal bicategory generated by this presentation~\cite{CSPthesis}. Given 2\-morphisms $P,Q$ in \free \S, we write $P=Q$ when they are equal as 2\-morphism in $\free \S$; that is, when they equivalent as composites of 2\-morphism generators modulo the equational structure.

\begin{theorem}[Coherence for Frobenius structures]
\label{thm:frobeniuscoherence}
Let $P, Q: X \to Y$ be 2\-morphisms in \free \F, such that $X$ is connected and acyclic, with nonempty boundary. Then $P = Q$.
\end{theorem}

\noindent
The proof strategy is to `rotate' such 2\-morphisms until they are in the image of the obvious embedding $\free \P \to \free \F$, and then to apply coherence for pseudomonoids to conclude the desired result. The central higher-categorical arguments are verified with the web-based proof assistant \href{https://ncatlab.org/nlab/show/Globular}{\emph{Globular}}~\cite{globular}, with direct links provided inline to the formalized proofs.

\subsection{3\-Dimensional proofs for linear logic}

In Section~\ref{sec:3dproofs} we apply these coherence results to \textit{multiplicative linear logic}, a  formal calculus for reasoning about resources~\cite{sep-ll, girard} which is similar to traditional logic, except that resources cannot be duplicated or neglected in the way that propositions can. A central problem is determining when two proofs should be considered equivalent. We describe a scheme for interpreting proofs as \textit{geometrical surfaces} embedded in 3\-dimensional space, and define two surfaces as \textit{equivalent} just when one can be deformed into the other, in sense we make precise. Our main theorem then reads as follows.

\begin{customthm}{\ref{thm:maintheorem}}\em
Two sequent proofs in multiplicative linear logic have equal interpretations in the free $*$-autonomous category just when their surfaces are equivalent.
\end{customthm}

\noindent
The theory of $*$-autonomous categories~\cite{mellies-categorical, seely-linearlogic} is a standard mathematical model for linear logic, so this theorem says that the notion of proof equality provided by the surface calculus agrees with the standard one.

This 3\-dimensional notation for linear proofs has a number of differences with respect to proof nets, the existing geometrical formalism for linear logic: correctness is local, with no long-trip criterion to be verified; there is no need for thinning links, and the associated non-local `jumps'; and it is close to categorical semantics.\footnote{Many of these desirable properties are shared by the \textit{sequent calculus}, the basic algebraic notation for proofs; in a sense, our scheme imports these properties from an algebraic setting to a geometrical one.} Furthermore, we argue that proof nets can be considered the 2\-dimensional `shadow' of the full 3\-dimensional geometry, with the correctness criterion and thinning links arising to compensate for the fact that this `shadow' has lost some essential geometrical data, such as the depth of individual sheets. We illustrate this in \autoref{fig:comparison}, which gives a sequent proof alongside its surface diagram and proof net representations; note that the thinning link, given in the proof net by a dotted line, is completely absent from the surface diagram.
 \begin{figure}
$$
\hspace{-5cm}
\begin{aligned}
\begin{tikzpicture}
\path [use as bounding box] (4.5,-0.75) rectangle +(3.7,5);
\node (4) [right] at (4-\xoff+0.2,1.5) {$\tree[$\otimes$-R]{}{A, B \vdash \bot, A \otimes B}$};
\node (1) [anchor=west] at (0,4.05 -| 4.west) {$\tree[AXIOM]{}{A \vdash A}$};
\node (2) [anchor=west] at (0,3.2 -| 4.west) {$\tree[$\bot$-INT]{}{A \vdash \bot, A}$};
\node [anchor=west] at (5.2,2.55) {$\tree[AXIOM]{}{B \vdash B}$};
\end{tikzpicture}
\end{aligned}
\quad\quad
\begin{aligned}
\begin{tikzpicture}
\path [use as bounding box] (-0.25,-0.75) rectangle +(4.5,5);
\draw [surface] (-\xoff,\yoff) to [out=right, in=145] (1,0) to (2,-0) to [out=30, in=left] (3,\yoff) to [out=25, in=left] (4-\xoff,2*\yoff) to (4-\xoff,2) to [out=up, in=up, looseness=0.5] (-\xoff,2) to (-\xoff,\yoff);
\node [label, below] at (-\xoff,\yoff) {$B$};
\draw [black] (4-\xoff,2*\yoff) to (4-\xoff,2) to [out=up, in=up, looseness=0.5] (-\xoff,2) to (-\xoff,\yoff);
\draw [surface] (\xoff,-\yoff) to [out=right, in=-150] (1,0) to [out=up, in=down] (2,1) to [out=up, in=left] (2.5,1.5) to [out=right, in=up] (3,1) to (3,\yoff) to [out=-35, in=left] (4+\xoff,0) to (4+\xoff,3.5) to [out=up, in=up, looseness=0.5] (\xoff,3.5) to (\xoff,-\yoff);
\draw [basic] (1,0) to [out=up, in=down] (2,1) to [out=up, in=left] (2.5,1.5);
\draw [adjoint] (2.5,1.5) to [out=right, in=up] (3,1) to (3,\yoff);
\draw [black] (4+\xoff,0) to (4+\xoff,3.5) to [out=up, in=up, looseness=0.5] (\xoff,3.5) to (\xoff,-\yoff);
\node [label, below, red] at (3,\yoff) {$\otimes$};
\draw [surface] (2,0) to [out=-35, in=left] (4+3*\xoff,-2*\yoff) to [out=up, in=-15, out looseness=1.8] (1,3.2) to (1,1) to [out=down, in=up] (2,0.0) to (2,0);
\draw [bare] (4+3*\xoff,-2*\yoff) to [out=up, in=-15, out looseness=1.8] (1,3.2);
\draw [basic] (4+3*\xoff,-2*\yoff) to [out=up, in=-15, out looseness=1.8] (1,3.2) to (1,1) to [out=down, in=up] (2,0.0) to (2,0);
\node [label, below] at (\xoff,-\yoff) {$A$};
\node [label, below] at (4+\xoff,0) {$A$};
\node [label, below] at (4+3*\xoff,-2*\yoff) {$\bot$};
\node [label, below] at (4-\xoff,2*\yoff) {$B$};
\end{tikzpicture}
\end{aligned}
\quad\quad
\begin{aligned}
\begin{tikzpicture}
\path [use as bounding box] (-0.25,-0.75) rectangle +(4.5,5);
\node [label, below] at (-\xoff,\yoff) {$B$};
\draw [black, thick] (3,1) to [out=60, in=down] (4-\xoff,2) to [out=up, in=up, looseness=0.5] (-\xoff,2) to (-\xoff,\yoff);
\draw [black, thick] (3,1) to [out=140, in=down] (4+\xoff,3.5) to [out=up, in=up, looseness=0.5] (\xoff,3.5) to (\xoff,-\yoff);
\draw [black, thick] (3,\yoff) to (3,1) node [circle, draw=black, text=black, fill=white, thick, inner sep=0.5pt] {$\otimes$};
\node [label, below] at (3,\yoff) {$A\!\otimes \!B$};
\draw [black, thick, densely dotted] (1,3.2) to (\xoff,3.5);
\draw [black, thick] (4+3*\xoff,-2*\yoff) to [out=up, in=-15, out looseness=1.8] (1,3.2) node [circle, draw=black, text=black, fill=white, thick, inner sep=0.5pt] {$\bot$};
\node [label, below] at (\xoff,-\yoff) {$A$};
\node [label, below] at (4+3*\xoff,-2*\yoff) {$\bot$};
\end{tikzpicture}
\end{aligned}
\hspace{-5cm}
$$
\caption{\label{fig:comparison}A deduction in the sequent calculus, and its
  surface calculus \\
  and proof net representations.}
\end{figure}
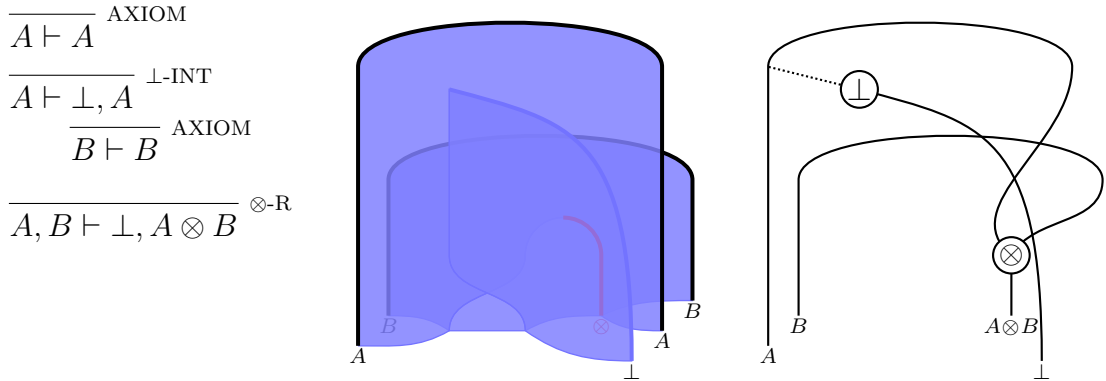

\subsection{Related work}
\label{sec:related}

Our combinatorial proof of coherence for Frobenius pseudomonoids can be considered a categorification of the normal form theorem for Frobenius monoids, described by Abrams~\cite{Abrams_1996}. Our coherence proofs do not make use of the formalism of rewrite systems and confluence, which has difficulties in the higher-dimensional setting~\cite{Mimram_2014} due to the emergence of infinite families of critical pairs from finitely-generated systems. When this formalism becomes better developed, recasting our proofs in such a formalism would likely be fruitful, and expose further interesting structure.

Regarding the application to logic, it was already shown by Street that Frobenius pseudomonoids with right adjoint data correspond to $*$\-autonomous categories~\cite{Street_2004} (up to Cauchy completeness), and from work on the graphical calculus for Gray categories~\cite{barrett-graydiagrams, bartlett-wire, hummon-thesis, CSPthesis}, it is already clear that a 3\-dimensional graphical calculus for $*$\-autonomous categories must therefore exist. However, the consequences of this for proof denotation have not previously been explored. Also, the coherence property of Frobenius pseudomonoids has not previously been discussed in this context; this is essential for the straightforward characterization of equivalence of the 3\-dimensional diagrams.

It is already well-recognized that ideas from topology are relevant to linear logic. The original proof nets of Girard~\cite{girard} are topological objects,  and Melli\`es has shown how the topology of ribbons gives a decision procedure for correctness of proof nets~\cite{mellies-ribbon}. Proof nets allow reasoning about proofs with units, but the formalism is complex, requiring a system of thinning links with moving connections~\cite{blute-coherence}. Hughes~\cite{hughes-simpleproof} gives a variant of proof nets which works well with units, but which lacks the local correctness properties of our notation. Our approach has a local flavour which is shared by the deep-inference model of proof analysis~\cite{guglielmi-deepinference} and the access to monoidal coherence that it allows~\cite{hughes-deepinference}; however, our coherence theorem is strictly more powerful, as it operates in a fragment that combines the $\otimes$ and $\parr$ connectives. We note also the work of Slavnov~\cite{slavnov} on linear logic and surfaces, which involves some similar ideas to the present article, but is technically quite unrelated.

\subsection{Acknowledgements}

 The authors are grateful to Samson Abramsky, Iliano Cervesato, Robin Cockett, Nick Gurski, Samuel Mimram, Robert Seely and Sam Staton for useful comments, and to the organizers of FMCS 2016,  HDRA 2016 and LINEARITY 2016 for the opportunities to present various aspects these results. 3\-dimensional graphics have been coded in Ti\textit{k}Z, and 2\-dimensional graphics have been produced by the proof assistant \textit{Globular}~\cite{globular}. The second author acknowledges funding from a Royal Society University Research Fellowship.

\section{Coherence for Frobenius pseudomonoids}
\def\myscale{2}
\def\miniscale{1.5}
\label{sec:coherence}

\subsection{Overview}

In this section we prove the coherence theorem for Frobenius pseudoalgebras, our main technical result. The mathematical setting is higher algebraic structures in finitely-presented monoidal bicategories. We keep technicalities to a minimum, emphasizing the main ideas behind the key proof steps.

The main technical results are formalized in the web-based proof assistant \emph{Globular}~\cite{globular}, and can be accessed at the address \glob. To our knowledge, this is one of the first uses of \textit{Globular} to prove a new result.

\newcommand\smashstackrel[2]{\ensuremath{\smash{\stackrel {\scriptsize\smash{#1}} #2}}}

Our string diagrams in this section run bottom-to-top. Also, it is convenient to only allow \textit{generic composites} of 1- and 2\-morphisms, meaning that composite 1- and 2\-morphisms have a definite length, and a linear order on their components from `first' to `last'; this is without loss of generality~(see~\cite{bartlett-wire} and~\cite[Section~2.2]{internalbicategories}.) As a point of notation, we write $\smashstackrel \sim \to$ to denote an interchanger 2\-morphism, and $\smashstackrel \sim \twoheadrightarrow$ to denote a composite of interchangers of length at least 0.

\subsection{Definitions}

Just as we can present a monoid by generators and relations, and monoidal category (such as a PRO) by a monoidal signature~\cite{lack-prop}, a similar approach can be used for monoidal bicategories, as developed in the thesis of Schommer-Pries~\cite[Section~2.10]{CSPthesis}. Informally, a presentation can be described as follows.

\begin{definition}
A \emph{presentation} of a monoidal bicategory is a list of generating objects, morphisms, 2\-morphisms and equations, along with appropriate source and target data. 
\end{definition}

\noindent
The sources and targets are given by composites of the data at lower levels. For example, the source and target of a generating 2\-morphism is a composite of the generating objects and morphisms, under tensor product, composition and taking identities.

\begin{definition}
Given a presentation \S, we write $\free \S$  for the monoidal bicategory generated by this data.
\end{definition}

\noindent
In $\free \S$, for 1\-morphisms $X,Y$ with the same source and target, we write $X \equiv Y$ if they are identical as composites, and $X \simeq Y$ if they are isomorphic using the 2\-morphism structure. For 2\-morphisms $P,Q$ with the same source and target, we write $P \equiv Q$ if they are identical as composites, and $P = Q$ if they are equivalent modulo the equational structure.

We can now state our main theorem, which makes reference to the Frobenius presentation \F given in the introduction.
\begin{definition}
A 1\-morphism in $\free \F$ is \emph{connected} or \emph{acyclic} when its string diagram graph is connected or acyclic respectively.
\end{definition}
\begin{theorem}[Coherence for Frobenius structures]
\label{thm:frobeniuscoherence}
Let $P, Q: X \to Y$ be 2\-morphisms in \free \F, such that $X$ is connected and acyclic, with nonempty boundary. Then $P = Q$.
\end{theorem}

\noindent
Thus, as long as one restricts to connected, acyclic diagrams with boundary, all diagrams commute. The proof strategy is to `rotate' 2\-morphisms in \free \F so that they are equal to 2\-morphisms in the image of the obvious inclusion $\free \P \to \free \F$, and then rely on coherence for pseudomonoids. 

\ignore{
\begin{definition}A \textit{$*$-autonomous category} is a monoidal category $(\cat C, \parr, \bot)$ equipped with an object $I$ and an equivalence $*: \cat C \to \cat C ^\mathrm{op}$, such that there is a natural isomorphism as follows:
\begin{equation}
\Hom(A,B^*) \simeq \Hom(I, A \parr B)
\end{equation}
\end{definition}
}

\subsection{Rotations and twistedness}

We begin by defining a new presentation, given by \F with some additional data added.

\begin{definition}
The \textit{extended Frobenius signature} $\E$ is the same as \F,  with the following additional data:
\begin{itemize}
\item Morphism $\cap: \cat C \boxtimes \cat C \to \cat 1$:
\[
\tikzpng[scale=3, yscale=-1]{cup}
\]
\item Invertible 2\-morphism $\pi$:
\[
\tikzpng[scale=\myscale, yscale=-1]{cup}
{\xto \pi}
\tikzpng[scale=\myscale]{multform}
\]
\item 2\-Morphisms $\sn_1$, $\sn_2$, $\sn_1^\inv$ and $\sn_2^\inv$, called the \textit{snakeorators} and \textit{inverse snakeorators}, equipped with additional equations as follows:
\begin{align*}
\hspace{-0.4cm}\sn_1 &=
\tikzpng[scale=\myscale, yscale=-1]{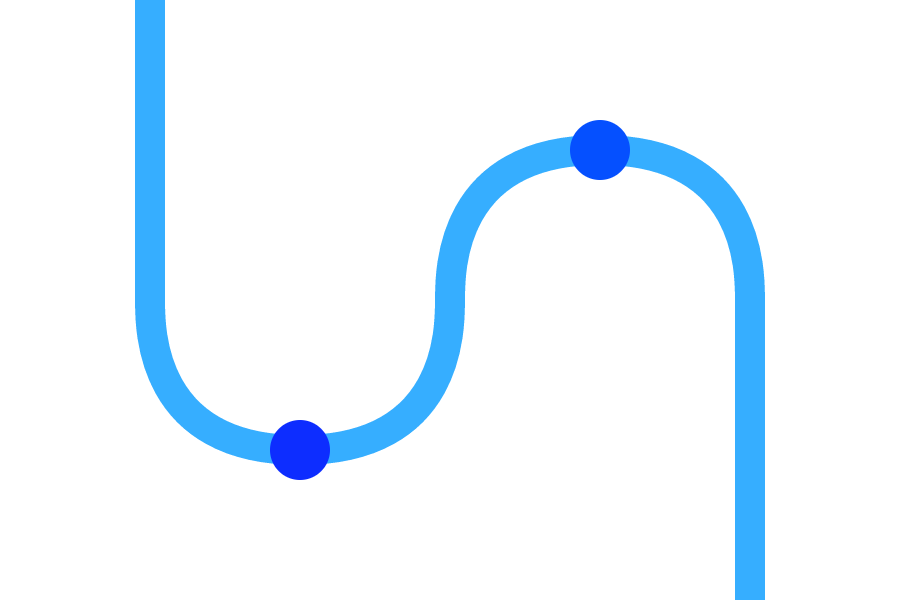}
{\xto \pi}
\tikzpng[scale=\myscale]{mu}
{\xto \mu}
\tikzpng[scale=\myscale,yscale=1.5]{bigidentity}
&
\sn_1^\inv &=
\tikzpng[scale=\myscale,yscale=1.5]{bigidentity}
{\xto {\mu^\inv}}
\tikzpng[scale=\myscale]{mu}
{\xto {\pi^\inv}}
\tikzpng[scale=\myscale, yscale=-1]{snake}
\\
\hspace{-1cm}\sn_2 &=
\tikzpng[scale=\myscale, xscale=1]{snake}
{\xto \pi}
\tikzpng[scale=\myscale, xscale=-1]{mu}
{\xto \nu}
\tikzpng[scale=\myscale,yscale=1.5]{bigidentity}
&
\sn_2^\inv &=
\tikzpng[scale=\myscale,yscale=1.5]{bigidentity}
{\xto {\nu^\inv}}
\tikzpng[scale=\myscale, xscale=-1]{mu}
{\xto {\pi^\inv}}
\tikzpng[scale=\myscale, xscale=1]{snake}
\end{align*}
\item 2\-Morphisms $R_m$, $L_m$, $R_u$, $L_u$, $R_f$, $L_f$, with additional equations as follows:
\begin{align*}
\hspace{-0.4cm}R_m &= \tikzpng[scale=\miniscale]{bigmult}
{\xto {\sn_1 ^\inv}} \tikzpng[xscale=-1, scale=\miniscale]{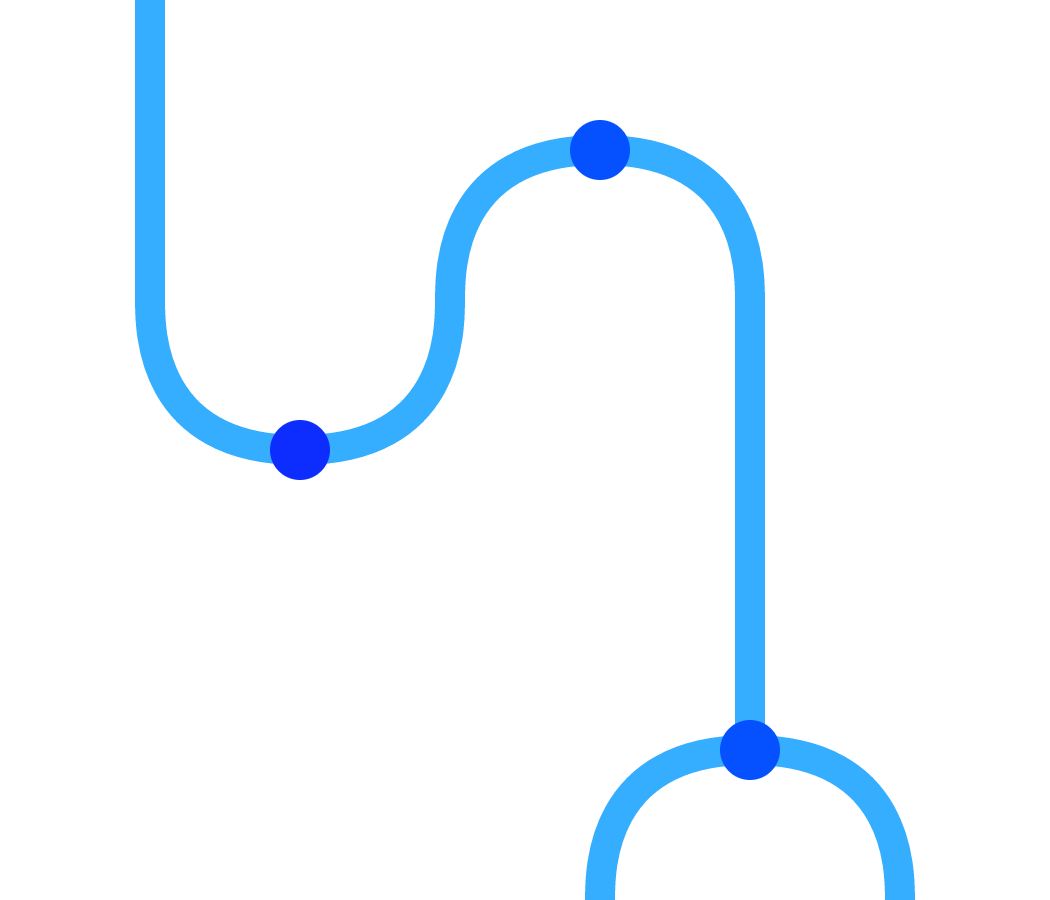}
{\xto {\sim}} \tikzpng[xscale=-1, scale=\miniscale]{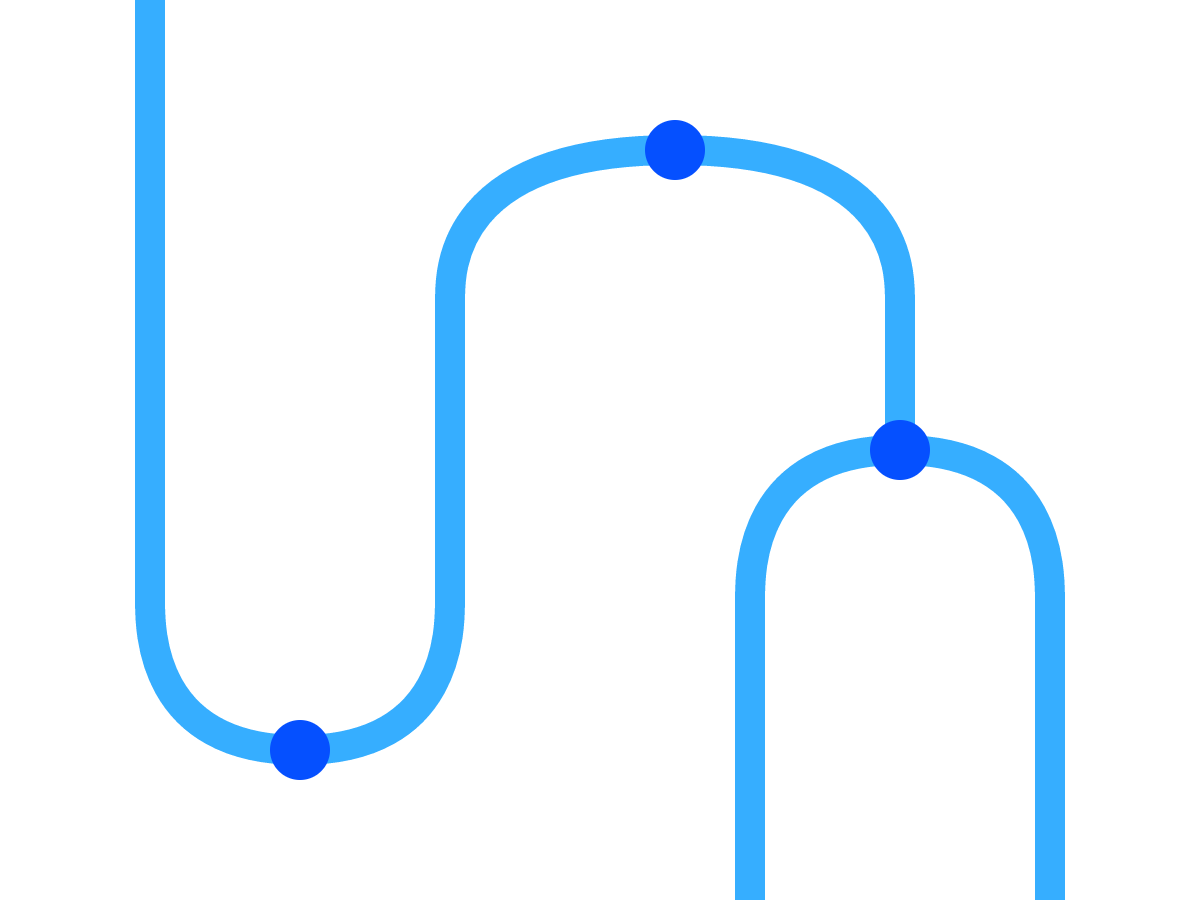}
{\xto \pi} \tikzpng[xscale=-1, scale=\miniscale]{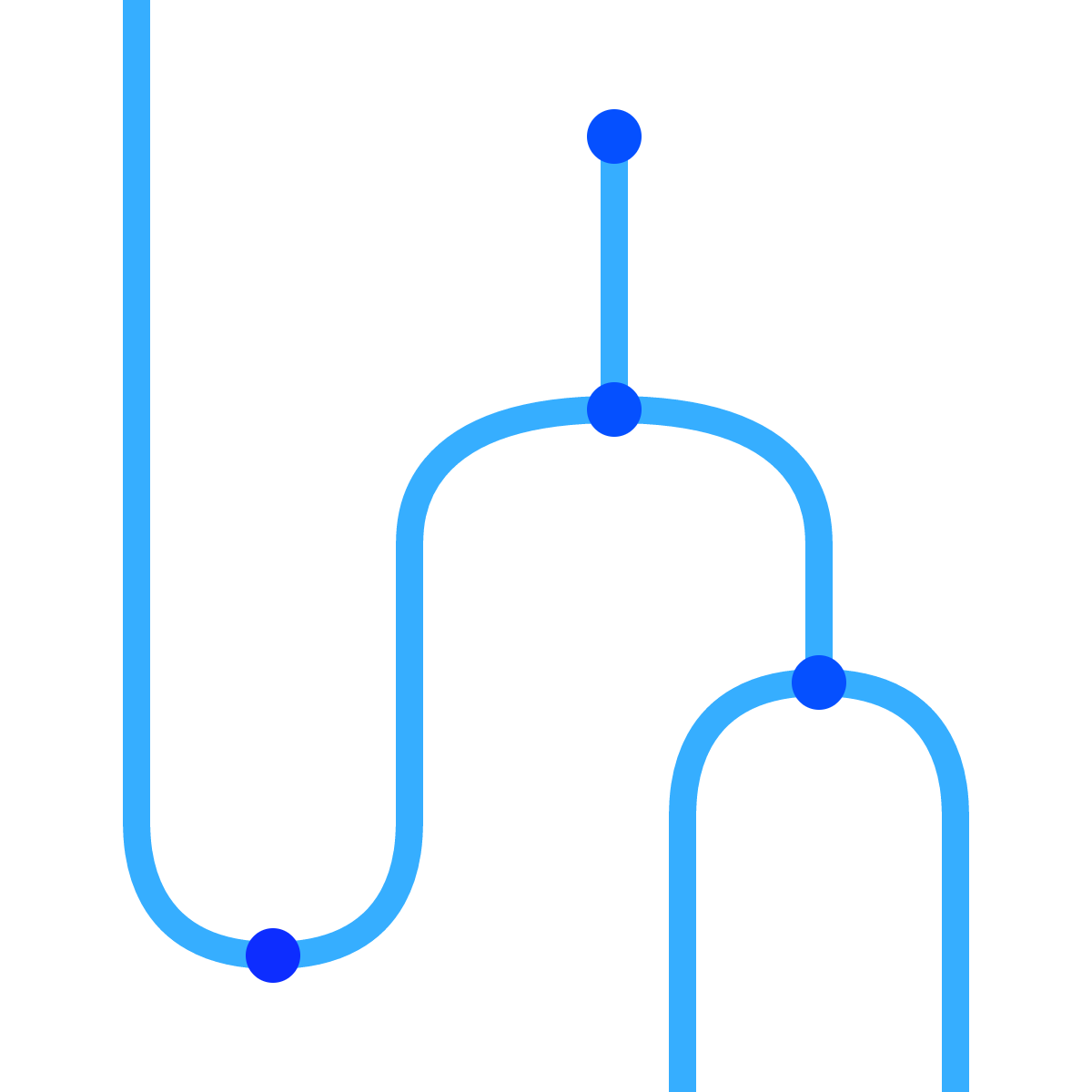}
{\xto {\alpha}} \tikzpng[xscale=-1, scale=\miniscale]{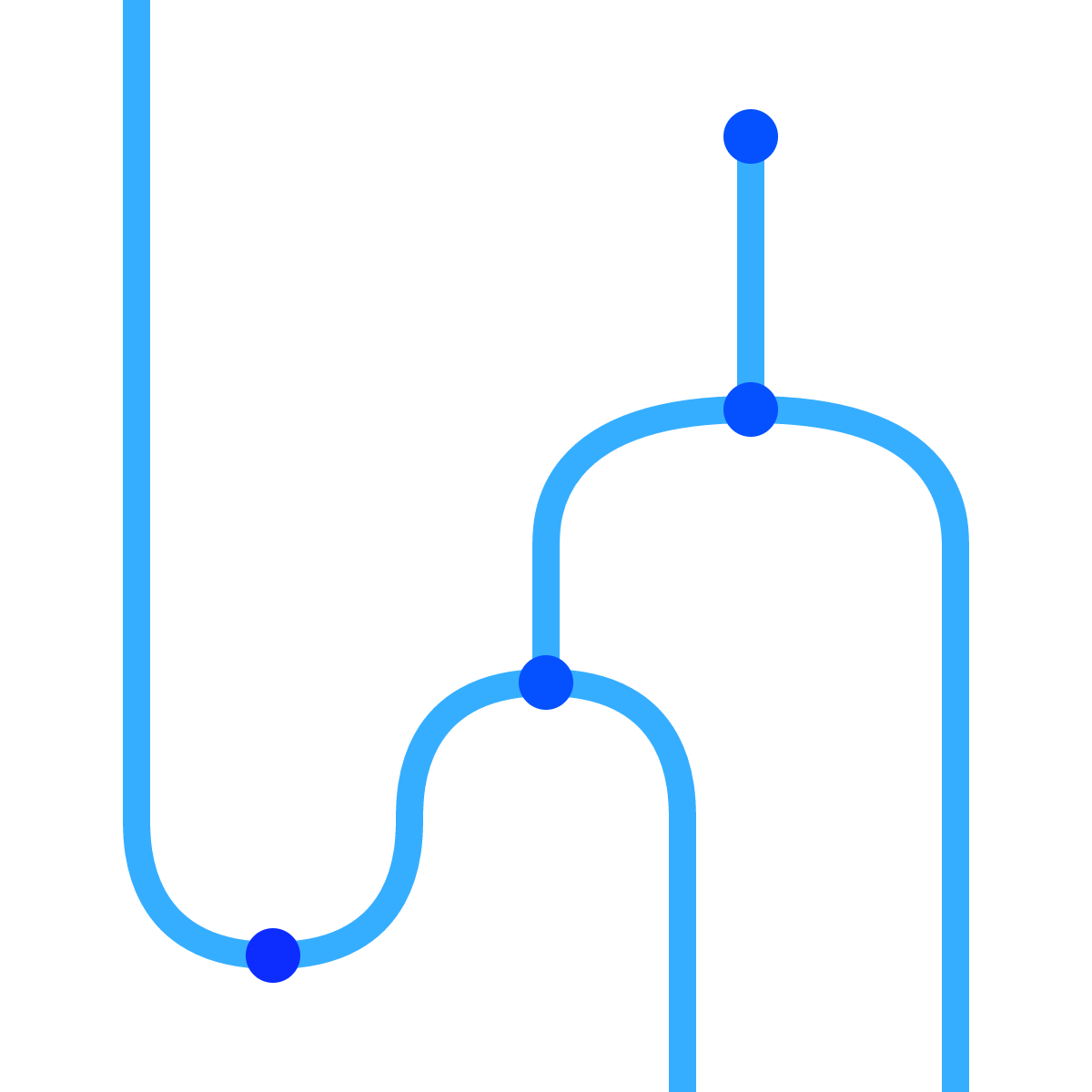}
{\xto {\pi ^\inv}} \tikzpng[xscale=-1, scale=\miniscale]{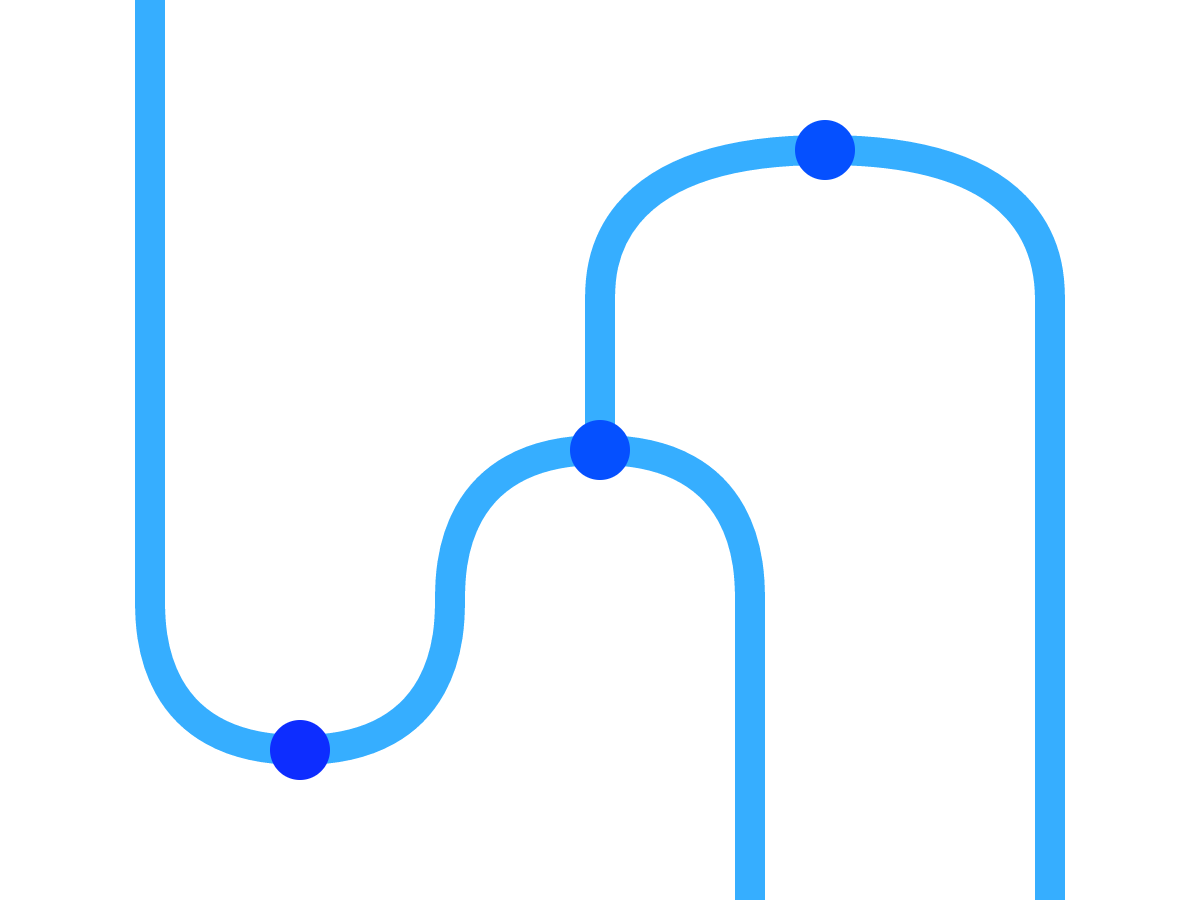}
\\
\hspace{-1cm}L_m &=
\tikzpng[scale=\miniscale]{bigmult}
{\xto {\sn_2^\inv}} \tikzpng[scale=\miniscale]{multL2}
{\xto \sim} \tikzpng[scale=\miniscale]{multL3}
{\xto \pi} \tikzpng[scale=\miniscale]{multL4}
{\xto {\alpha^\inv}} \tikzpng[scale=\miniscale]{multL5}
{\xto {\pi ^\inv}} \tikzpng[scale=\miniscale]{multL6}
\\
\hspace{-1cm}R_u &= \tikzpng[scale=\miniscale]{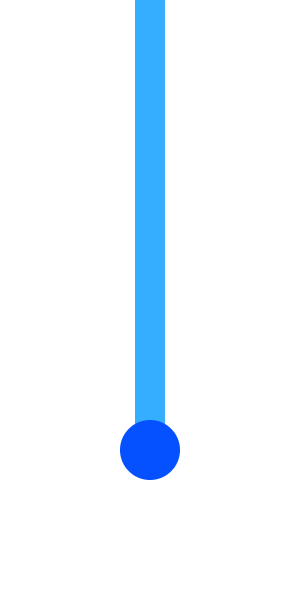}
{\xto {\sn_1^\inv}} \tikzpng[scale=\miniscale]{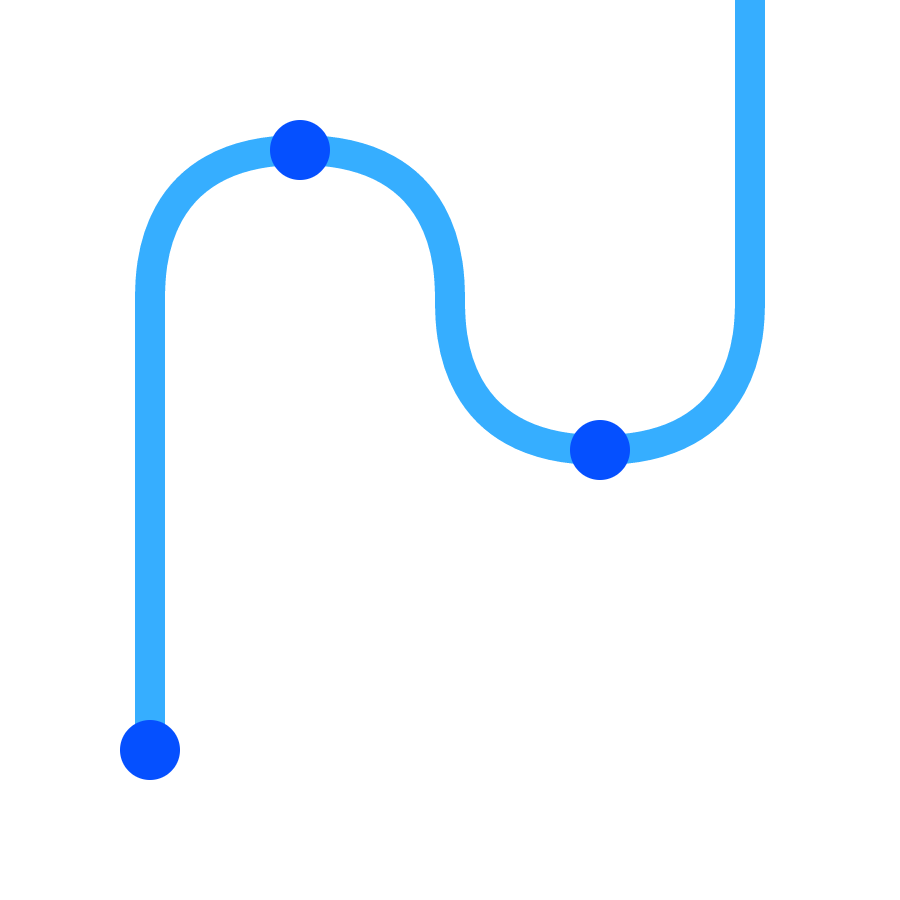}
{\xto \sim} \tikzpng[scale=\miniscale]{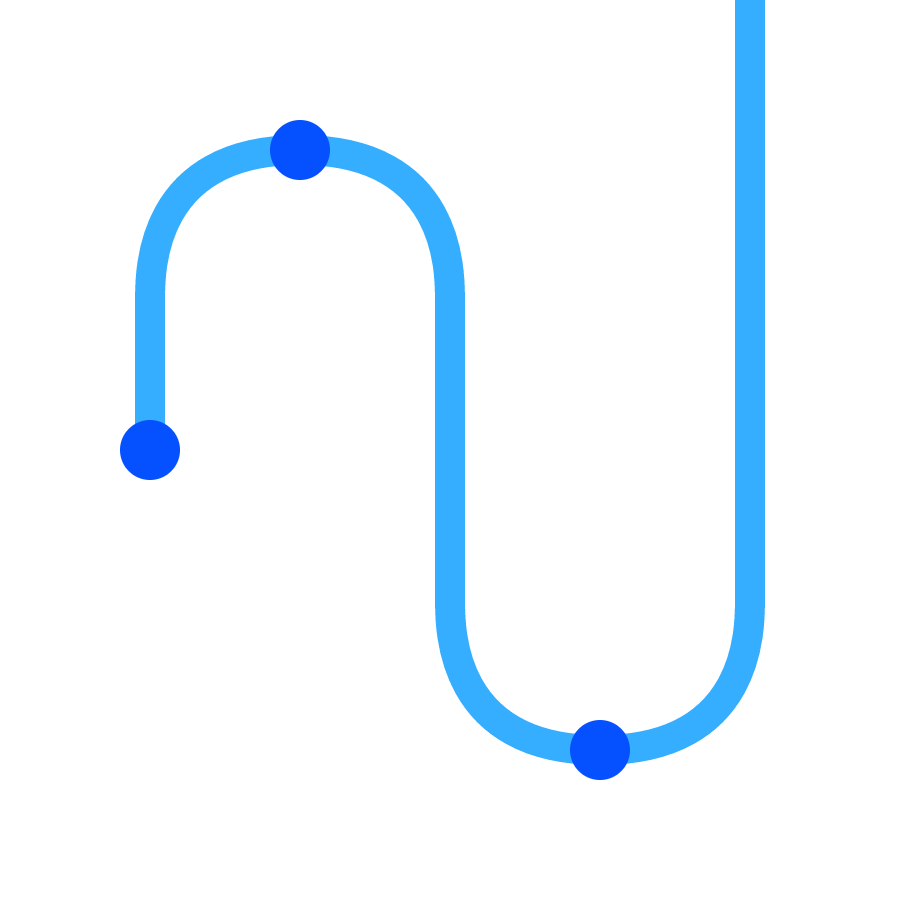}
{\xto \pi} \tikzpng[scale=\miniscale]{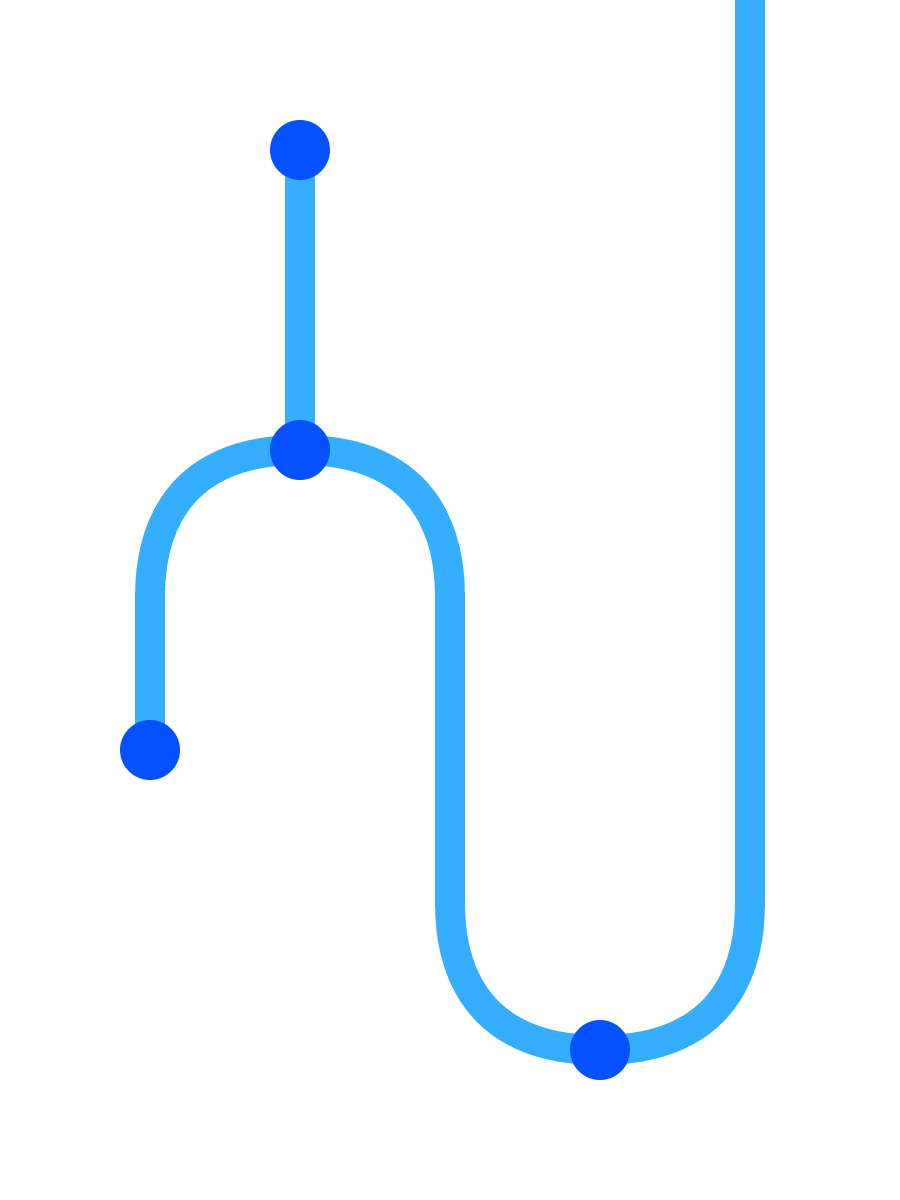}
{\xto \lambda} \tikzpng[scale=\miniscale]{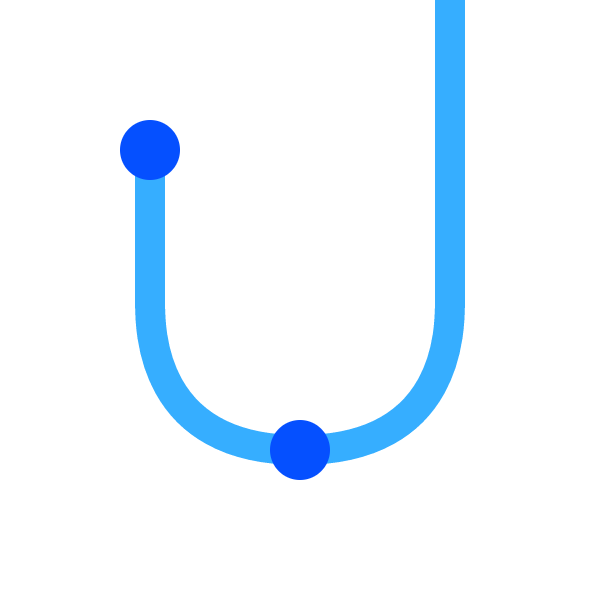}
\\
\hspace{-1cm}L_u &= \tikzpng[scale=\miniscale]{bigunit}
{\xto {\sn_2^\inv}} \tikzpng[xscale=-1, scale=\miniscale]{Ru2}
{\xto \sim} \tikzpng[xscale=-1, scale=\miniscale]{Ru3}
{\xto \pi} \tikzpng[xscale=-1, scale=\miniscale]{Ru4}
{\xto \rho} \tikzpng[xscale=-1, scale=\miniscale]{Ru5}
\\
\hspace{-1cm}R_f &= \tikzpng[yscale=-1, scale=\miniscale]{bigunit}
{\xto {\rho^\inv}} \tikzpng[xscale=-1, scale=\miniscale]{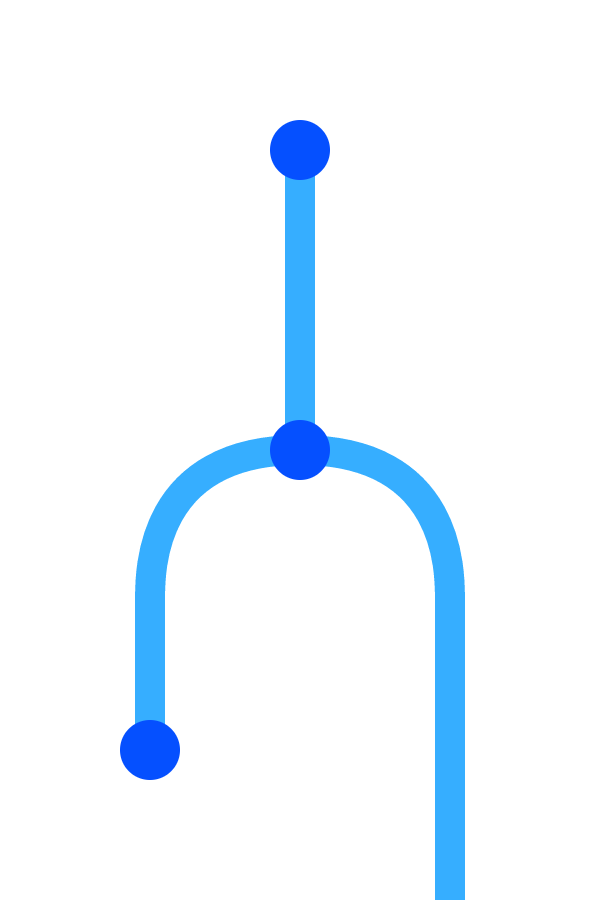}
{\xto {\pi^\inv}} \tikzpng[xscale=-1, scale=\miniscale]{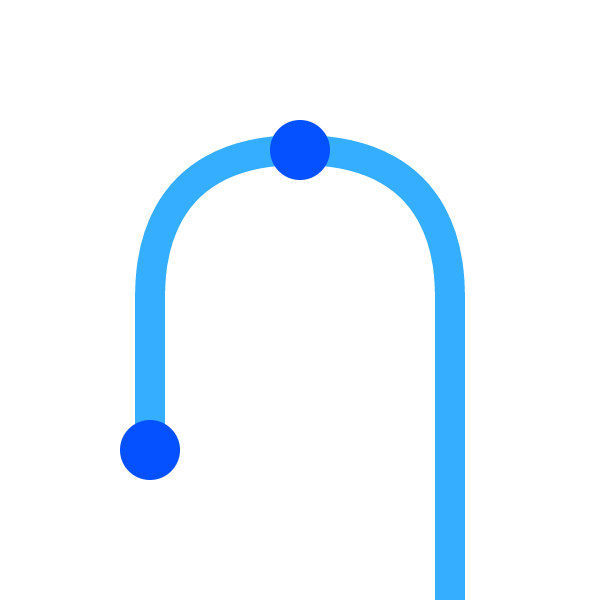}
\\
\hspace{-1cm}L_f &= \tikzpng[yscale=-1, scale=\miniscale]{bigunit}
{\xto {\lambda^\inv}} \tikzpng[scale=\miniscale]{Rf1}
{\xto {\pi ^\inv}} \tikzpng[scale=\miniscale]{Rf2}
\end{align*}
\end{itemize}
\end{definition}

The presentations \F and \E are \textit{equivalent}~\cite{CSPthesis}, in the sense that $\free\F \simeq \free\E$ as monoidal 2\-categories, a strong property which implies in particular that they have the same representation theory. In fact, \E is a \textit{simple homotopy extension} of \F.
\begin{lemma}
\label{lem:equivalentpresentations}
The presentations \F and \E are equivalent.
\end{lemma}
\begin{proof}
A technical proof is tedious to give. A clear intuitive argument can be given, which is easy to formalize: the presentation \E is constructed by taking \F and adding new generators, equipped with higher cells that show how these new generators can be expressed in terms of existing cells. 
\end{proof}

\noindent
The reason for constructing \E is that it has more convenient formal properties, which will allow us to construct our proof.

The 2\-morphisms $R_m$, $L_m$, $R_u$, $L_u$, $R_f$ and $L_f$ can be interpreted as \textit{rotating} their domain, either to the right or the left. For this reason we make the following definition.
\begin{definition}
In \free \E, a 2\-morphism is \emph{rotational} if it is composed only from $R_m$, $L_m$, $R_u$, $L_u$, $R_f$, $L_f$, $\sn_1$, $\sn_1^\inv$, $\sn_2$, $\sn_2^\inv$ and interchangers; we write $\stackrel R \twoheadrightarrow$ to denote a composite of rotational generators of length at least 0.
\end{definition}

\begin{lemma}
In \free \E, the snake maps satisfy the swallowtail equations:
\begin{align*}
\id \quad &= \quad\!\!\! \tikzpng[scale=\myscale]{cup} {\xto {\sn_1^\inv}} \tikzpng[scale=\myscale]{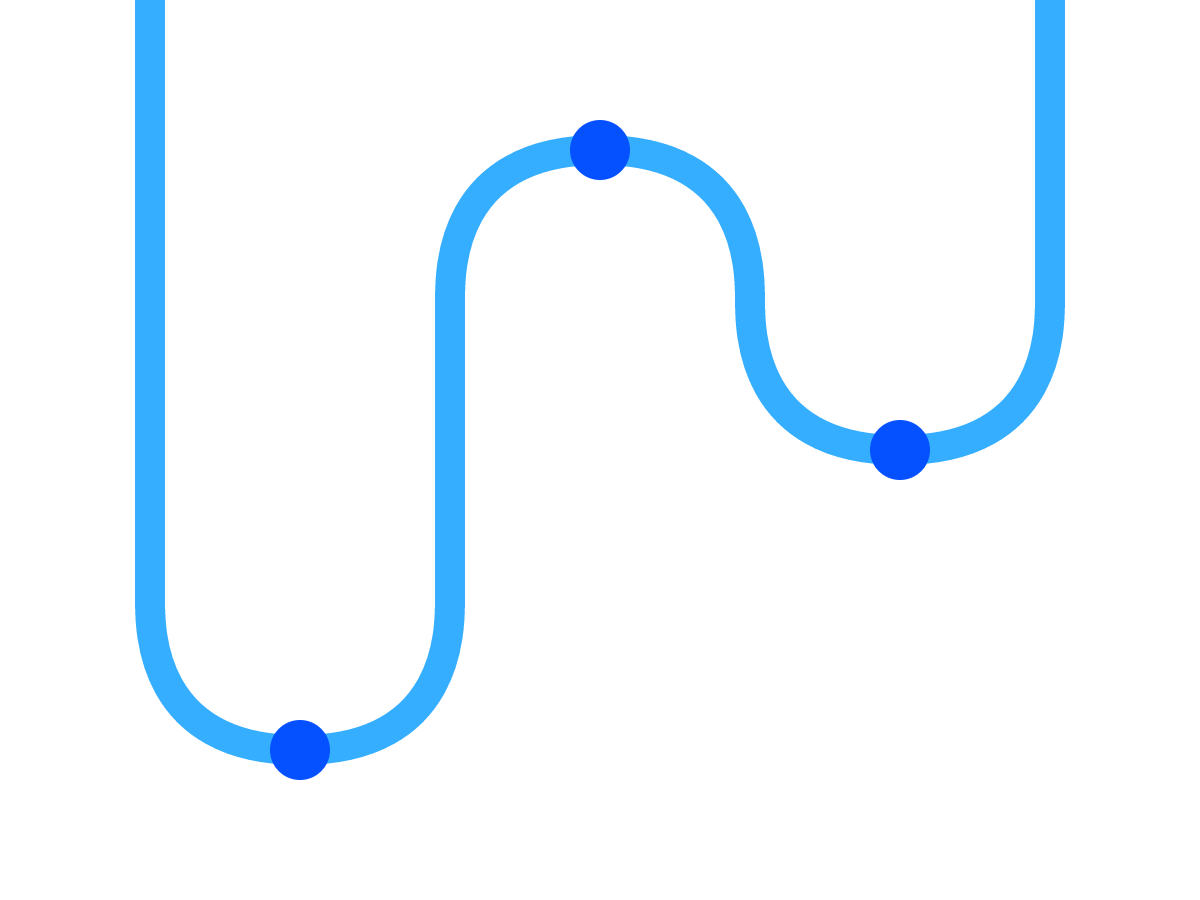} {\xto \sim} \tikzpng[xscale=-1, scale=\myscale]{swallowtail3} {\xto {\sn_2}} \tikzpng[scale=\myscale]{cup}
\\
\id \quad &= \quad\!\!\! \tikzpng[scale=\myscale, yscale=-1]{cup} {\xto {\sn_1^\inv}} \tikzpng[scale=\myscale, yscale=-1, xscale=-1]{swallowtail3} {\xto \sim} \tikzpng[xscale=1, scale=\myscale, yscale=-1]{swallowtail3} {\xto {\sn_2}} \tikzpng[scale=\myscale, yscale=-1]{cup}
\end{align*}
\end{lemma}
\begin{proof}
Follows easily from the swallowtail equations that form part of the Frobenius presentation \F. See \glob, 5-cells \textit{``Pf: Swallowtail 1''} and \textit{``Pf: Swallowtail~2''}.
\end{proof}

\begin{lemma}
\label{lem:inversepairs}
In \free \E, the elements of each pair $(R_m, L_m)$, $(R_u,L_f)$, $(L_u, R_f)$ are mutually inverse, up to interchangers and snake maps; that is, the following equations hold:
\def\sc{0.15}
\begin{align*}
\hspace{-0.1cm}\id \quad&=\quad
{\scalepng[\sc]{bigmult}}
{\xto{L_m}}
\scalepng[\sc]{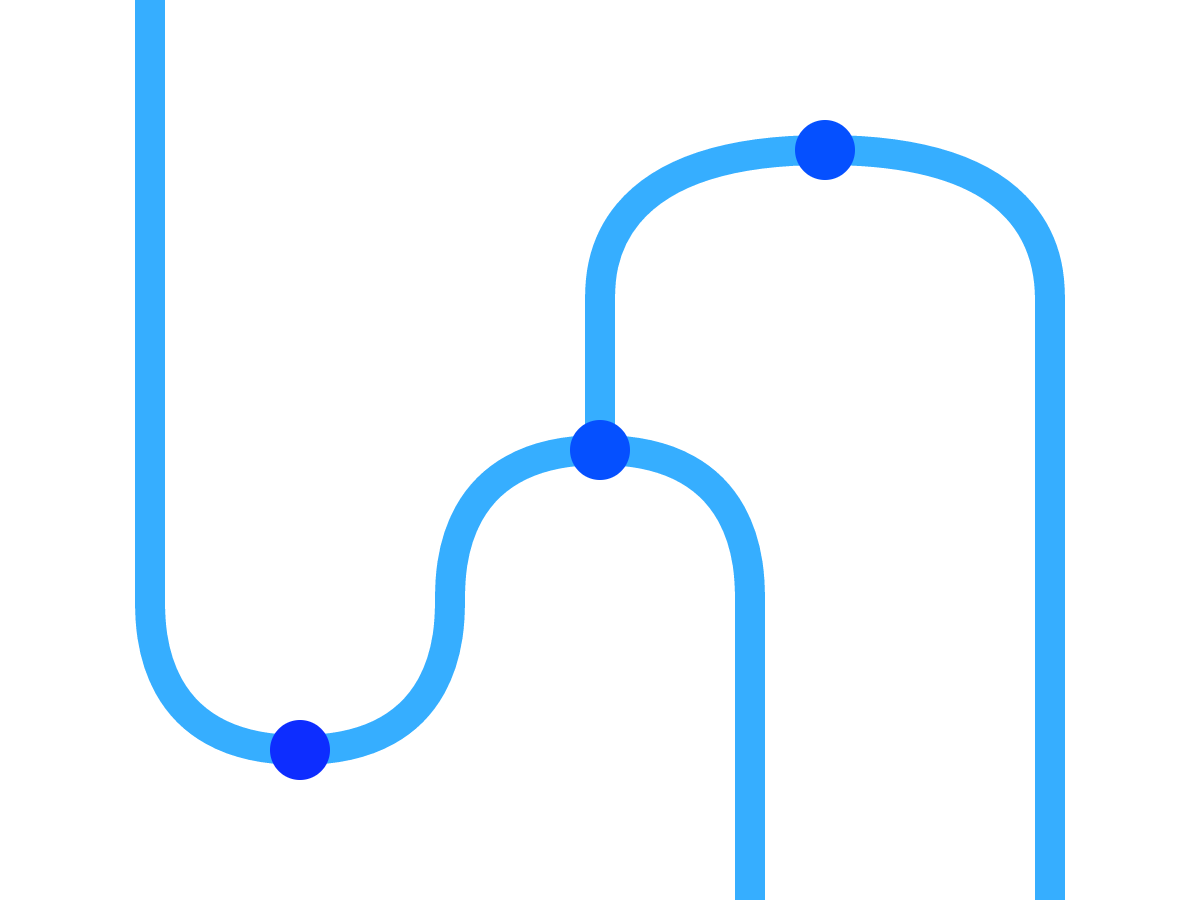}
{\xto{R_m}}
\scalepng[\sc]{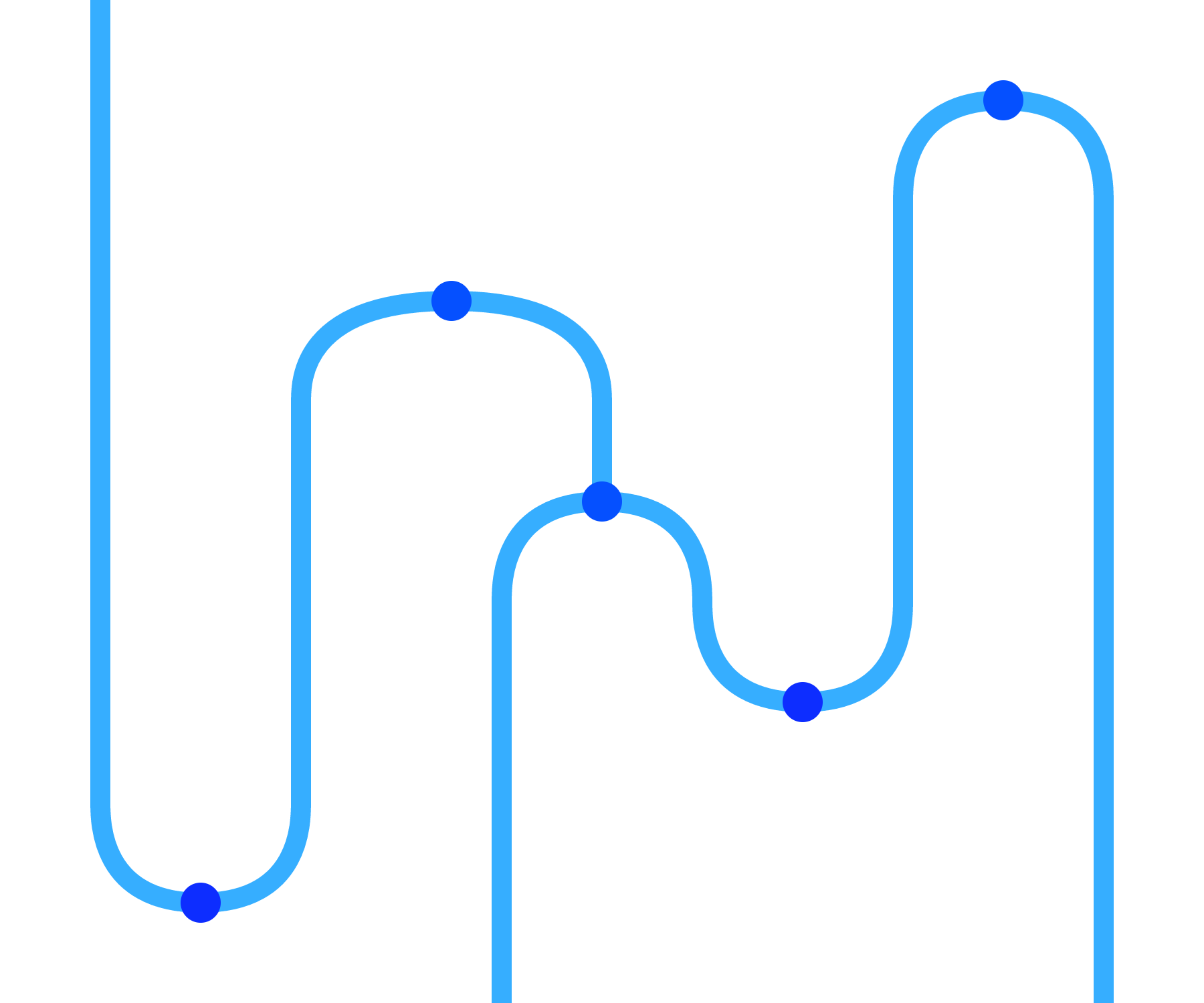}
{\xto{\sim}}
\scalepng[\sc]{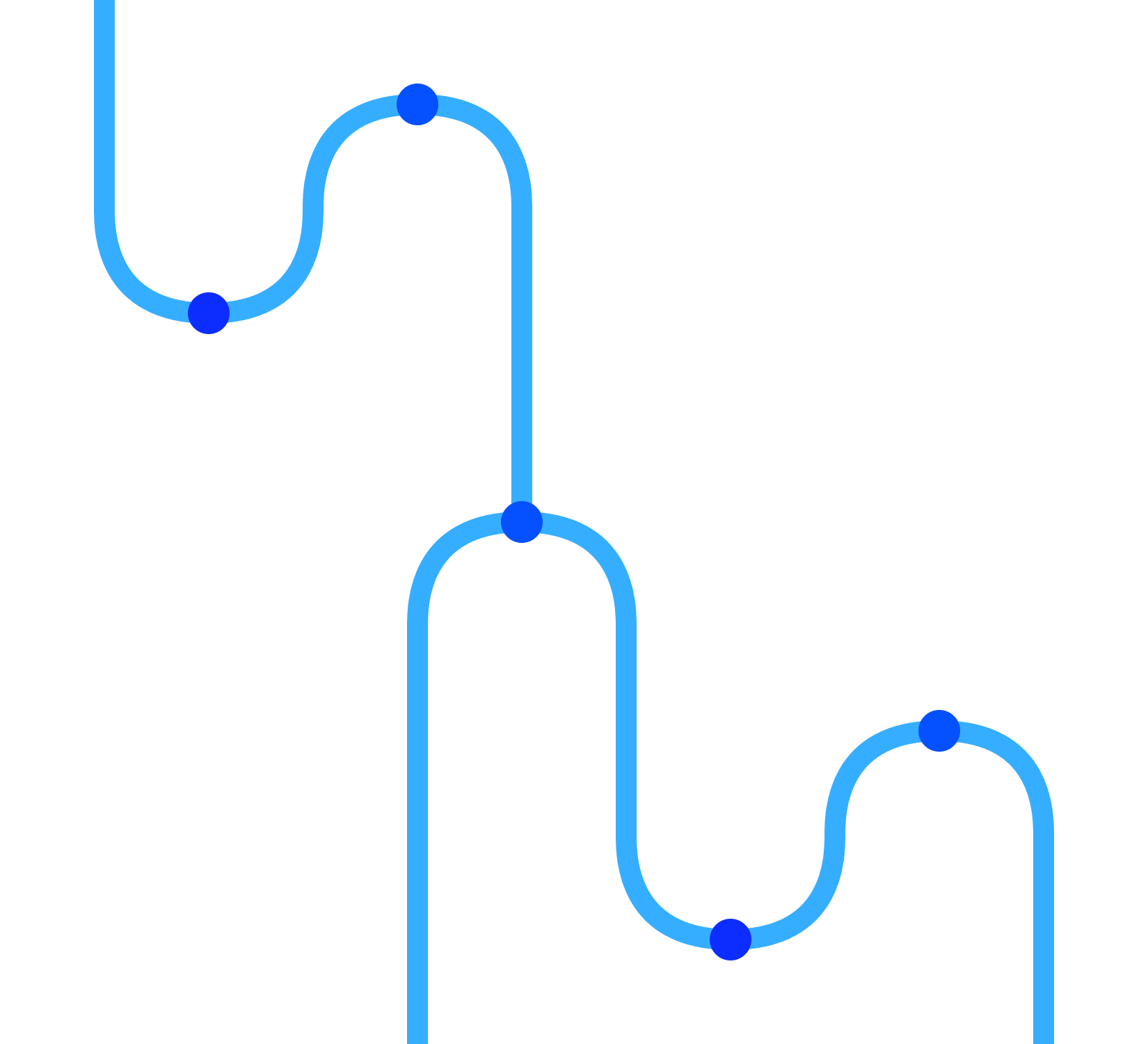}
{\xto{\sn_2}}
{\scalepng[\sc]{bigmult}}
\hspace{-0.2cm}
\\
\hspace{-1cm}\id\quad&=\quad
{\scalepng[\sc]{bigmult}}
{\xto{R_m}}
\scalepng[\sc]{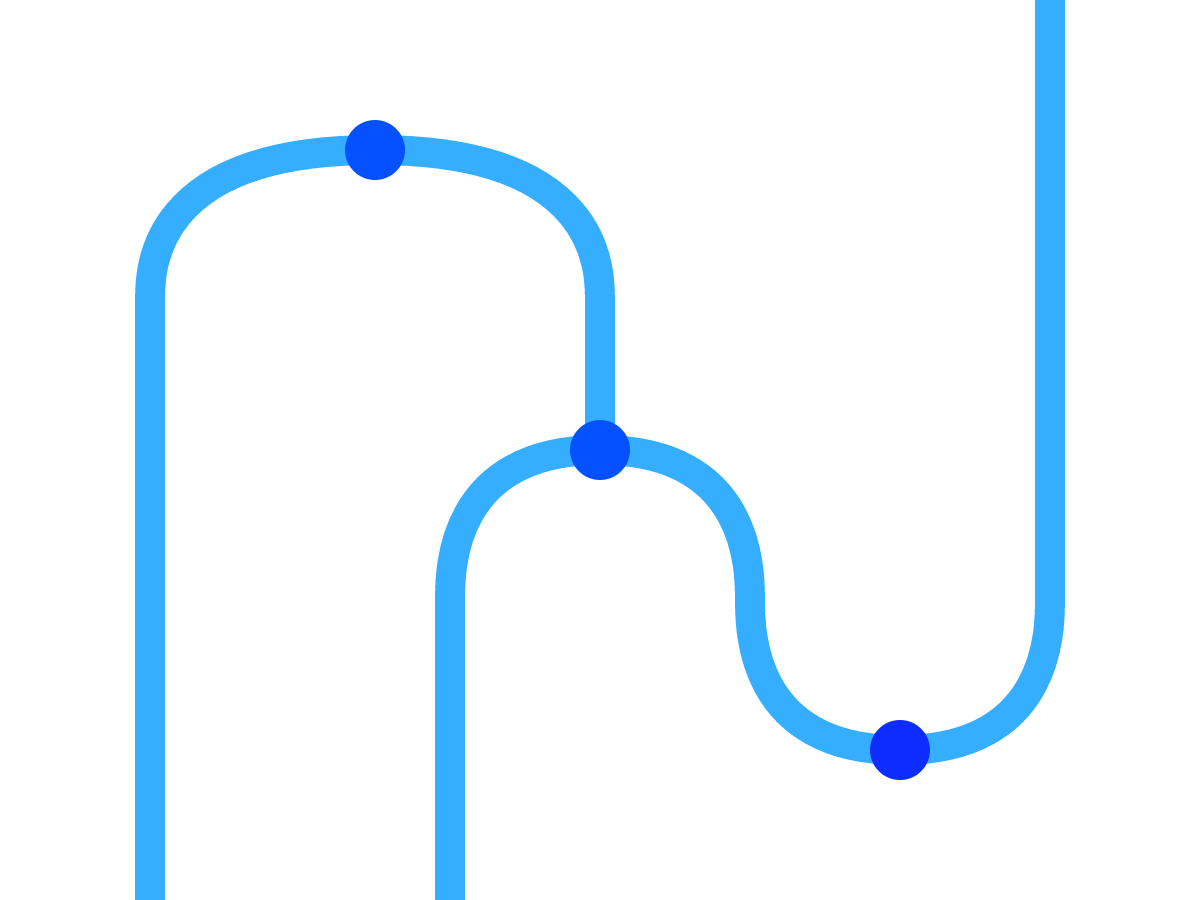}
{\xto{L_m}}
\scalepng[\sc]{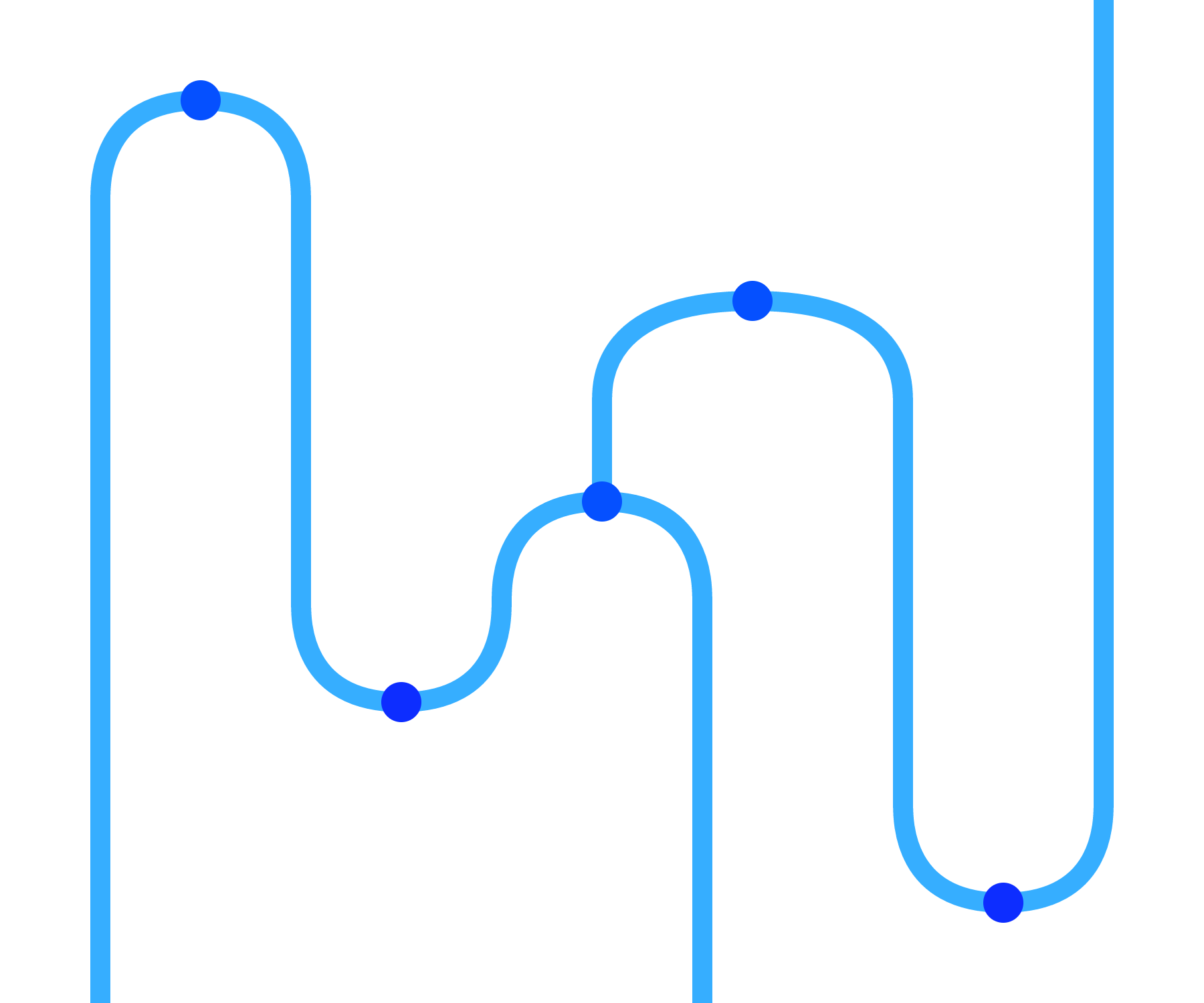}
{\xto{\sim}}
\scalepng[\sc]{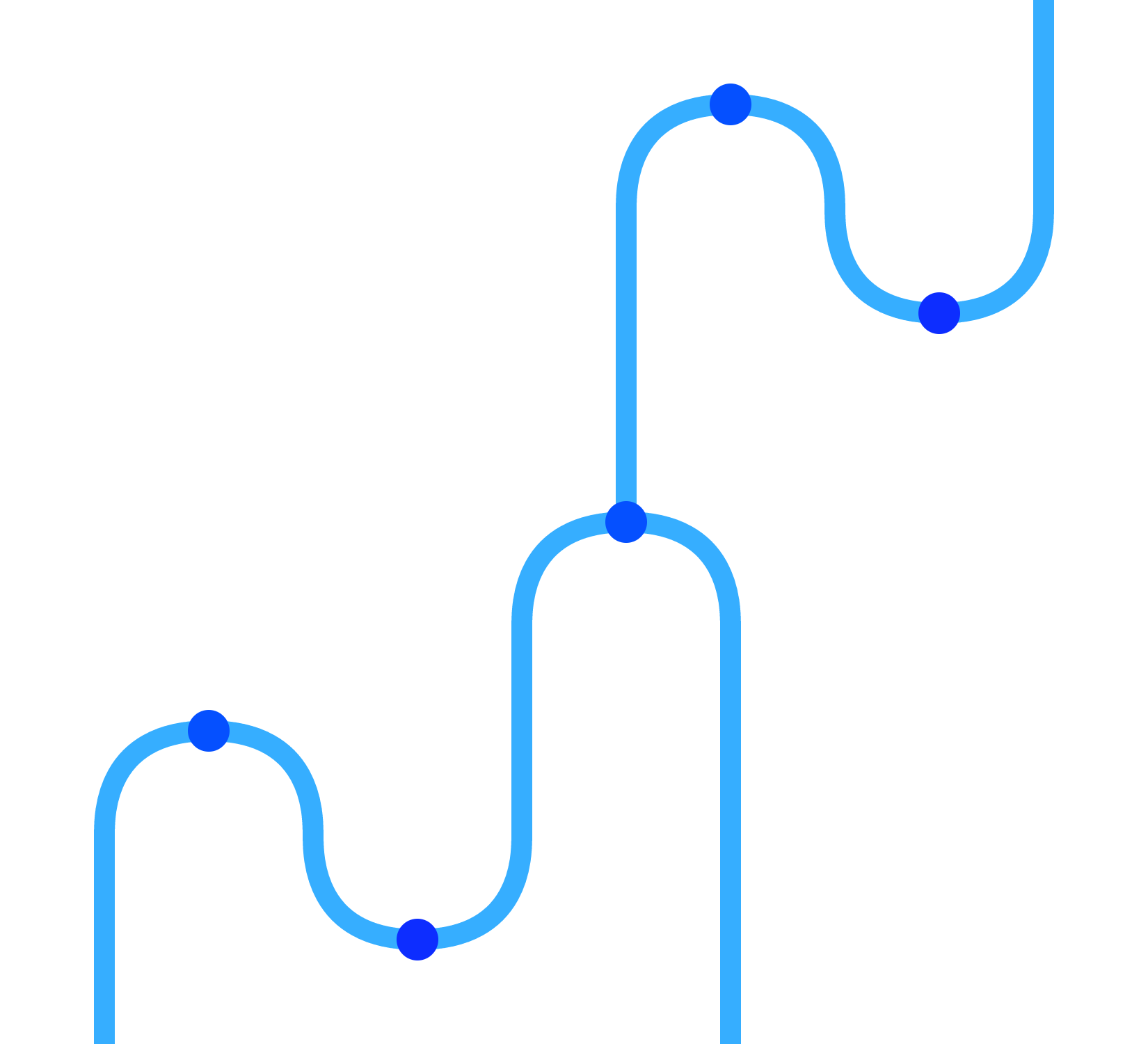}
{\xto{\sn_1}}
{\scalepng[\sc]{bigmult}}
\hspace{-0.2cm}
\\
\hspace{-1cm}\id\quad&=\quad
{\scalepng[\sc]{bigunit}}
{\xto{R_u}}
\scalepng[\sc]{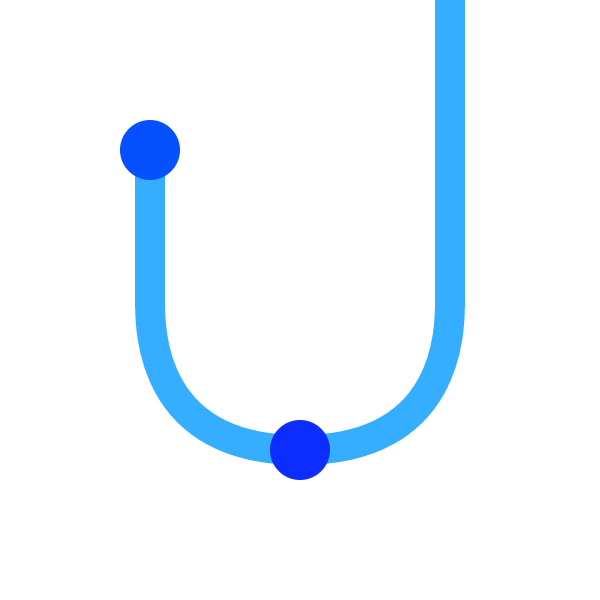}
{\xto{L_f}}
\scalepng[\sc]{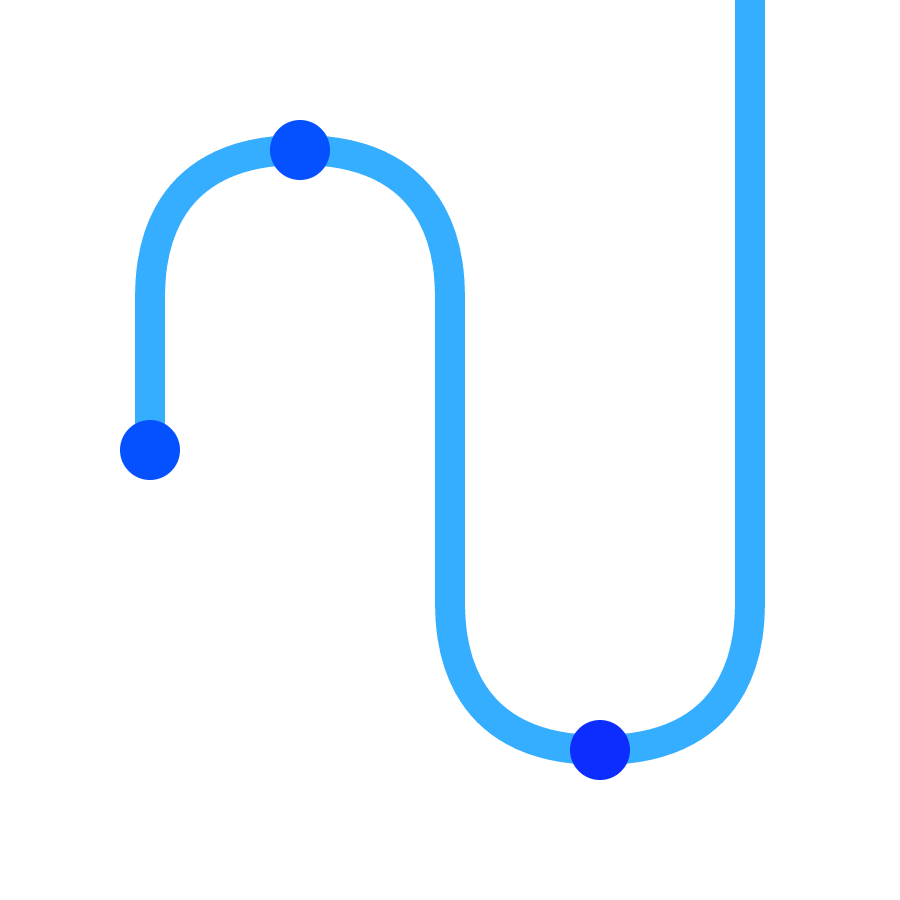}
{\xto{\sim}}
\scalepng[\sc]{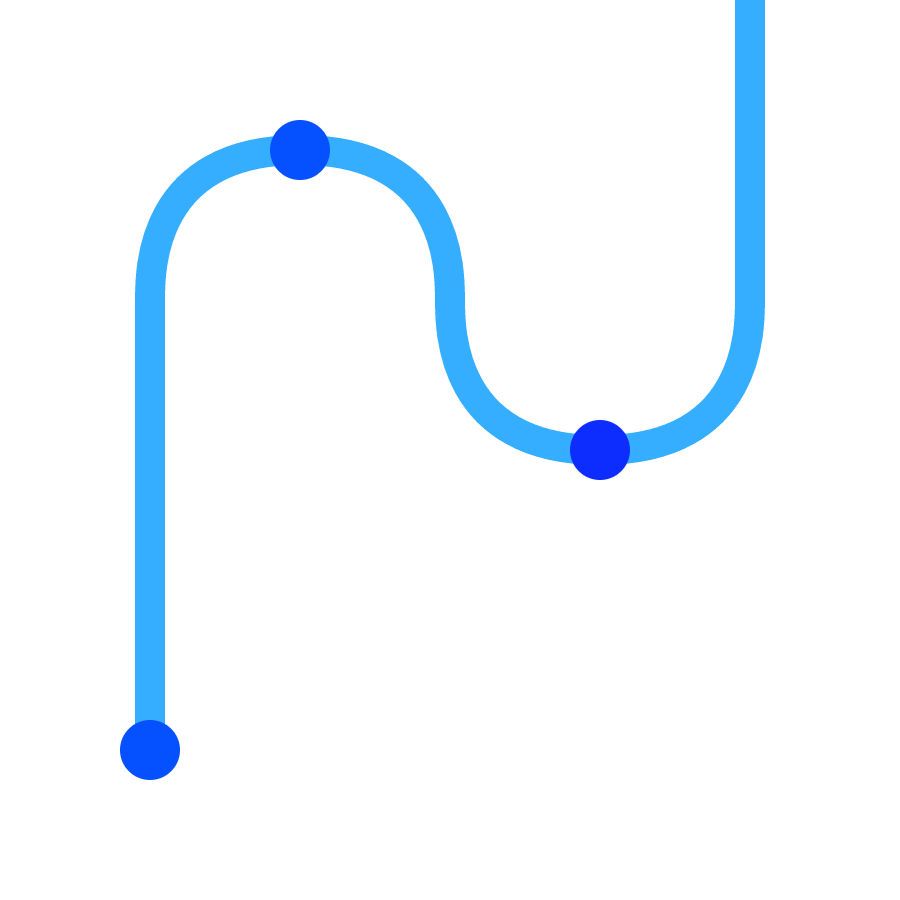}
{\xto{\sn_1}}
{\scalepng[\sc]{bigunit}}
\\
\hspace{-1cm}\id\quad&=\quad
{\scalepng[\sc]{bigunit}}
{\xto{L_u}}
\tikzpng[xscale=-1, scale=1.5]{RuLf2}
{\xto{R_f}}
\tikzpng[xscale=-1, scale=1.5]{RuLf3}
{\xto{\sim}}
\tikzpng[xscale=-1, scale=1.5]{RuLf4}
{\xto{\sn_2}}
{\scalepng[\sc, scale=1]{bigunit}}
\end{align*}
\end{lemma}
\begin{proof}
See \glob, 5\-cells \emph{``Pf: LmRm=1''}, \emph{``Pf: RmLm=1''}, \emph{``Pf: RuLf=1''} and \emph{``Pf: LuRf=1''}.
\end{proof}

\begin{lemma}
\label{lem:killmunu}
In $\free \E$, the 2\-morphisms $\mu$ and $\nu$ can be decomposed as follows:
\begin{align*}
\mu \quad&{=}\quad
\tikzpng[scale=1.5]{mu}
{\xto {L_m,L_f}}
\tikzpng[scale=1.5]{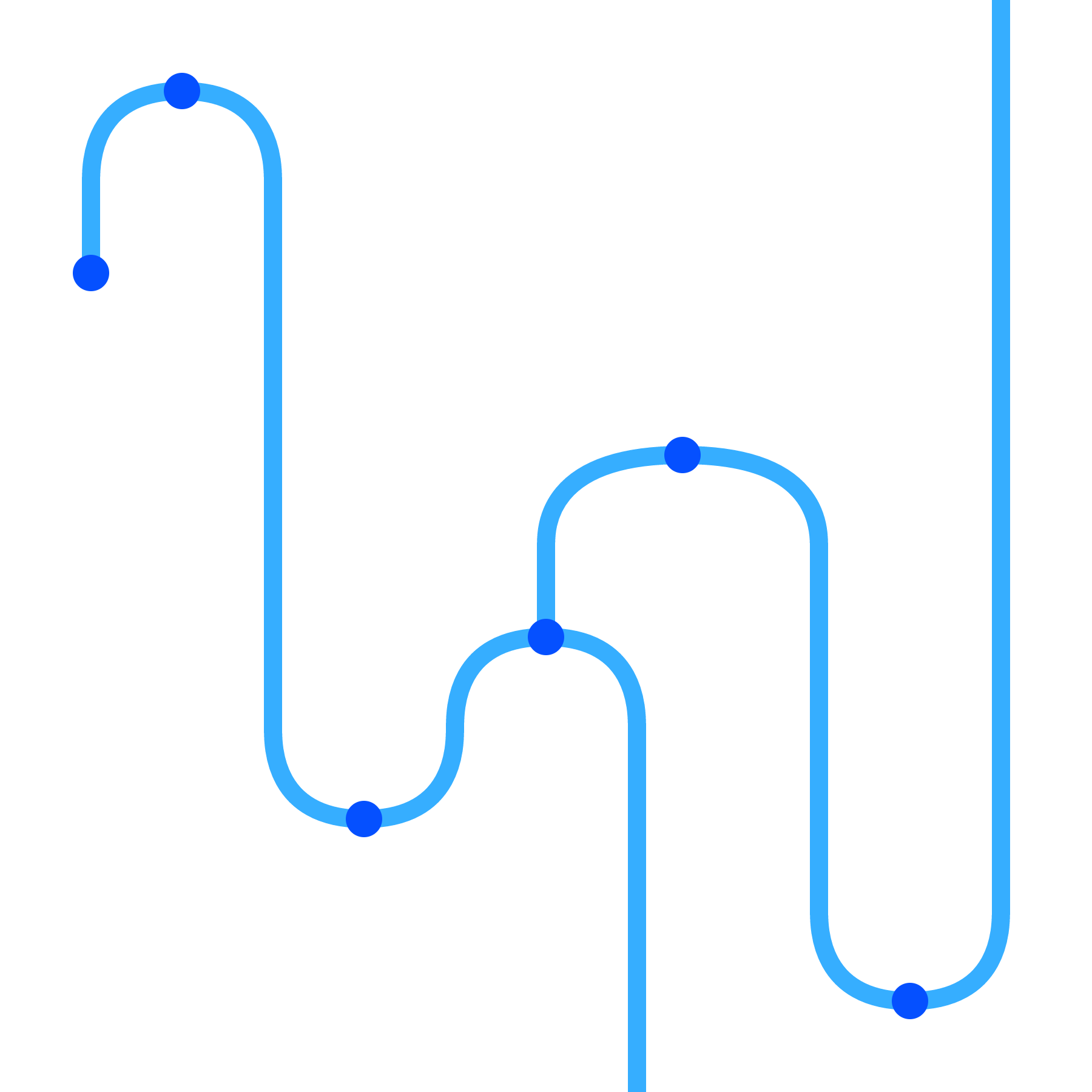}
{\xto {\sim}}
\tikzpng[scale=1.5]{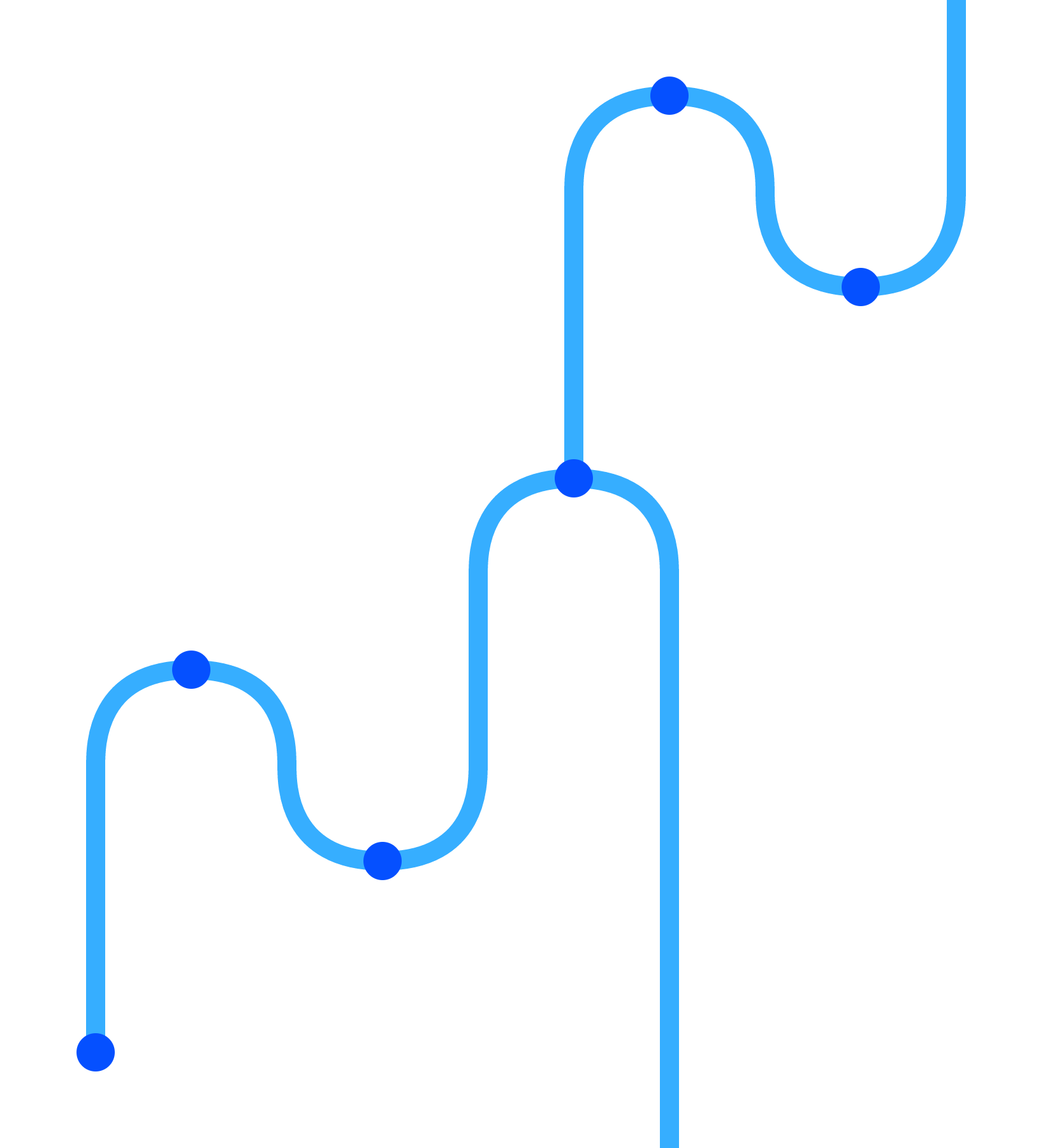}
\!\!\!\!{\xto {\sn_1}}
\tikzpng[scale=1.5]{lambda_source}
{\xto {\lambda}}
\tikzpng[scale=1.5]{bigidentity}
\\
\nu\quad&{=}\quad
\tikzpng[scale=1.5]{nu}
{\xto {R_m,R_f}}
\tikzpng[scale=1.5,xscale=-1]{mudecomp1}
{\xto {\sim}}
\tikzpng[scale=1.5,xscale=-1]{mudecomp2}
\!\!\!\!{\xto {\sn_2}}
\tikzpng[scale=1.5,xscale=-1]{lambda_source}
{\xto {\rho}}
\tikzpng[scale=1.5,xscale=-1]{bigidentity}
\end{align*}
\end{lemma}
\begin{proof}
See \glob, 5-cells \emph{``Pf: Mu Decomposition''} and \emph{``Pf: Nu Decomposition''}.
\end{proof}

For our result, an important process is the elimination of adjacent cup-cap pairs in a string diagram.
\begin{definition}
\label{def:eliminable}
For a 1\-morphism $X$ in \free \E, a chosen $(\cup, \cap)$ pair are \textit{eliminable} if they match one of the following patterns:
\tikzset{sc/.style={scale=0.512}}
\begin{calign}
\nonumber
\begin{tz}[sc]
\node [anchor=south west, inner sep=0pt] at (0,0) {\tikzpng[scale=2]{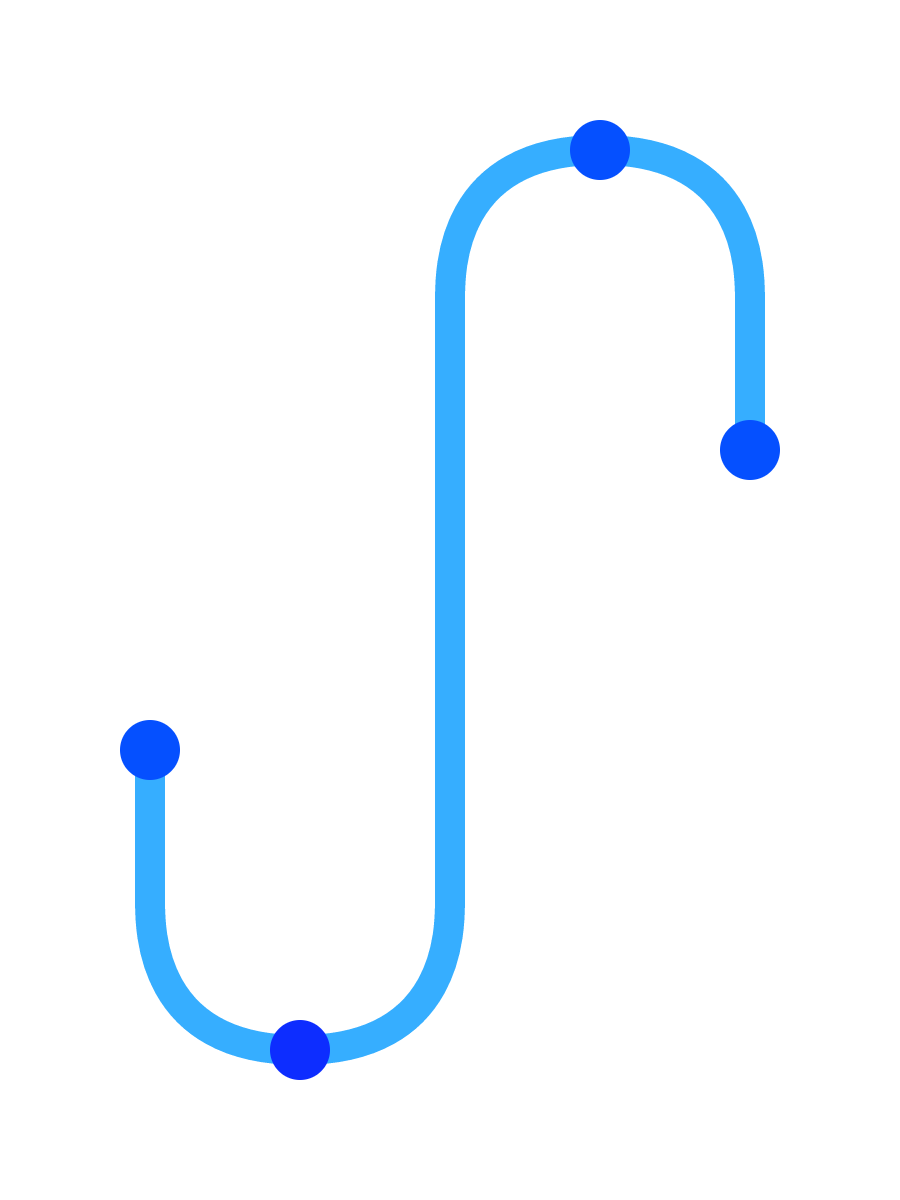}};
\draw [fill=white, draw=none] (0,1) rectangle +(1,1);
\draw [lightblue, line width=1.5pt, dotted] (0.5,0.95) to +(0,0.75);
\draw [fill=white, draw=none] (2,2) rectangle +(1,1);
\draw [lightblue, line width=1.5pt, dotted] (2.5,3) to +(0,-0.75);
\end{tz}
&
\begin{tz}[sc]
\node [anchor=south west, inner sep=0pt] at (0,0) {\tikzpng[scale=2, xscale=-1]{eliminable}};
\draw [fill=white, draw=none] (2,1) rectangle +(1,1);
\draw [lightblue, line width=1.5pt, dotted] (2.505,0.95) to +(0,0.75);
\draw [fill=white, draw=none] (0,2) rectangle +(1,1);
\draw [lightblue, line width=1.5pt, dotted] (0.51,3) to +(0,-0.75);
\end{tz}
\end{calign}
That is, they are directly connected by a central straight wire, and turn in opposite directions. The surrounding diagram may be nontrivial, and the $\cup$ and $\cap$ may not have adjacent heights.
\end{definition}

\begin{definition}
\label{def:srs}
For a 1\-morphism $X$ in \free \E, with a chosen eliminable cup and cap pair, the \emph{snake removal scheme} $X \to X'$ is the following sequence of 2\-morphisms:
\begin{enumerate}
\item Identify the obstructing generators, which are the components of the diagram lying \textit{between} the $\cup$ and $\cap$ in height.
\item Gather obstructing generators into maximal groups, with respect to whether they are left or the right of the central wire.
\item Interchange groups vertically beyond the $\cup$ or $\cap$.
\item Apply the $\sn_1$ or $\sn_2$ map.
\end{enumerate}
The following graphics explain the scheme. We begin with a central wire with an eliminable cup and cup. First, we identify the obstructing generators, and we gather them into groups as large as possible, which we draw in red to the left of the central wire and in green to the right, giving the first graphic below. Then we interchange red groups up and green groups down, giving in the third graphic. Finally we cancel the cup and cap, giving the final graphic.
\tikzset{sc/.style={scale=0.377, scale=1.333}}
\tikzset{leftbox/.style={fill=red, draw=lightblue, line width=1pt}}
\tikzset{rightbox/.style={fill=green, draw=lightblue, line width=1pt}}
\[
\begin{tz}[sc]
\node (1) [anchor=south west, inner sep=0pt] at (0,0) {\tikzpng[scale=2]{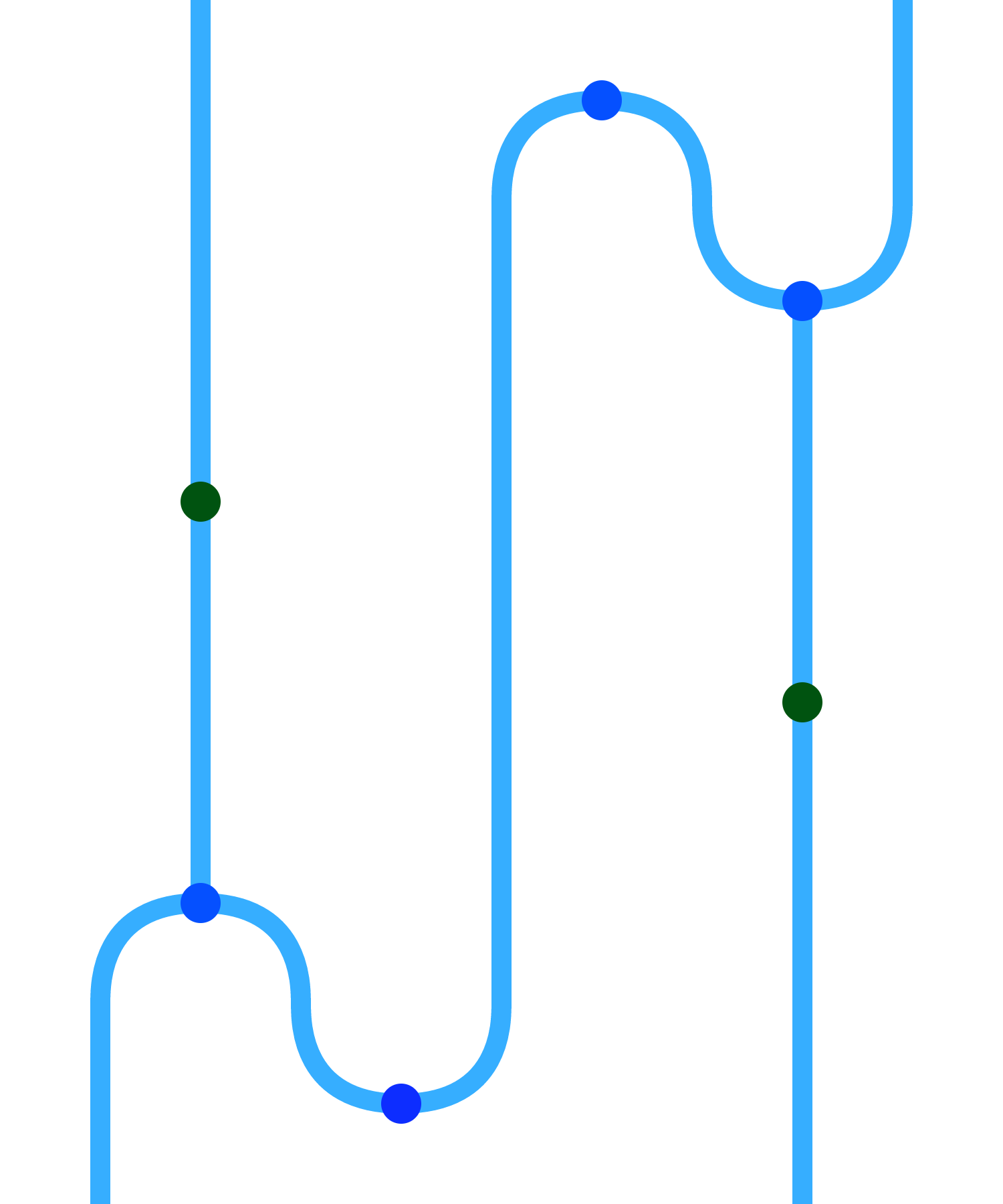}};
\draw [leftbox] (0.25,1.1) rectangle +(1.5,0.8);
\draw [leftbox] (0.25,3.1) rectangle +(1.5,.8);
\draw [rightbox] (3.25,2.1) rectangle +(1.5,.8);
\draw [rightbox] (3.25,4.1) rectangle +(1.5,.8);
\end{tz}
\ignore{{\xto\sim}
\begin{tz}[sc]
\node (1) [anchor=south west, inner sep=0pt] at (0,0) {\tikzpng[scale=2]{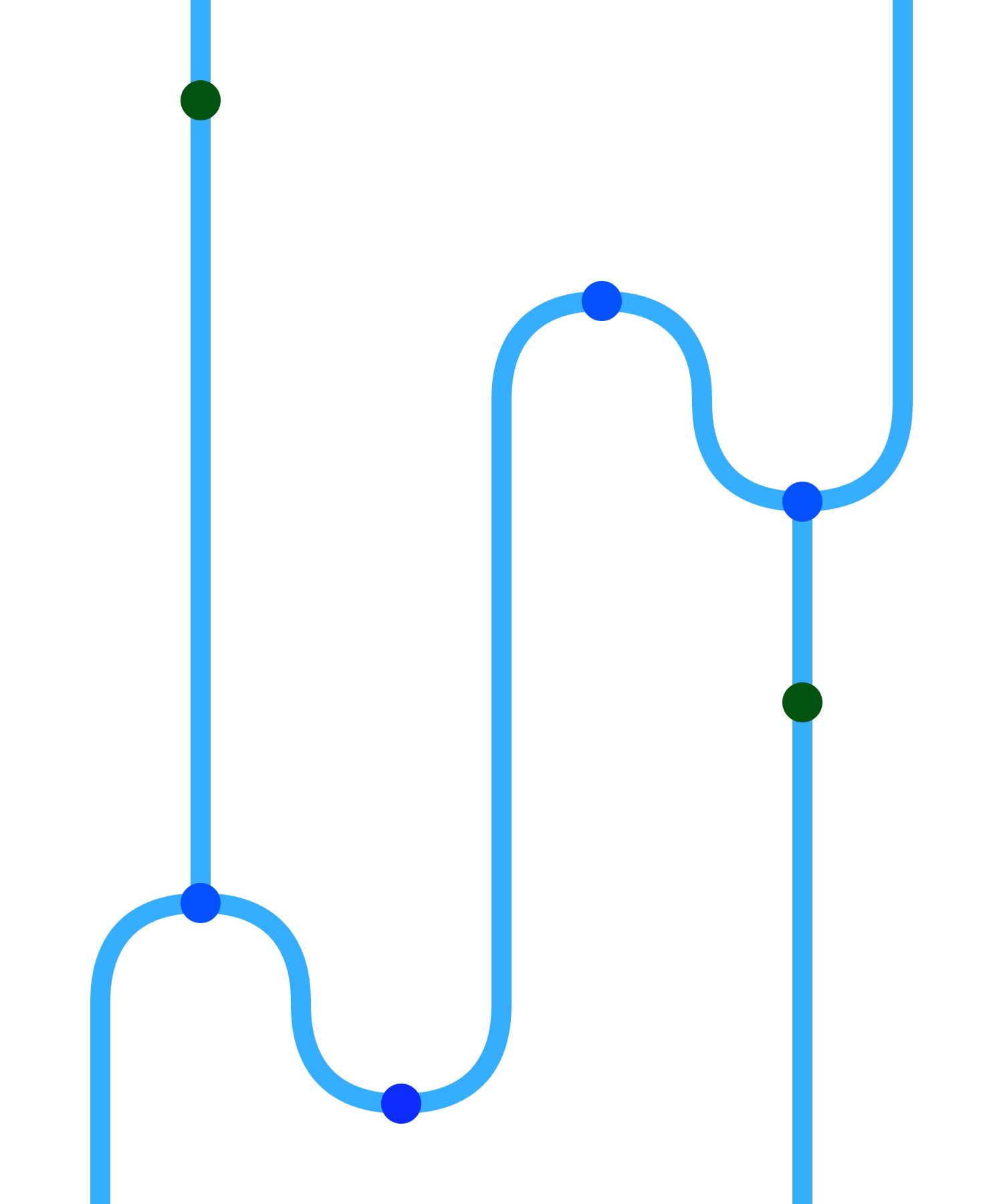}};
\draw [leftbox] (0.25,1.1) rectangle +(1.5,.8);
\draw [leftbox] (0.25,5.1) rectangle +(1.5,.8);
\draw [rightbox] (3.25,2.1) rectangle +(1.5,.8);
\draw [rightbox] (3.25,3.1) rectangle +(1.5,.8);
\end{tz}}
{\xto\sim}
\begin{tz}[sc]
\node (1) [anchor=south west, inner sep=0pt] at (0,0) {\tikzpng[scale=2]{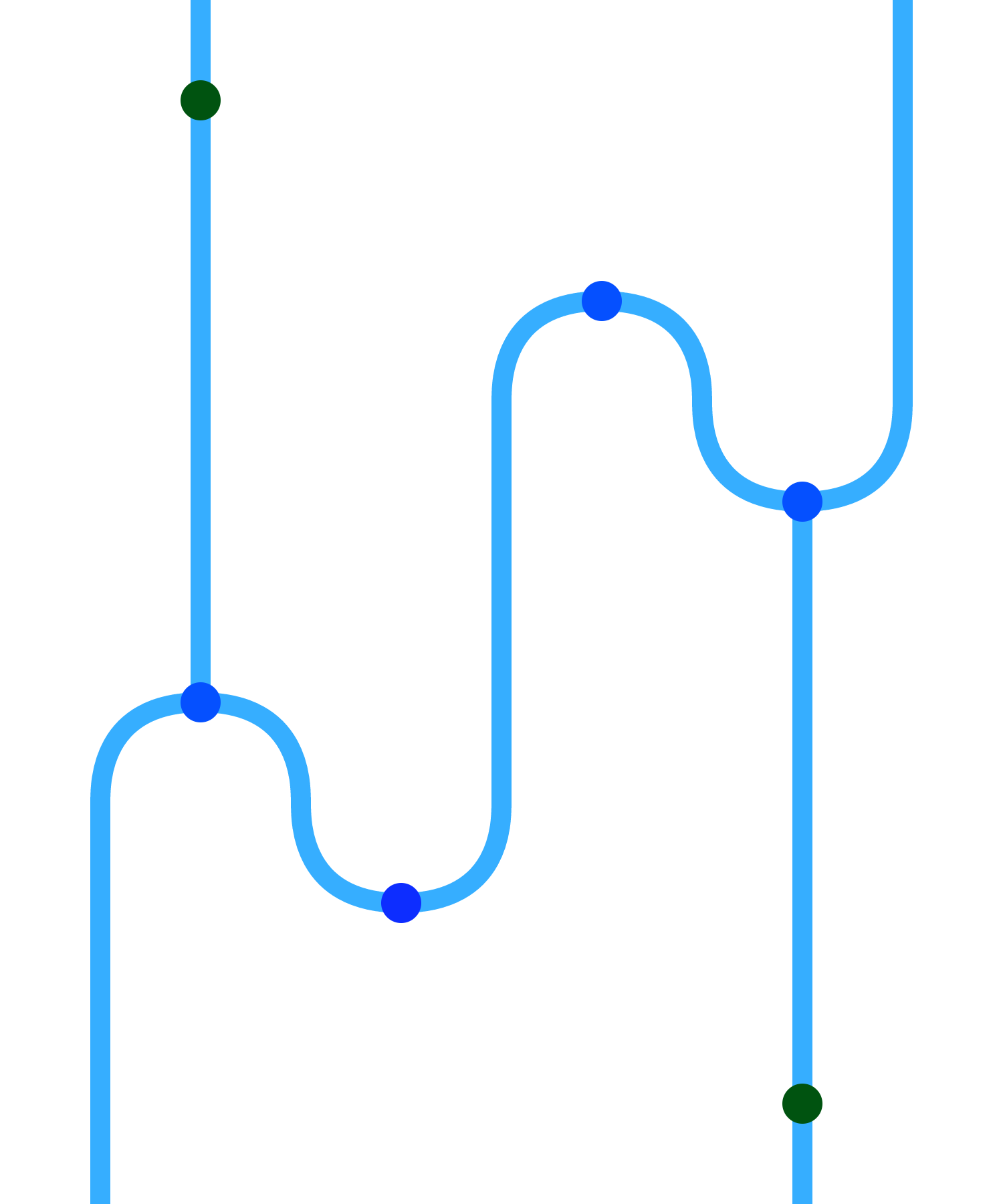}};
\draw [leftbox] (0.25,2.1) rectangle +(1.5,.8);
\draw [leftbox] (0.25,5.1) rectangle +(1.5,.8);
\draw [rightbox] (3.25,0.3) rectangle +(1.5,.8);
\draw [rightbox] (3.25,3.1) rectangle +(1.5,.8);
\end{tz}
{\xto\sim}
\begin{tz}[sc]
\node (1) [anchor=south west, inner sep=0pt] at (0,0) {\tikzpng[scale=2]{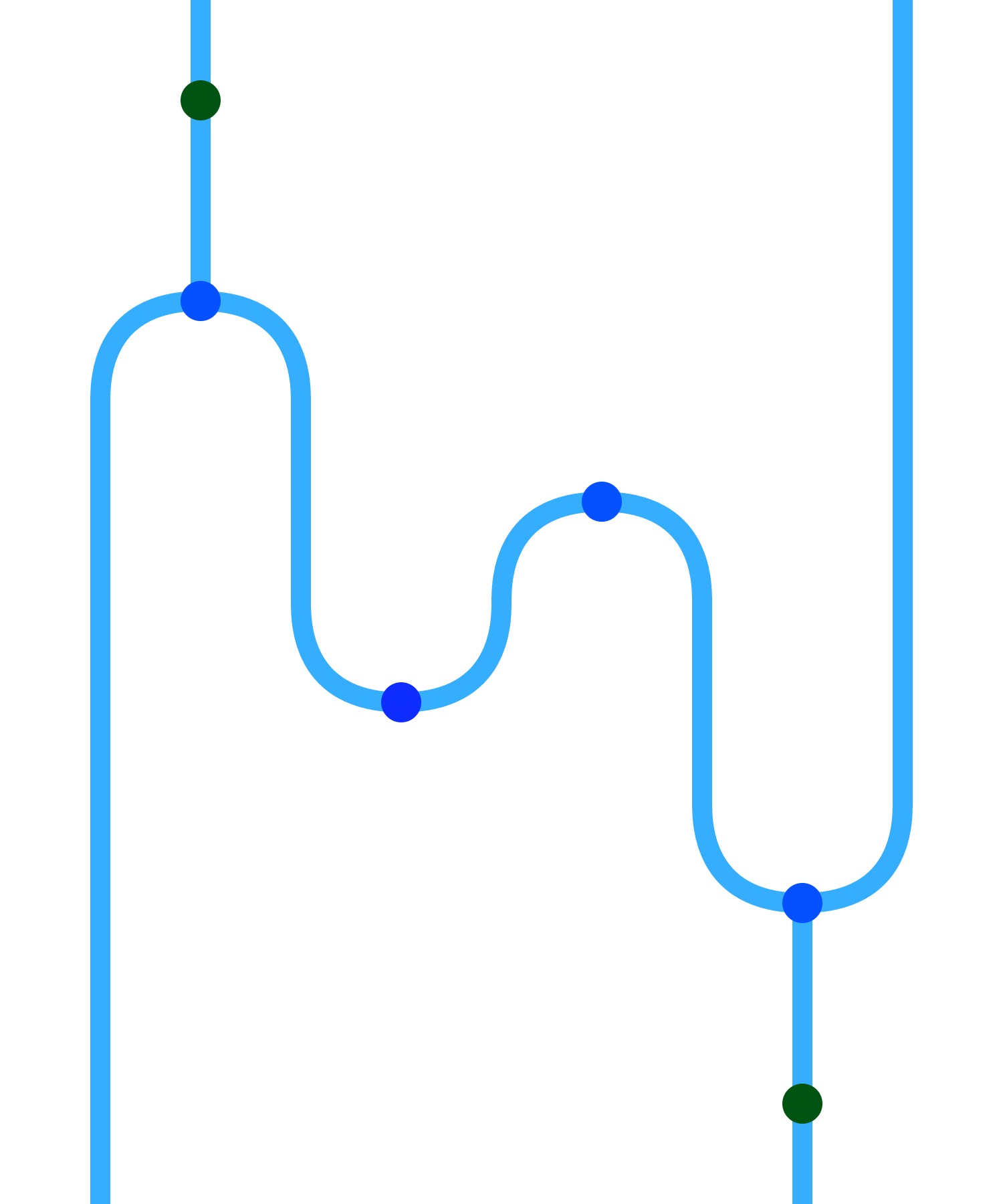}};
\draw [leftbox] (0.25,4.1) rectangle +(1.5,.8);
\draw [leftbox] (0.25,5.1) rectangle +(1.5,.8);
\draw [rightbox] (3.25,0.3) rectangle +(1.5,.8);
\draw [rightbox] (3.25,1.3) rectangle +(1.5,.8);
\end{tz}
{\xto{\sn_2}}
\begin{tz}[sc]
\node (1) [anchor=south west, inner sep=0pt] at (0,0) {\tikzpng[scale=2]{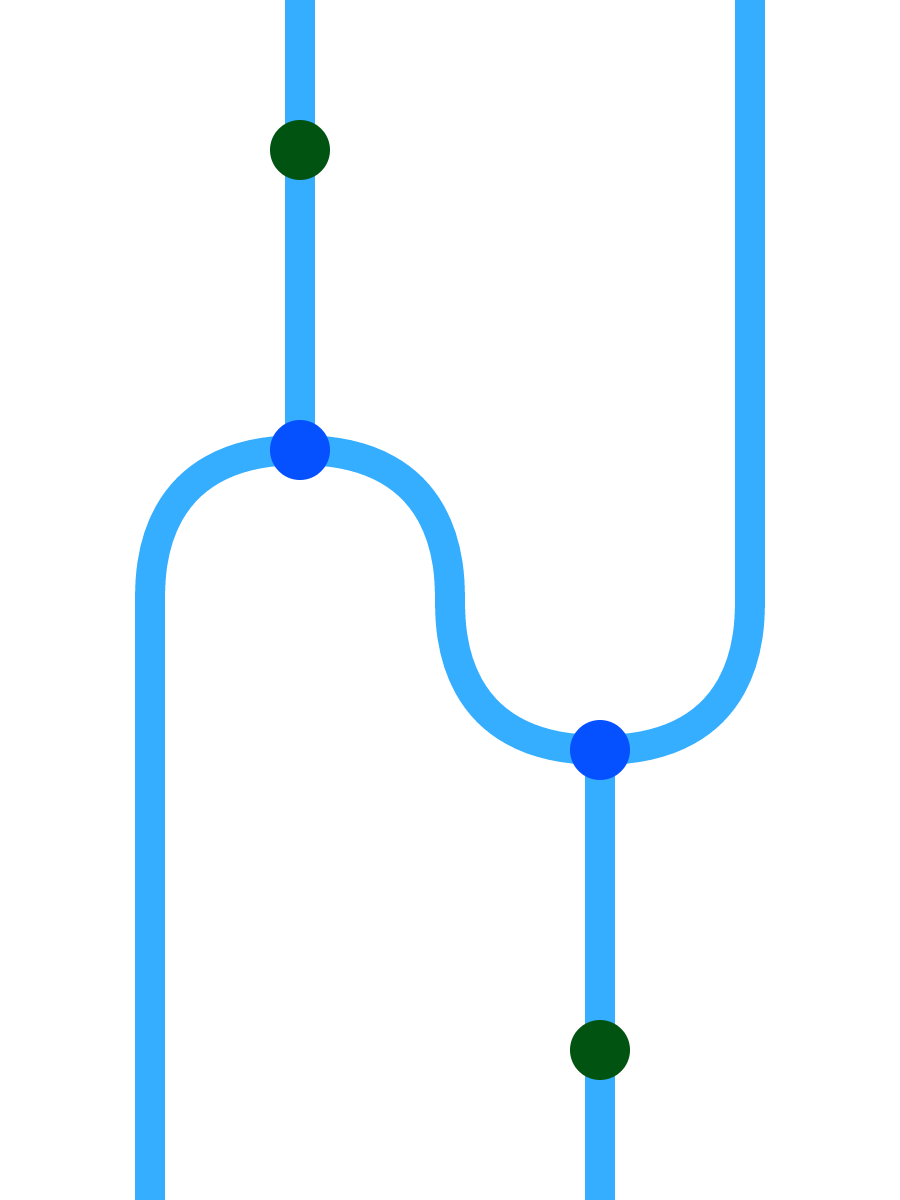}};
\draw [leftbox] (0.25,2.15) rectangle +(1.5,.75);
\draw [leftbox] (0.25,3.1) rectangle +(1.5,.75);
\draw [rightbox] (1.25,0.2) rectangle +(1.5,.75);
\draw [rightbox] (1.25,1.15) rectangle +(1.5,.75);
\end{tz}
\]
\end{definition}

\begin{lemma}
\label{lem:srsnatural}
The snake removal scheme is natural with respect to transformations on obstructing groups.
\end{lemma}
\begin{proof}
In the snake removal scheme, obstructing blocks are acted on by interchangers only, which are natural  with respect to transformations of their arguments.
\end{proof}
\begin{example}
To illustrate this naturalness property, suppose that $P : X \to Y$ is a 2\-morphism in \free \E. Then the following equation holds, where we write SRS for the 2\-morphism constructed by the snake removal scheme:
\tikzset{sc/.style={scale=0.51}}
\begin{equation*}
\begin{tz}[xscale=3,yscale=2.5]
\node (A) at (0,0) {$\begin{tikzpicture}[sc]
    \node [anchor=south west, inner sep=0pt] (1) at (0,0) {\tikzpng[scale=2]{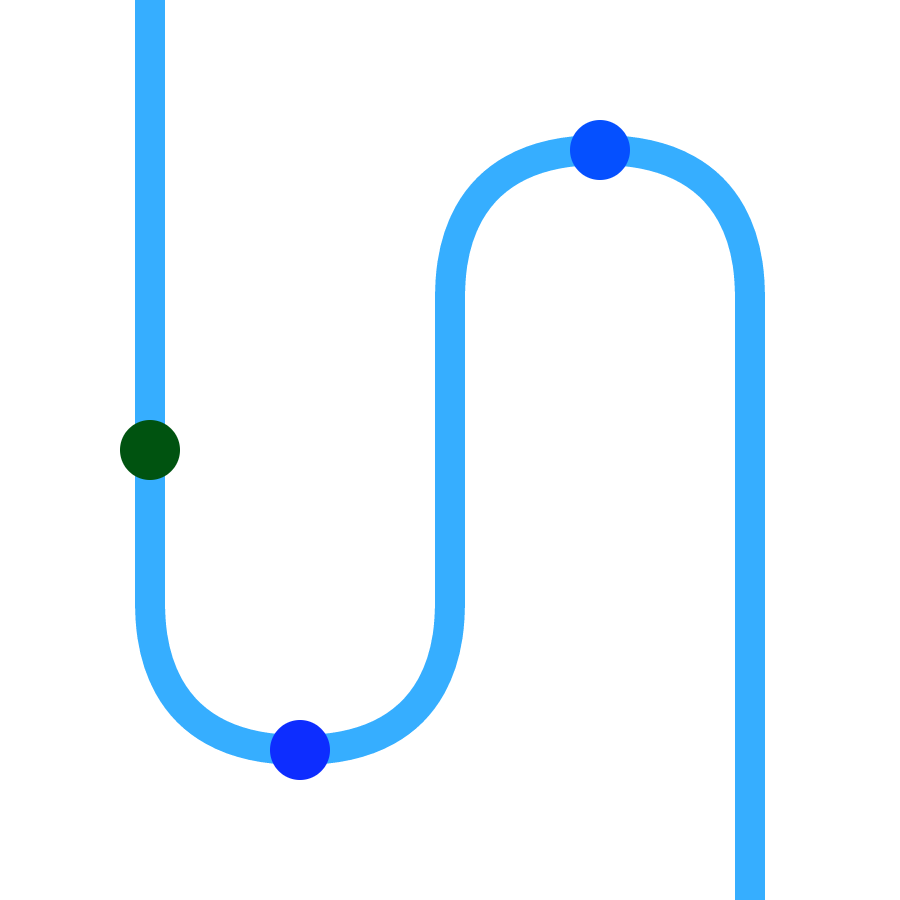}};
    \draw [fill=white, draw=lightblue, line width=1.2pt] (0,1) rectangle +(1,1);
    \node at (0.5,1.5) {$X$};
\end{tikzpicture}$};
\node (B) at (1,0) {$\begin{tikzpicture}[sc]
    \node [inner sep=0pt] (1) at (0,0) {\tikzpng[scale=2.5]{bigidentity}};
    \draw [fill=white, draw=lightblue, line width=1.2pt] (-0.5,-0.5) rectangle +(1,1);
    \node at (0,0) {$X$};
\end{tikzpicture}$};
\node (C) at (0,-1) {$\begin{tikzpicture}[sc]
    \node [anchor=south west, inner sep=0pt] (1) at (0,0) {\tikzpng[scale=2]{killsnakeexample1}};
    \draw [fill=white, draw=lightblue, line width=1.2pt] (0,1) rectangle +(1,1);
    \node at (0.5,1.5) {$Y$};
\end{tikzpicture}$};
\node (D) at (1,-1) {$\begin{tikzpicture}[sc]
    \node [inner sep=0pt] (1) at (0,0) {\tikzpng[scale=2.5]{bigidentity}};
    \draw [fill=white, draw=lightblue, line width=1.2pt] (-0.5,-0.5) rectangle +(1,1);
    \node at (0,0) {$Y$};
\end{tikzpicture}$};
\draw [->] (A) to node [above] {SRS} (B);
\draw [->] (A) to node [left] {$P$} (C);
\draw [->] (C) to node [below] {SRS} (D);
\draw [->] (B) to node [right] {$P$} (D);
\end{tz}
\end{equation*}
\end{example}

\begin{definition}[Connected, acyclic, simple]
A 1\-morphism in $\free \E$ is \emph{connected} or \emph{acyclic} when its string diagram graph  is connected or acyclic, respectively; it is \textit{simple} when it is connected and acyclic, and has a unique output wire.
\end{definition}

\begin{definition}[Twistedness]
Let $X$ be a simple 1\-morphism, and let $v$ be an $m$, $u$ or $f$ vertex in its string diagram. Then the \emph{twistedness} of $v$, written $\Tw(v)$, is the number of right turns minus the number of left turns along the shortest path from $v$ to the unique output. We say $X$ is \textit{untwisted} when for all such vertices $v$, $\Tw(v)=0$. We say $X$ is \emph{locally untwisted} when $\Tw(v)=0$ for $v$ in some subset of vertices.
\end{definition}

\newcommand\darkgreen[1]{{\color{green!70!black}#1}}
\newcommand\red[1]{{\color{red}#1}}
\begin{example}
Twistedness is best understood by example. The image below shows a simple 1\-morphism with four $m$ vertices (1, 3, 4, 5), one $u$ vertex (7), and two $f$ vertices (2, 6). The arrows show the direction of the shortest path to the output wire at the turning points, with red indicating a right turn, and green a left turn. For each of the numbered vertices, we compute the twistedness as the number of right turns minus the number of left turns.
\[
\ignore{
\tikzset{
  redcurvy/.pic={
    \draw [->, red, thick] (-0.125,-0.35) to [out=up, in=up, looseness=2.4] +(0.25,0);
  },
  greencurvy/.pic={
    \draw [->, green!70!black, thick] (-0.125,-0.35) to [out=up, in=up, looseness=2.4] +(0.25,0);
  }
}}\begin{tz}[scale=0.511]
\tikzset{number/.style={draw, line width=0.7pt, circle, inner sep=1pt, fill=blue!40, font=\scriptsize}}
\node [scale=0.2, anchor=south west] at (0,-1) {$\includegraphics{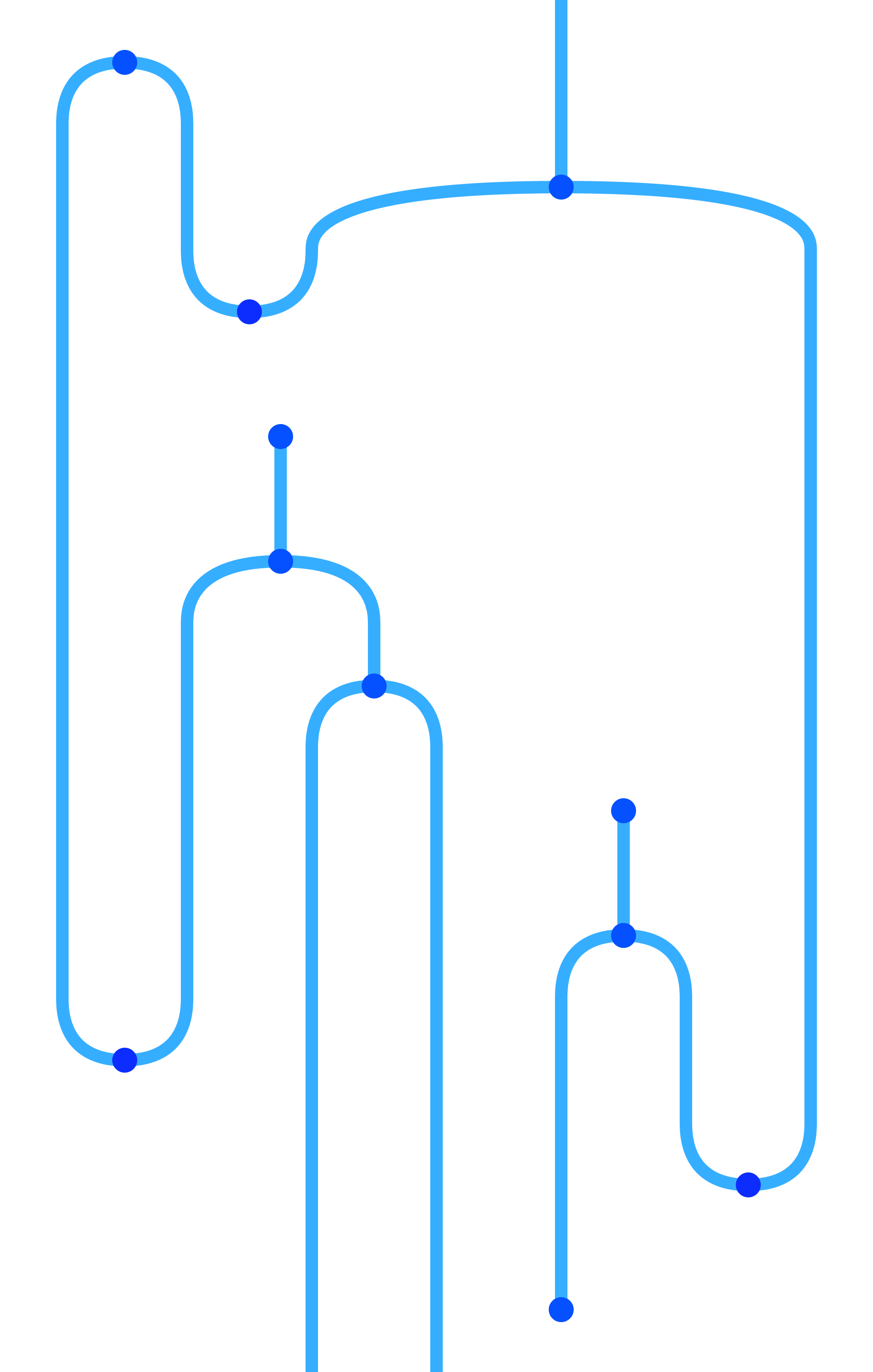}$};
\node [number] at (4.5,8.5) {$1$};
\node [number] at (2.25,6.5) {$2$};
\node [number] at (2.25,5.5) {$3$};
\node [number] at (3,4.5) {$4$};
\node [number] at (5,2.5) {$5$};
\node [number] at (5,3.5) {$6$};
\node [number] at (4.5,-0.5) {$7$};
\draw [->, red, thick] (1.3,2.1) to [out=down, in=down, looseness=2.4] +(-0.5,0);
\draw [->, red, thick] (0.76,8.8) to [out=up, in=up, looseness=2.4] +(0.5,0);
\draw [->, red, thick] (4.76,1.7) to [out=up, in=up, looseness=2.4] +(0.5,0);
\draw [->, green!70!black, thick] (5.76,1.18) to [out=down, in=down, looseness=2.4] +(0.5,0);
\draw [->, green!70!black, thick] (1.76,8.18) to [out=down, in=down, looseness=2.4] +(0.5,0);
\draw [->, green!70!black, thick] (2.65,4.64) to [out=up, in=up, looseness=2.] +(-0.75,0);
\node [green!70!black, font=\scriptsize, below=0.12cm] at (2.28,5.4) {$D$};
\node [green!70!black, font=\scriptsize, above=0.18cm] at (2.02,7.5) {$A$};
\node [red, font=\scriptsize, above=0.18cm] at (1.05,1.5) {$C$};
\node [red, font=\scriptsize, below=0.18cm] at (1.02,9.5) {$B$};
\node [red, font=\scriptsize, below=0.18cm] at (5.02,2.4) {$F$};
\node [green!70!black, font=\scriptsize, above=0.18cm] at (6,0.5) {$E$};
\begin{scope}[xshift=7.5cm, yshift=7.5cm]
\node [anchor=west] at (0,0) {$\Tw(1) = 0-0=0$};
\node [anchor=west] at (0,-1) {$\Tw(2) = 2-1=1$ (\red {$C$}, \red {$B$}, \darkgreen {$A$})};
\node [anchor=west] at (0,-2) {$\Tw(3) = 2-1=1$ (\red {$C$}, \red {$B$}, \darkgreen {$A$})};
\node [anchor=west] at (0,-3) {$\Tw(4) = 2-2=0$ (\darkgreen{$D$}, \red {$C$}, \red {$B$}, \darkgreen {$A$})};
\node [anchor=west] at (0,-4) {$\Tw(5) = 0-1 = {-}1$ (\darkgreen{$E$})};
\node [anchor=west] at (0,-5) {$\Tw(6) = 0-1 = {-}1$ (\darkgreen{$E$})};
\node [anchor=west] at (0,-6) {$\Tw(7) = 1-1 = 0$ (\red{$F$}, \darkgreen{$E$})};
\end{scope}
\end{tz}
\]
Since some $m$, $u$ or $f$ vertices have nonzero twistedness, this 1\-morphism is not untwisted. Note that for computing $\Tw(2)$,  $\darkgreen{D}$ does not count as a left turn, since we are passing top-to-bottom through vertex 3; but for computing $\Tw(4)$, $\darkgreen D$ does count as we are passing from bottom-right to bottom-left through vertex 3.
\end{example}

Untwisted simple 1\-morphisms have some good properties, which we now explore.
\begin{lemma}
\label{lem:nof}
An untwisted simple 1\-morphism has no $f$ vertices.
\end{lemma}
\begin{proof}
Suppose $v$ is an $f$ vertex. Then travelling along the shortest path to the unique output, suppose there are $L$ left turns and $R$ right turns. Since the diagram is untwisted $L=R$, and hence $L+R$ is an even number. Also, the only path out of an $f$ vertex begins travelling downwards. So the path begins travelling downwards, turns left or right an even number of times, and finishes at the unique output travelling upwards; a contradiction.
\end{proof}

\begin{lemma}
\label{lem:alwayscupcap}
In an untwisted simple 1\-morphism, every wire is either straight, or contains an eliminable $(\cup, \cap)$ pair. 
\end{lemma}
\begin{proof}
Choose a wire $W$, and let $v$ be the $m$ or $u$ vertex at the end closest to the unique output, and let $v'$ be the $m$ or $u$ vertex at the other end. Write $L$ for the number of left-turns along the path from $v'$ to $v$, and write $R$ for the number of left turns. Since the composite is untwisted at every vertex, we must have $L=R$. Suppose $L=0$; then the wire is straight. Suppose $L>0$; then the must be at least one adjacent ``turn left--turn right'' or ``turn right--turn left'' pair along the path from $v'$ to $v$, and such a pair is eliminable.
\end{proof}

We now consider diagrams which arise from the pseudomonoid data only.
\begin{definition}
A 1\-morphism in \free \E is in \emph{pseudomonoid form} when it is simple and formed only from $m$ and $u$ generators; that is, just when it is connected and in the image of the obvious embedding $\free \P \to \free \E$.
\end{definition}

\begin{lemma}
\label{lem:pseudomonoiduntwisted}
A 1\-morphism in pseudomonoid form is untwisted.
\end{lemma}
\begin{proof}
Clearly a 1\-morphism in pseudomonoid form is simple. From any $m$ or $u$ vertex, we build a path that travels upwards through the string diagram as far as possible. Suppose we do not reach the unique output at the top of the diagram: then we must have encountered some component with no output legs. But that is impossible, since the 1\-morphism is built solely from $m$ and $u$ vertices, which do have output legs. So from any vertex, there is a path to the unique output with no left or right turns; hence the diagram is untwisted.
\end{proof}

\begin{lemma}
\label{lem:straighten}
For every simple 1\-morphism $X$ in \free \E, there is a rotational 2\-morphism $\Omega_X: X \stackrel R \twoheadrightarrow \widetilde X$ where $\widetilde X$ is in pseudomonoid form.
\end{lemma}
\begin{proof}
The generators $R_m$, $R_u$ and $R_f$ decrease the twistedness of the diagram locally, and $L_m$, $L_u$ and $L_f$ increase it. We apply these generators repeatedly to obtain a rotational 2\-morphism $X \twoheadrightarrow X'$, where $X'$ is untwisted. By \autoref{lem:nof}, $X'$ has no $f$ vertices.

The 1\-morphism $X'$ is untwisted and simple, so by \autoref{lem:alwayscupcap} every wire is either straight, or contains an eliminable cup-cap pair. We now proceed to eliminate every cup and cap, by induction on the total number of cups and caps. Suppose every wire is straight: then we are done. Otherwise, we apply the snake removal scheme of \autoref{def:srs} to obtain a rotational 2\-morphism $X' \twoheadrightarrow X''$. Note that $X''$ contains strictly fewer cups and caps than $X'$, and is simple since $X'$ is, and is untwisted since it was produced from an untwisted composite by removing an adjacent left-right moving pair. So by induction, we are done, and we have eliminated all $\cup$ and $\cap$ generators.

The result is in pseudomonoid form, since by construction it is a simple 1\-morphism with no $f$, $\cup$ or $\cap$ generators.
\end{proof}

We now define a stronger variant of pseudomonoid form.
\begin{definition}
A 1\-morphism in \free \E is in \emph{left-pseudomonoid form} when it is identical to $u$ or $m$, or is identical to the following composite, where $U,V$ are in left-pseudomonoid form:
\tikzset{sc/.style={scale=0.75}}
\tikzset{box/.style={draw=lightblue, fill=white, line width=1.7pt}}
\[
\begin{tz}[sc]
\node (1) [anchor=south west, inner sep=0pt] at (0,0) {\tikzpng[scale=3]{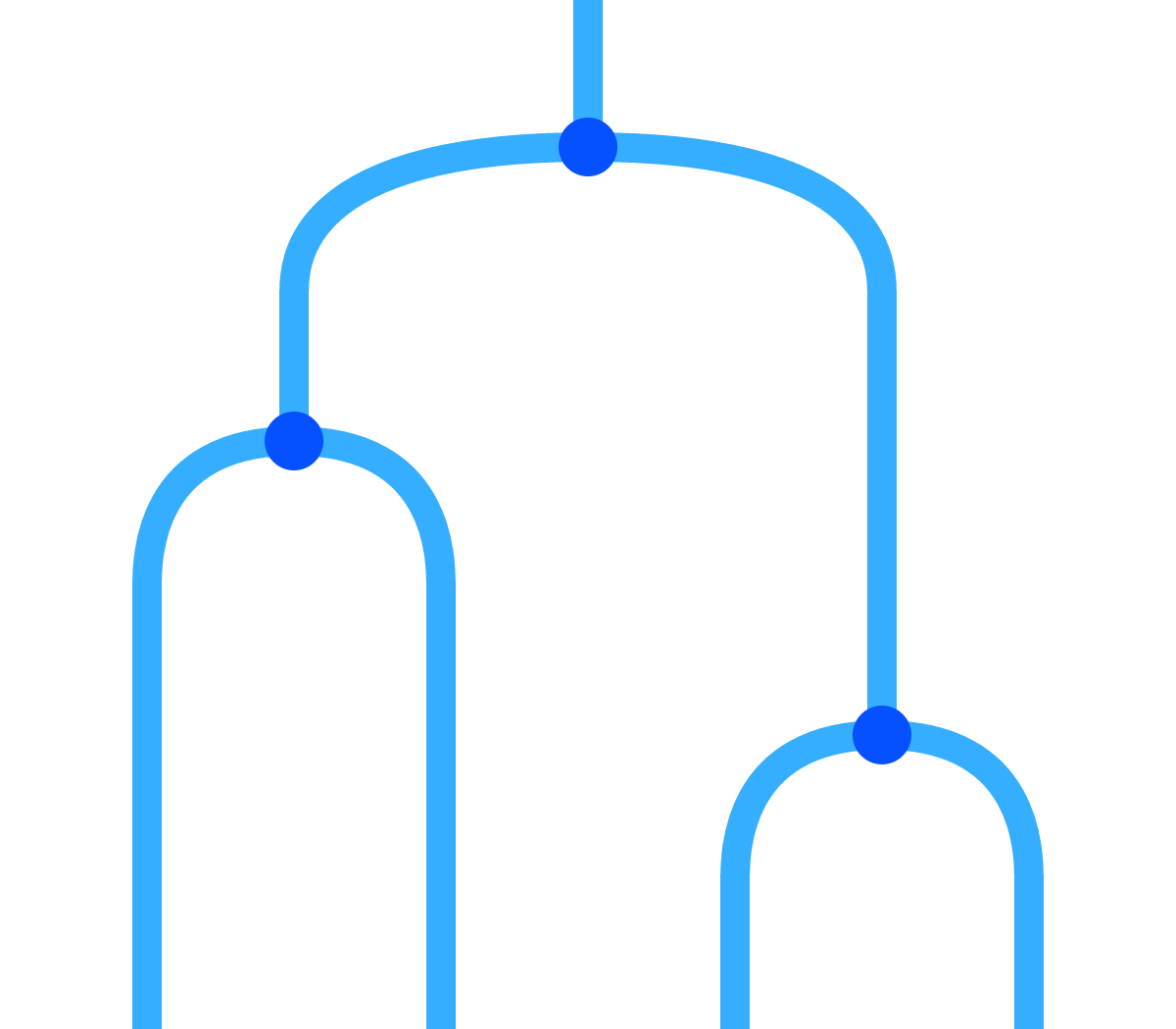}};
\draw [box] (0.25,1.5) rectangle +(1.5,1);
\draw [box] (2.25,0.5) rectangle +(1.5,1);
\node at (1,2){$U$};
\node at (3,1){$V$};
\node at (1.06,0.18) {$\cdots$};
\node at (3.06,0.18) {$\cdots$};
\end{tz}
\]
\end{definition}
\begin{lemma}
\label{lem:leftpseudomonoid}
For any 1\-morphism $X$ in \free \E in pseudomonoid form, there exists $\Theta_X: X \stackrel \sim \twoheadrightarrow \widehat X$ where $\widehat X$ is in left-pseudomonoid form.
\end{lemma}
\begin{proof}
We argue by induction on the number of $m$ and $u$ vertices. Suppose $X \equiv u$; then we are done. Otherwise the uppermost vertex in $X$ is $m$, and we write $L$ and $R$ for the sets of vertices in the subtrees below the left and right legs of $m$ respectively. We apply interchangers to move all the vertices of $L$ above all the vertices of $R$. By the inductive hypothesis, we use interchangers to arrange $L$ and $R$ in left-pseudomonoid form, and we are done.
\end{proof}

\noindent
Left-pseudomonoid form is useful since it is unaltered by rotational 2\-morphisms, as we now explore.

\begin{lemma}
\label{lem:leftpseudomonoidunique}
Let $X,Y$ be in left-pseudomonoid form, with $X \stackrel R \twoheadrightarrow Y$. Then $X \equiv Y$.
\end{lemma}
\begin{proof}
Given a simple 1\-morphism, its \emph{structure tree} is a binary tree, defined as follows. Start at the unique output, and travel along the path of the string diagram. If a $u$ or $f$ vertex is encountered, append a leaf to the tree. If an $m$ vertex is encountered, append a binary branching, with subtrees given by following the other two legs of the $m$ vertex in a similar way.

By inspection, none of the rotational 2\-morphisms change the structure tree of a diagram. Also, it is clear that if two left-pseudomonoid forms have the same structure tree, they are identical. This completes the proof.
\end{proof}

\ignore{ We now give a procedure to  arrange it in pseudomonoid form, by applying interchangers. Suppose  the uppermost vertex is a $u$; then we are done. Otherwise the uppermost vertex is an $m$, with subdiagrams $L$ and $R$ attached to the left and right input legs respectively. Note that $L$ and $R$ are not connected, since the original 1\-morphism was acyclic. We apply interchangers to arrange $L$ completely above $R$, and then apply this same procedure to each subdiagram respectively. The recursion terminates, since $L$ and $R$ have strictly fewer $m$ and $u$ vertices than the original 1\-morphism.}
\ignore{
We now follow the path in $X'$ from the unique output to the first $m$ or $u$ vertex encountered, which we label $v$. Since $X'$ is untwisted, this path has an equal number of cups and caps, and we eliminate them by applying $\sn_1$ and $\sn_2$ maps.

We now argue by induction on the total number of $m$ and $u$ vertices. Suppose $v$ is a $u$ vertex; then this is clearly the base case for the induction, and we are done, since there is only one wire in $X$ and it has been made straight. Alternatively, suppose $v$ is an $m$ vertex: then its two lower legs are connected to simple 1\-morphisms $Y$ and $Z$. Since $Y$ and $Z$ have strictly fewer vertices than $X$, by induction we construct $\Omega_Y : Y \to \widetilde Y$ and $\Omega_Z : Z \to \widetilde Z$ and apply them locally, and we are done.

The only complication that can arise is if the legs $Y$ and $Z$ are somehow entangled with one another. In this case, by acyclicity of $X$, we deduce that $Y$ and $Z$ are not connected to one another, and therefore that they can be separated using interchangers:
\[
\tikzset{sc/.style={scale=0.7}}
\hspace{-1cm}
\begin{tz}[sc]
\node [inner sep=0pt] (1) at (0,0) {\tikzpng[sc]{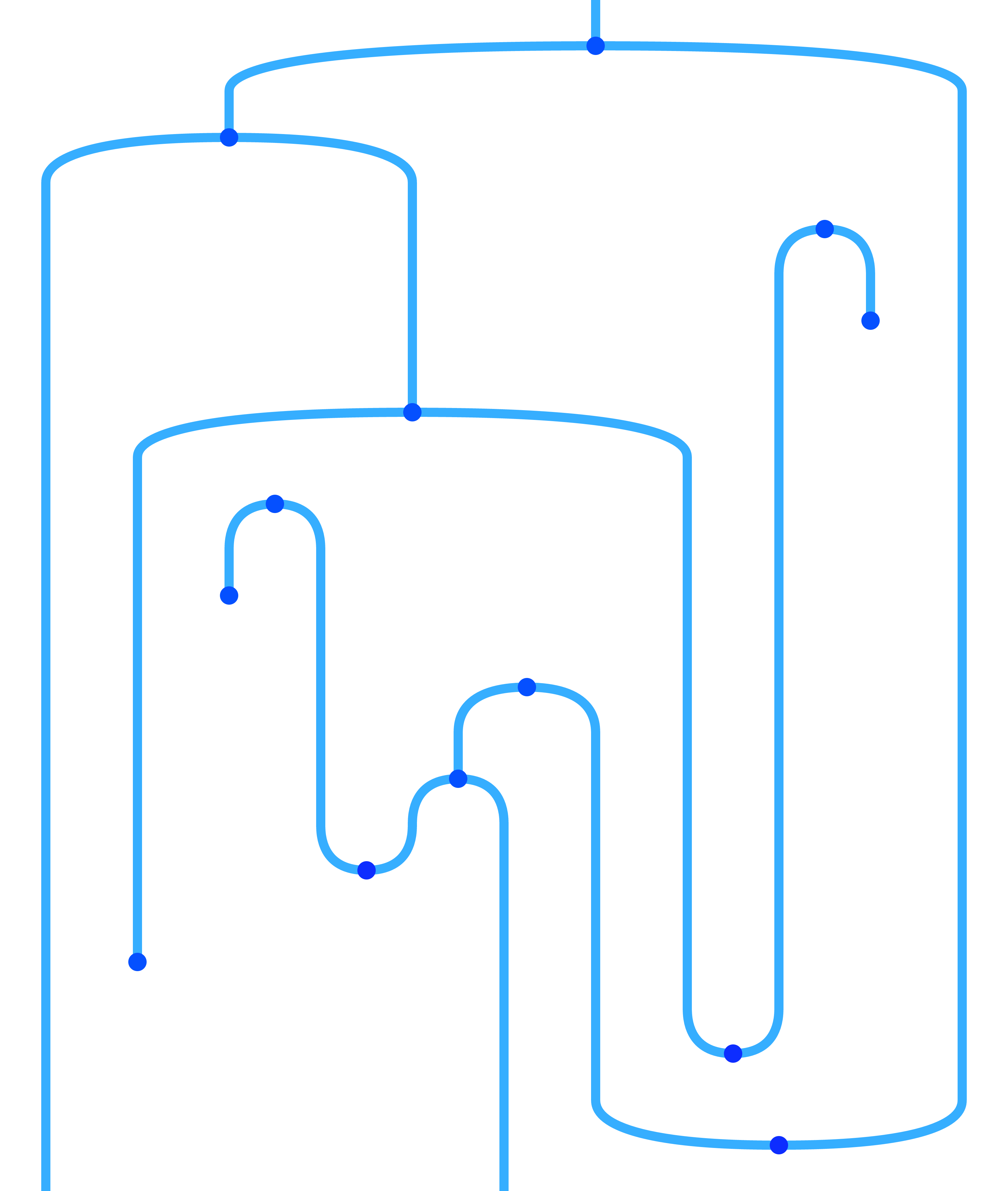}};
\node (2) [anchor=west, inner sep=0pt] at ([xshift=15pt] 1.east) {\tikzpng[sc]{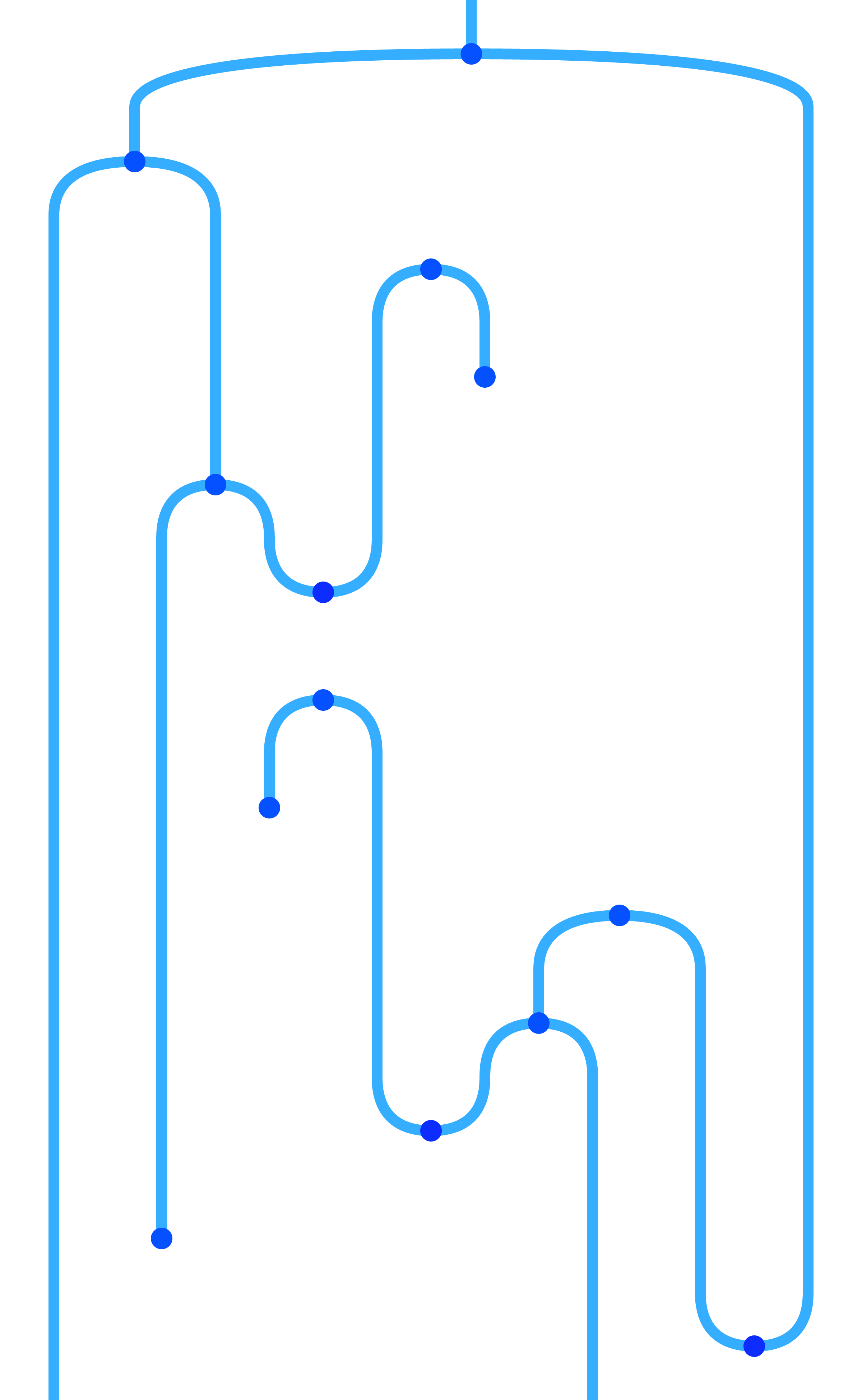}};
\node (3) [anchor=west, inner sep=0pt] at ([xshift=15pt] 2.east) {\tikzpng[sc]{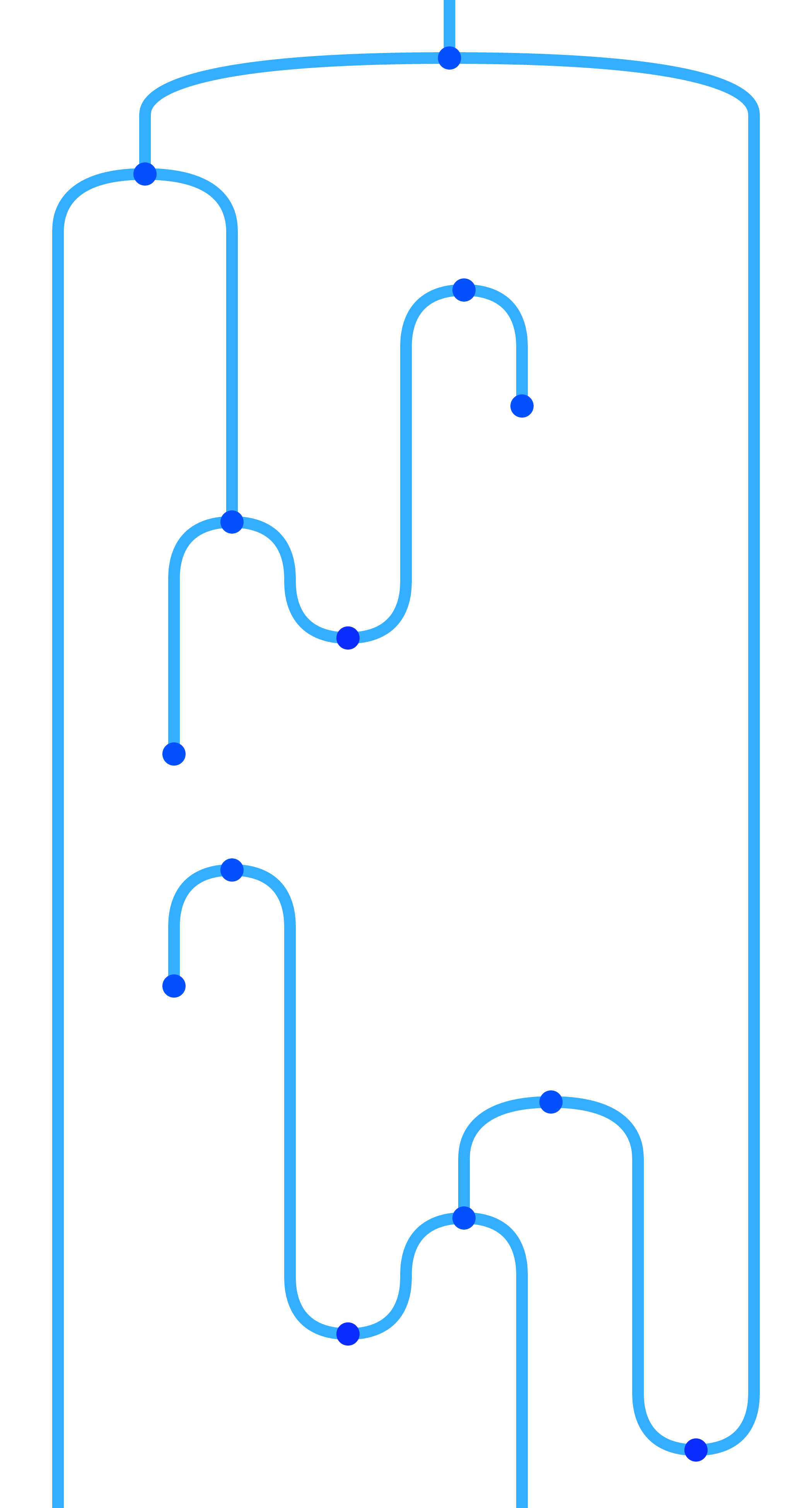}};
\node (4) [inner sep=0pt, anchor=west] at ([xshift=15pt] 3.east) {\tikzpng[sc]{disentangle3}};
\draw (5) [draw=lightblue, fill=white, line width=0.7pt] ([xshift=-0.55cm, yshift=-1.56cm] 4.center) node (A) {} rectangle +(1.36,1.32) node (B) {};
\node at ($(A)!0.5!(B)$) {$Z$};
\draw [draw=lightblue, fill=white, line width=0.7pt] ([xshift=-0.82cm, yshift=-0.04cm] 4.center) node (A) {} rectangle +(1.08,1.31) node (B) {};
\node at ($(A)!0.5!(B)$) {$Y$};
\node (6) [anchor=west, inner sep=0pt] at ([xshift=15pt] 4.east) {\tikzpng{disentangle4}};
\draw [->>] (1) to node [above] {$\sim$} (2);
\draw [->>] (2) to node [above] {$\sim$} (3);
\node at ($(3.east)!0.5!(4.west)$) {$\equiv$};
\draw [->>] (4) to node [above] {$\stackrel{\displaystyle\Omega_Y}{\Omega_Z}$} (6);
\end{tz}
\hspace{-1cm}
\]
This completes the proof.}

We give a name for a certain subset of the generators as follows.
\begin{definition}
In \free \E, the \emph{pseudomonoid 2\-cells} are $\alpha$, $\lambda$, $\rho$ and their inverses.
\end{definition}

\noindent
We now give the central proposition, showing how pseudomonoid 2\-cells can be ``untwisted'' to act in pseudomonoid form. The details of the proof are quite complex, and the proof assistant \texttt{globular.science} was essential for their development and presentation.
\begin{proposition}[Untwisting]
\label{prop:untwistalgebraic}
Let $P:X \to Y$ be a 2\-morphism in \free \E such that $X$ is simple. Then $P$ is equal to a 2\-morphism in which the pseudomonoid 2\-cells act on untwisted diagrams only.
\end{proposition}

\begin{proof}
We show that the following equations hold in \free \E, where $R$\ denotes some composite of rotational generators:
\begin{calign}
\nonumber
\begin{tz}[xscale=1.3, scale=1.5]
\node [inner sep=1pt] (1) at (0,0) {\tikzpng[scale=\myscale]{untwistalpha1}};
\node [inner sep=1pt] (2) at (0,-1.5) {\tikzpng[scale=\myscale]{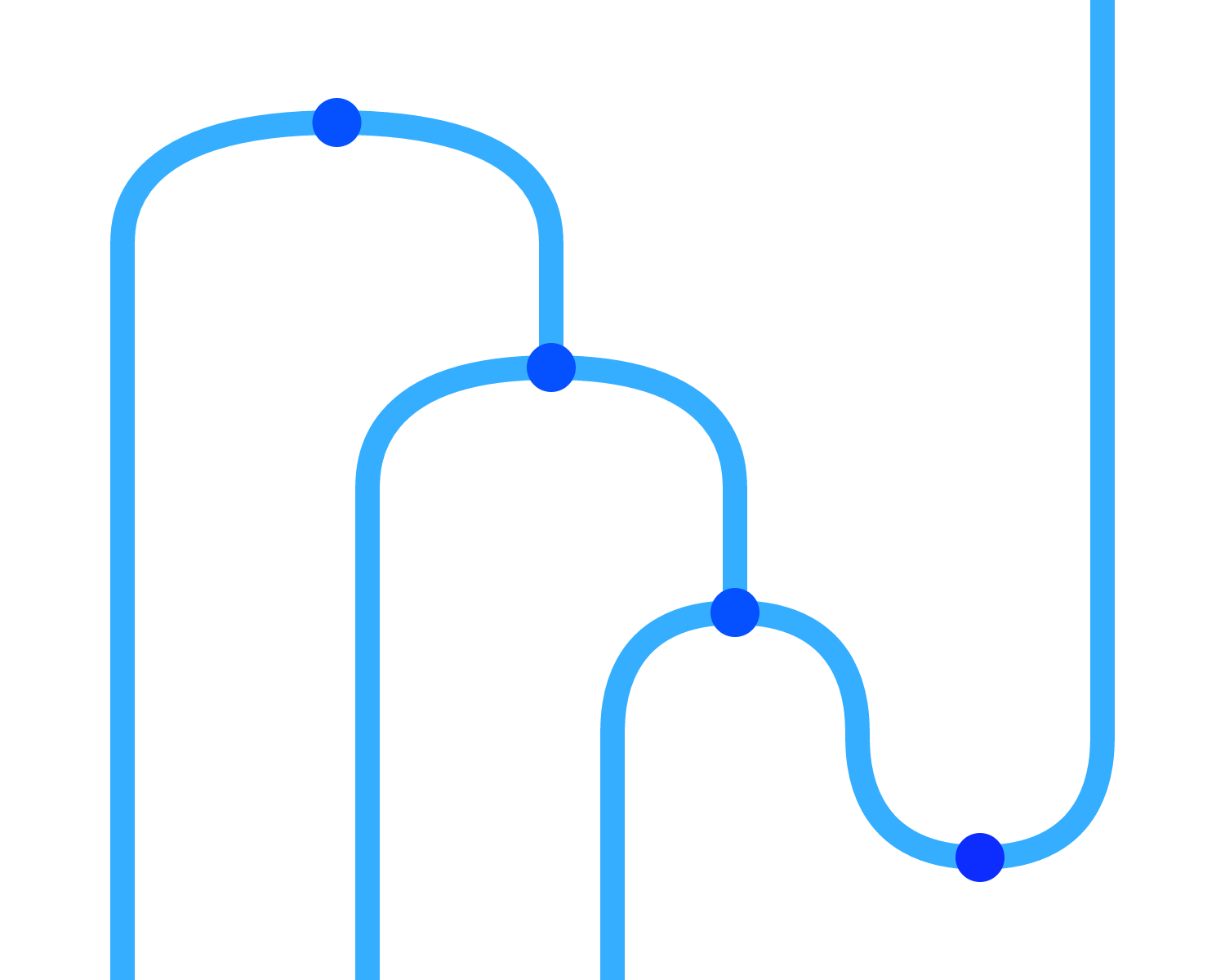}};
\node [inner sep=1pt] (3) at (2,-1.5) {\tikzpng[scale=\myscale]{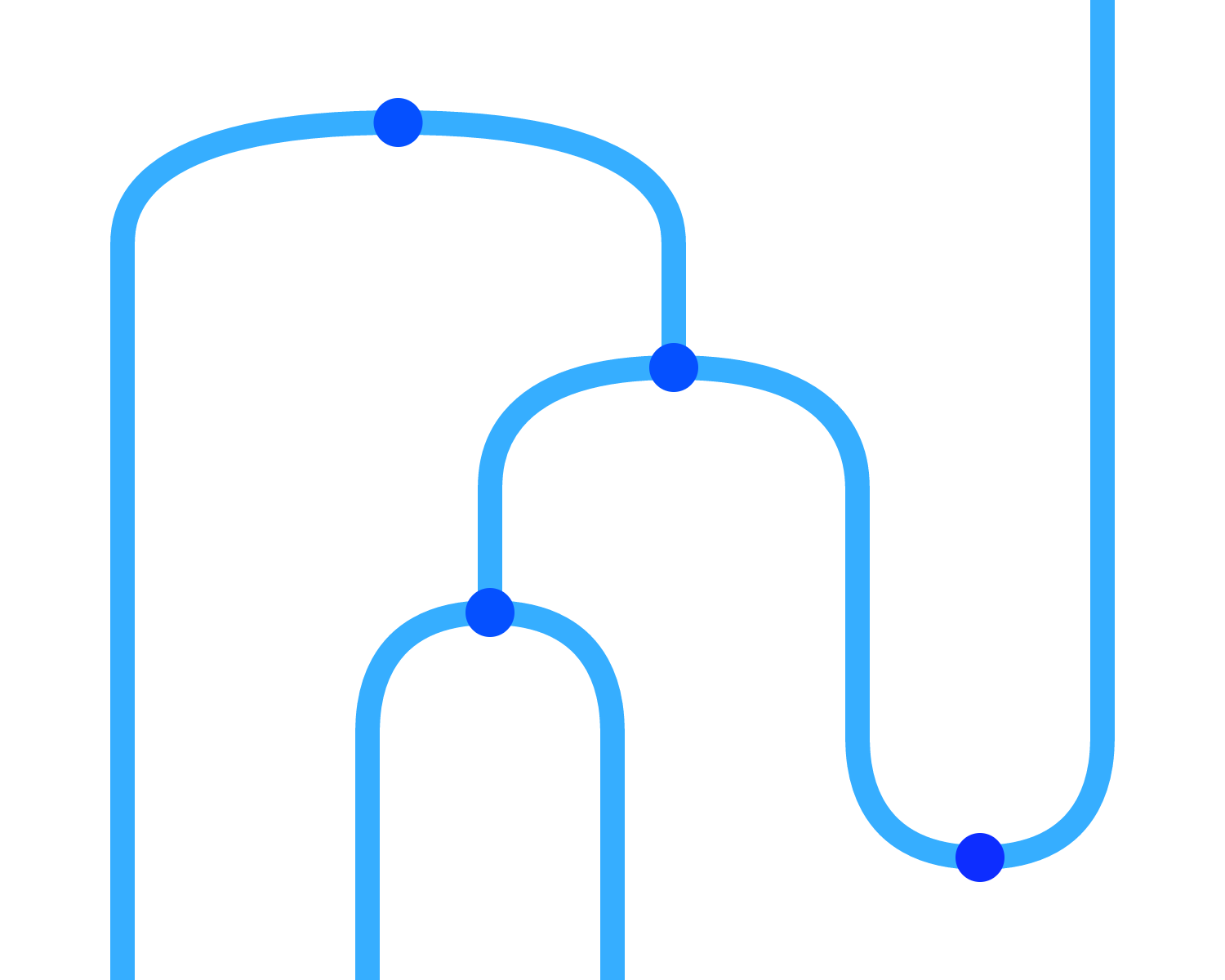}};
\node [inner sep=1pt] (4) at (2,0) {\tikzpng[scale=\myscale]{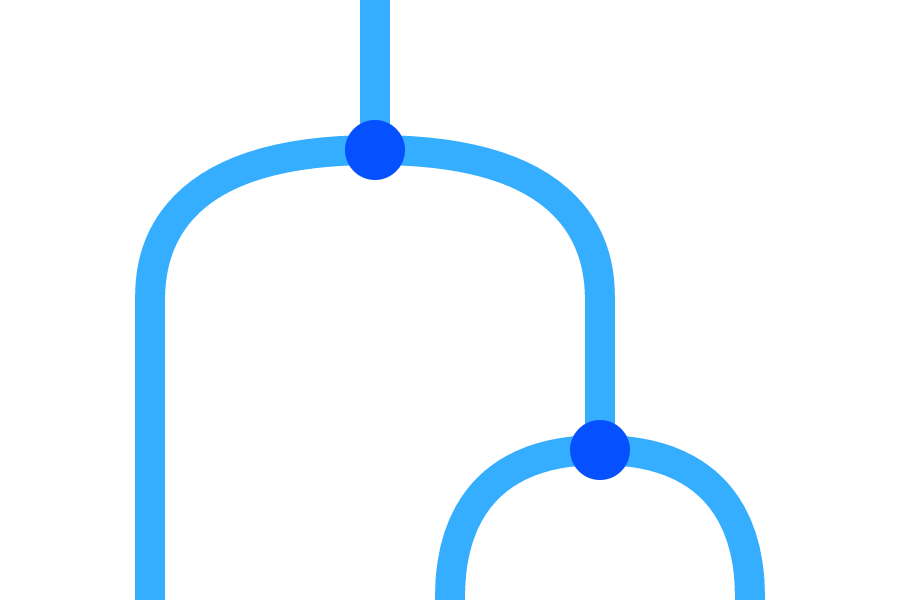}};
\draw [->>] (1) to node [right] {$R$} (2);
\draw [->] (1) to node [above] {$\alpha$} (4);
\draw [->] (2) to node [below] {$\alpha ^\inv$} (3);
\draw [->>] (3) to node [left] {$R$} (4);
\end{tz}
&
\begin{tz}[xscale=1.3, scale=1.5]
\node [inner sep=1pt] (1) at (0,0) {\tikzpng[xscale=1, scale=\myscale]{untwistalpha1}};
\node [inner sep=1pt] (2) at (0,-1.5) {\tikzpng[xscale=-1, scale=\myscale]{untwistalpha3}};
\node [inner sep=1pt] (3) at (2,-1.5) {\tikzpng[xscale=-1, scale=\myscale]{untwistalpha2}};
\node [inner sep=1pt] (4) at (2,0) {\tikzpng[scale=\myscale]{untwistalpha4}};
\draw [->>] (1) to node [right] {$R$} (2);
\draw [->] (1) to node [above] {$\alpha$} (4);
\draw [->] (2) to node [below] {$\alpha^\inv$} (3);
\draw [->>] (3) to node [left] {$R$} (4);
\end{tz}
\\
\nonumber
\begin{tz}[xscale=1.3, scale=1.5]
\node [inner sep=1pt] (1) at (0,0) {\tikzpng[scale=\myscale]{lambda_source}};
\node [inner sep=1pt] (2) at (0,-1.5) {\tikzpng[scale=\myscale]{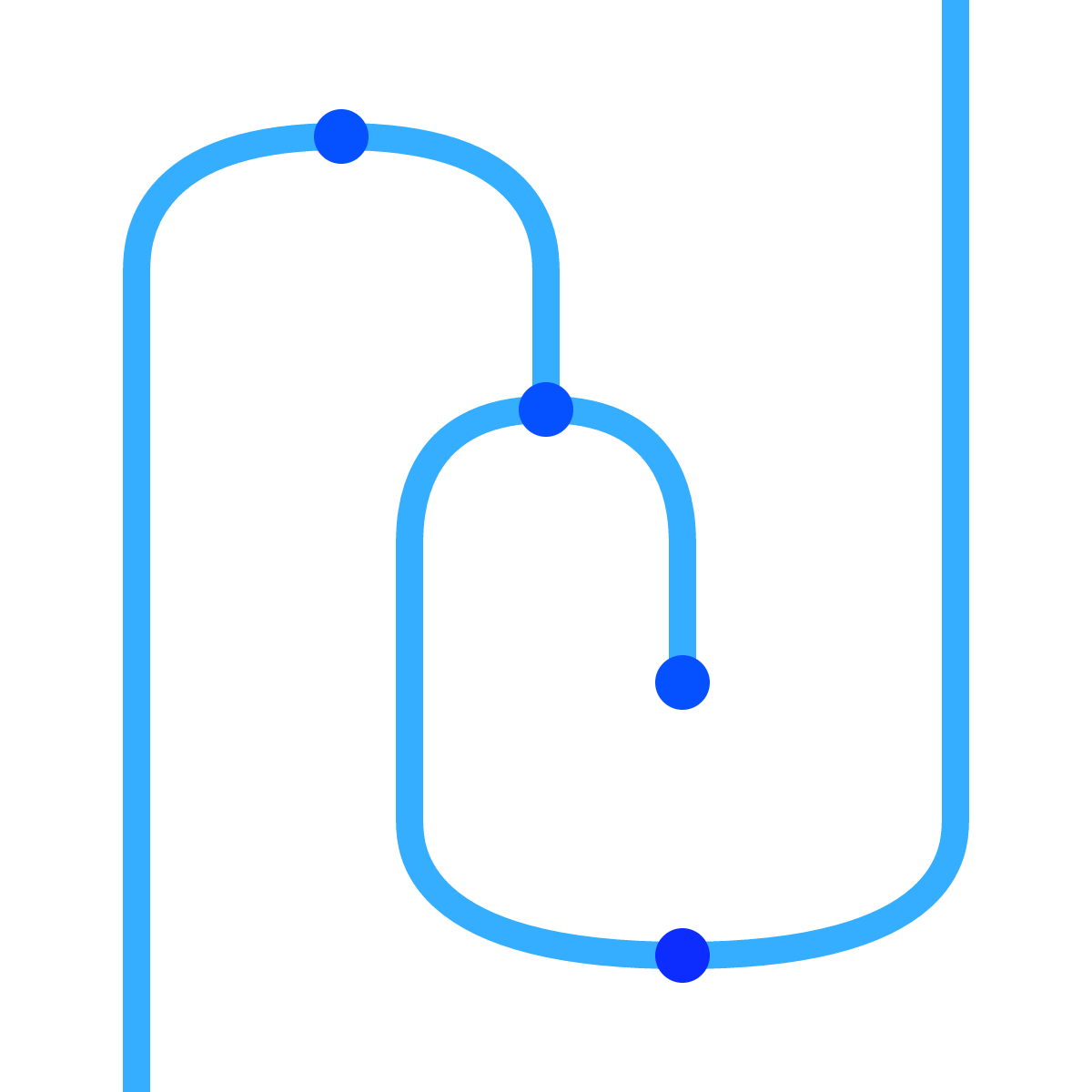}};
\node [inner sep=1pt] (3) at (2,-1.5) {\tikzpng[yscale=-1, scale=\myscale]{snake}};
\node [inner sep=1pt] (4) at (2,0) {\tikzpng[scale=\myscale]{bigidentity}};
\draw [->>] (1) to node [right] {$R$} (2);
\draw [->] (1) to node [above] {$\lambda$} (4);
\draw [->] (2) to node [below] {$\rho$} (3);
\draw [->>] (3) to node [left] {$R$} (4);
\end{tz}
&
\begin{tz}[xscale=1.3, scale=1.5]
\node [inner sep=1pt] (1) at (0,0) {\tikzpng[scale=\myscale]{lambda_source}};
\node [inner sep=1pt] (2) at (0,-1.5) {\tikzpng[xscale=-1, scale=\myscale]{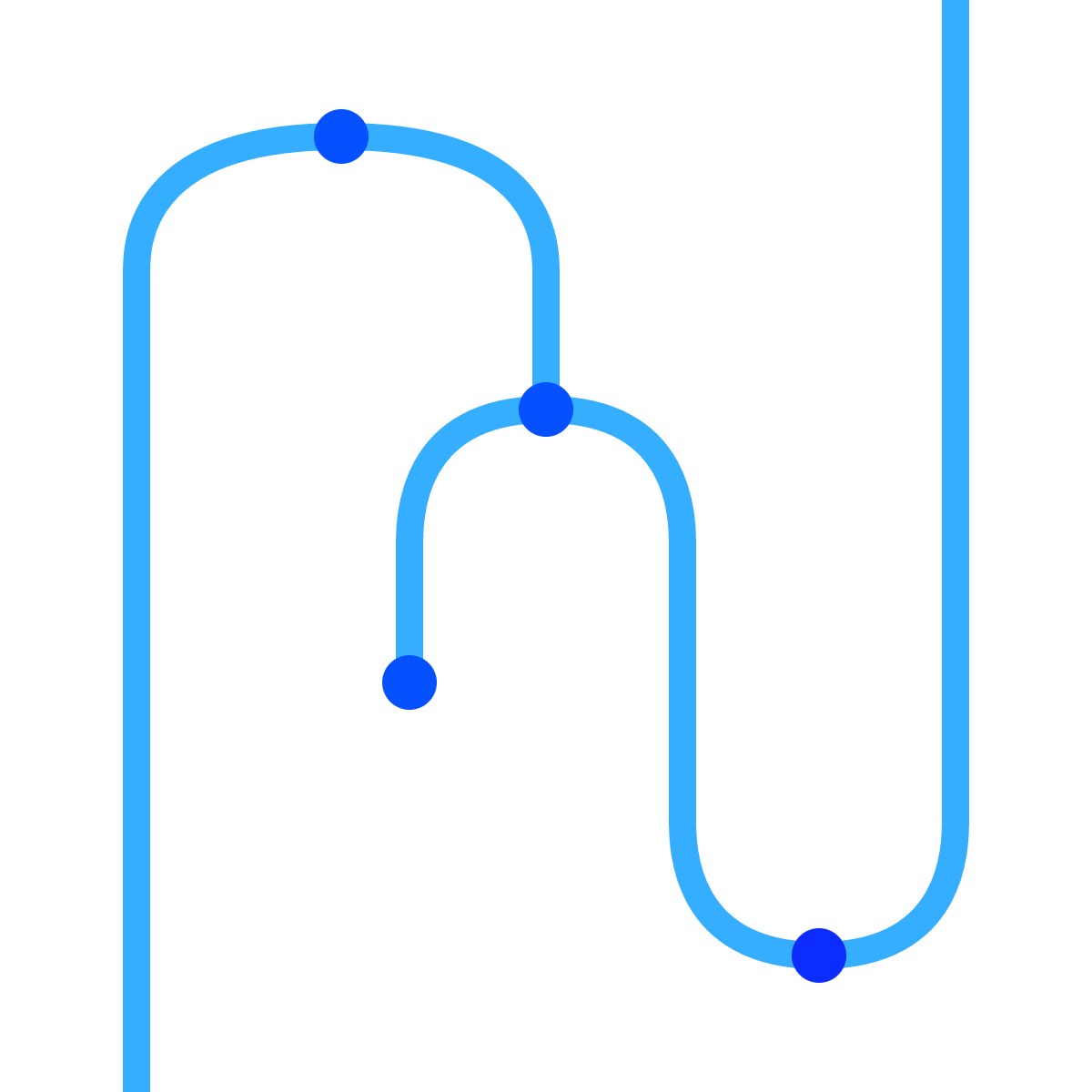}};
\node [inner sep=1pt] (3) at (2,-1.5) {\tikzpng[scale=\myscale]{snake}};
\node [inner sep=1pt] (4) at (2,0) {\tikzpng[scale=\myscale]{bigidentity}};
\draw [->>] (1) to node [right] {$R$} (2);
\draw [->] (1) to node [above] {$\lambda$} (4);
\draw [->] (2) to node [below] {$\rho$} (3);
\draw [->>] (3) to node [left] {$R$} (4);
\end{tz}
\\
\nonumber
\begin{tz}[xscale=1.3, scale=1.5]
\node [inner sep=1pt] (1) at (0,0) {\tikzpng[scale=\myscale]{rho_source}};
\node [inner sep=1pt] (2) at (0,-1.5) {\tikzpng[scale=\myscale]{unitunfurl2}};
\node [inner sep=1pt] (3) at (2,-1.5) {\tikzpng[scale=\myscale]{snake}};
\node [inner sep=1pt] (4) at (2,0) {\tikzpng[scale=\myscale]{bigidentity}};
\draw [->>] (1) to node [right] {$R$} (2);
\draw [->] (1) to node [above] {$\rho$} (4);
\draw [->] (2) to node [below] {$\lambda$} (3);
\draw [->>] (3) to node [left] {$R$} (4);
\end{tz}
&
\begin{tz}[xscale=1.3, scale=1.5]
\node [inner sep=1pt] (1) at (0,0) {\tikzpng[scale=\myscale]{rho_source}};
\node [inner sep=1pt] (2) at (0,-1.5) {\tikzpng[xscale=-1, scale=\myscale]{unitunfurl1}};
\node [inner sep=1pt] (3) at (2,-1.5) {\tikzpng[scale=\myscale]{snake}};
\node [inner sep=1pt] (4) at (2,0) {\tikzpng[scale=\myscale]{bigidentity}};
\draw [->>] (1) to node [right] {$R$} (2);
\draw [->] (1) to node [above] {$\rho$} (4);
\draw [->] (2) to node [below] {$\lambda$} (3);
\draw [->>] (3) to node [left] {$R$} (4);
\end{tz}
\end{calign}
See \glob, \textit{``Pf: Left Alpha''}, \textit{``Pf: Right Alpha''}, \textit{``Pf: Right Lambda''}, \textit{``Pf: Left Lambda''}, \textit{``Pf: Right Rho''} and \textit{``Pf: Left Rho''}. These equations imply a partner set of equations where $\alpha^\inv$, $\lambda^\inv$ or $\rho^\inv$ appear at the top of each diagram.

These equations tell us that whenever an pseudomonoid 2\-cell acts on a diagram, we can replace it by a sequence where we twist the diagram locally, apply a pseudomonoid 2\-cell (perhaps not identical to the original one), and then twist back. A simple case analysis shows that this is sufficient to rewrite $P$ such that any pseudomonoid 2\-cell is acting on a locally untwisted part of the diagram. Here we analyze the possible ways that the equations above can be used to locally modify the source of an $\alpha$, $\alpha^\inv$, $\lambda$, $\lambda^\inv$, $\rho$ or $\rho^\inv$ generator to reduce local twistedness, indicating in each picture the path to the single output wire with a dot, and giving the local twistedness as a label adjacent to each vertex:
\def\extrascale{0.6}
\begin{gather}
\scriptsize
\begin{tz}[scale=1.27,scale=\extrascale]
\node [anchor=south west, inner sep=0pt, scale=\extrascale] at (0,0) {\includegraphics[scale=0.5]{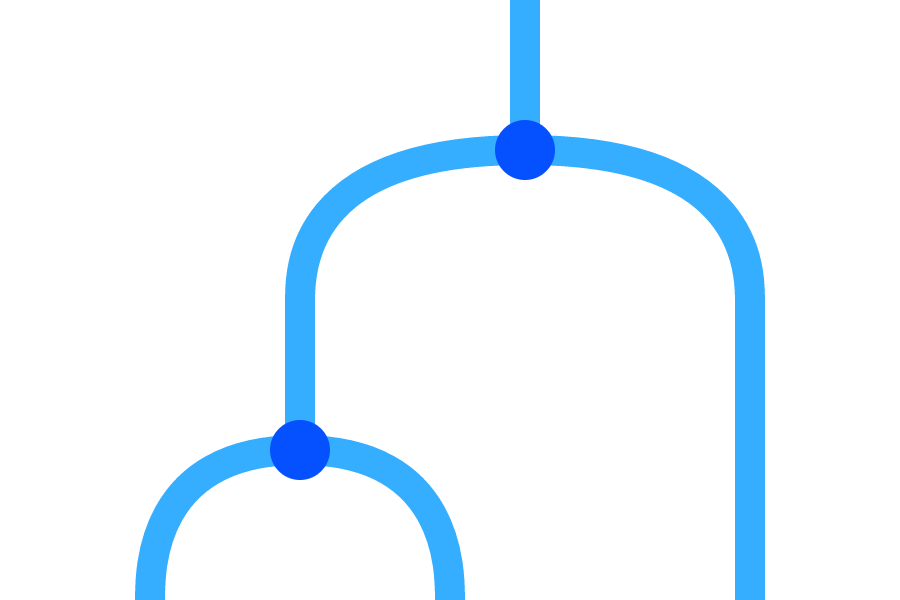}};
\path [use as bounding box] (current bounding box.south west) rectangle (current bounding box.north east);
\node [anchor=south east] at (1,0.5) {$n$};
\node [anchor=south east] at (1.75,1.5) {$n$};
\node [circle, draw=none, fill=lightblue, inner sep=1pt] at (1.75,2.15) {};
\end{tz}
\leadsto
\begin{tz}[scale=1.27, scale=\extrascale]
\node [anchor=south east, inner sep=0pt, xscale=-1, scale=\extrascale] at (0,0) {\includegraphics[scale=0.5]{img/untwistalpha1}};
\path [use as bounding box] (current bounding box.south west) rectangle (current bounding box.north east);
\node [anchor=south west] at (2,0.5) {$n{-}1$};
\node [anchor=south west] at (1.25,1.5) {$n{-}1$};
\node [circle, draw=none, fill=lightblue, inner sep=1pt] at (2.5,-0.15) {};
\end{tz}
\leadsto
\begin{tz}[scale=1.27,scale=\extrascale]
\node [anchor=south west, inner sep=0pt, scale=\extrascale] at (0,0) {\includegraphics[scale=0.5]{img/untwistalpha1}};
\path [use as bounding box] (current bounding box.south west) rectangle (current bounding box.north east);
\node [anchor=south east] at (1,0.5) {$n{-}1$};
\node [anchor=south east] at (1.75,1.5) {$n{-}1$};
\node [circle, draw=none, fill=lightblue, inner sep=1pt] at (1.5,-0.15) {};
\end{tz}
\leadsto
\begin{tz}[scale=1.27, scale=\extrascale]
\node [anchor=south east, inner sep=0pt, xscale=-1, scale=\extrascale] at (0,0) {\includegraphics[scale=0.5]{img/untwistalpha1}};
\path [use as bounding box] (current bounding box.south west) rectangle (current bounding box.north east);
\node [anchor=south west] at (2,0.5) {$n{-}2$};
\node [anchor=south west] at (1.25,1.5) {$n{-}1$};
\node [circle, draw=none, fill=lightblue, inner sep=1pt] at (0.5,-0.15) {};
\end{tz}
\leadsto
\begin{tz}[scale=1.27,scale=\extrascale]
\node [anchor=south west, inner sep=0pt, scale=\extrascale] at (0,0) {\includegraphics[scale=0.5]{img/untwistalpha1}};
\path [use as bounding box] (current bounding box.south west) rectangle (current bounding box.north east);
\node [anchor=south east] at (1,0.5) {$n{-}2$};
\node [anchor=south east] at (1.75,1.5) {$n{-}2$};
\node [circle, draw=none, fill=lightblue, inner sep=1pt] at (1.75,2.15) {};
\end{tz}
\hspace{-1pt}
\\[10pt]
\scriptsize
\begin{tz}[scale=1.27, scale=\extrascale]
\node [anchor=south east, inner sep=0pt, xscale=-1, scale=\extrascale] at (0,0) {\includegraphics[scale=0.5]{img/untwistalpha1}};
\path [use as bounding box] (current bounding box.south west) rectangle (current bounding box.north east);
\node [anchor=south west] at (2,0.5) {$n$};
\node [anchor=south west] at (1.25,1.5) {$n$};
\node [circle, draw=none, fill=lightblue, inner sep=1pt] at (1.25,2.15) {};
\end{tz}
\leadsto
\begin{tz}[scale=1.27,scale=\extrascale]
\node [anchor=south west, inner sep=0pt, scale=\extrascale] at (0,0) {\includegraphics[scale=0.5]{img/untwistalpha1}};
\path [use as bounding box] (current bounding box.south west) rectangle (current bounding box.north east);
\node [anchor=south east] at (1,0.5) {$n$};
\node [anchor=south east] at (1.75,1.5) {$n{-}1$};
\node [circle, draw=none, fill=lightblue, inner sep=1pt] at (2.5,-0.15) {};
\end{tz}
\leadsto
\begin{tz}[scale=1.27, scale=\extrascale]
\node [anchor=south east, inner sep=0pt, xscale=-1, scale=\extrascale] at (0,0) {\includegraphics[scale=0.5]{img/untwistalpha1}};
\path [use as bounding box] (current bounding box.south west) rectangle (current bounding box.north east);
\node [anchor=south west] at (2,0.5) {$n{-}1$};
\node [anchor=south west] at (1.25,1.5) {$n{-}1$};
\node [circle, draw=none, fill=lightblue, inner sep=1pt] at (1.5,-0.15) {};
\end{tz}
\leadsto
\begin{tz}[scale=1.27,scale=\extrascale]
\node [anchor=south west, inner sep=0pt, scale=\extrascale] at (0,0) {\includegraphics[scale=0.5]{img/untwistalpha1}};
\path [use as bounding box] (current bounding box.south west) rectangle (current bounding box.north east);
\node [anchor=south east] at (1,0.5) {$n{-}1$};
\node [anchor=south east] at (1.75,1.5) {$n{-}1$};
\node [circle, draw=none, fill=lightblue, inner sep=1pt] at (0.5,-0.15) {};
\end{tz}
\leadsto
\begin{tz}[scale=1.27,scale=\extrascale]
\node [anchor=south west, inner sep=0pt, scale=\extrascale] at (0,0) {\includegraphics[scale=0.5]{img/untwistalpha1}};
\path [use as bounding box] (current bounding box.south west) rectangle (current bounding box.north east);
\node [anchor=south east] at (1,0.5) {$n{-}2$};
\node [anchor=south east] at (1.75,1.5) {$n{-}2$};
\node [circle, draw=none, fill=lightblue, inner sep=1pt] at (1.75,2.15) {};
\end{tz}
\hspace{-1pt}
\\[10pt]
\scriptsize
\begin{tz}[scale=1.27,scale=\extrascale]
\node [anchor=south west, inner sep=0pt, scale=\extrascale] at (0,0) {\includegraphics[scale=0.5]{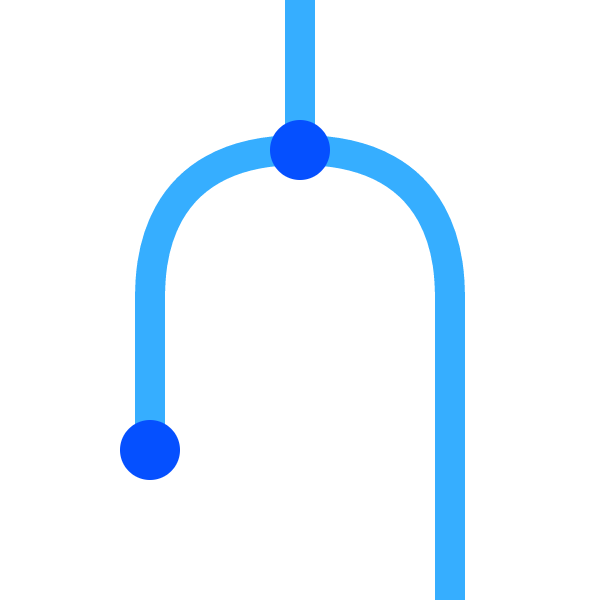}};
\path [use as bounding box] (current bounding box.south west) rectangle (current bounding box.north east);
\node [anchor=south east] at (0.5,0.5) {$n$};
\node [anchor=south east] at (1,1.5) {$n$};
\node [circle, draw=none, fill=lightblue, inner sep=1pt] at (1,2.15) {};
\end{tz}
\,\,\leadsto\,\,
\begin{tz}[scale=1.27,scale=\extrascale]
\node [anchor=south east, inner sep=0pt, scale=\extrascale, xscale=-1] at (0,0) {\includegraphics[scale=0.5]{img/lambda_source}};
\path [use as bounding box] (current bounding box.south west) rectangle (current bounding box.north east);
\node [anchor=south west] at (1.5,0.5) {$n{-}2$};
\node [anchor=south west] at (1,1.5) {$n{-}1$};
\node [circle, draw=none, fill=lightblue, inner sep=1pt] at (.5,-0.15) {};
\end{tz}
\,\,\leadsto\,\,
\begin{tz}[scale=1.27,scale=\extrascale]
\node [anchor=south west, inner sep=0pt, scale=\extrascale] at (0,0) {\includegraphics[scale=0.5]{img/lambda_source}};
\path [use as bounding box] (current bounding box.south west) rectangle (current bounding box.north east);
\node [anchor=south east] at (0.5,0.5) {$n{-}2$};
\node [anchor=south east] at (1,1.5) {$n{-}2$};
\node [circle, draw=none, fill=lightblue, inner sep=1pt] at (1,2.15) {};
\end{tz}
\qquad
\begin{tz}[scale=1.27,scale=\extrascale]
\node [anchor=south west, inner sep=0pt, scale=\extrascale] at (0,0) {\includegraphics[scale=0.5]{img/lambda_source}};
\path [use as bounding box] (current bounding box.south west) rectangle (current bounding box.north east);
\node [anchor=south east] at (0.5,0.5) {$n{+}1$};
\node [anchor=south east] at (1,1.5) {$n$};
\node [circle, draw=none, fill=lightblue, inner sep=1pt] at (1.5,-0.15) {};
\end{tz}
\,\,\leadsto\,\,
\begin{tz}[scale=1.27,scale=\extrascale]
\node [anchor=south east, inner sep=0pt, scale=\extrascale, xscale=-1] at (0,0) {\includegraphics[scale=0.5]{img/lambda_source}};
\path [use as bounding box] (current bounding box.south west) rectangle (current bounding box.north east);
\node [anchor=south west] at (1.5,0.5) {$n{-}1$};
\node [anchor=south west] at (1,1.5) {$n{-}1$};
\node [circle, draw=none, fill=lightblue, inner sep=1pt] at (1,2.15) {};
\end{tz}
\,\,\leadsto\,\,
\begin{tz}[scale=1.27,scale=\extrascale]
\node [anchor=south west, inner sep=0pt, scale=\extrascale] at (0,0) {\includegraphics[scale=0.5]{img/lambda_source}};
\path [use as bounding box] (current bounding box.south west) rectangle (current bounding box.north east);
\node [anchor=south east] at (0.5,0.5) {$n{-}1$};
\node [anchor=south east] at (1,1.5) {$n{-}2$};
\node [circle, draw=none, fill=lightblue, inner sep=1pt] at (1.5,-0.15) {};
\end{tz}
\end{gather}
These modifications involve the left-hand equations given above, and their partners. The twisting can be locally increased in a similar way. Thus, we can ensure that every pseudomonoid 2\-cell acts on a part of the diagram that is locally untwisted.

However, we must show that $P$ is equal to 2\-morphism in which the  \textit{entire} diagram on which the pseudomonoid 2\-cell acts is untwisted. Suppose $\gamma : X \to Y$ is a single pseudomonoid 2\-cell within the composite $P$, acting on a locally untwisted part of the diagram. Then by \autoref{lem:srsnatural}, the following equation holds:
\begin{equation}
\begin{tz}[yscale=0.7]
\node (1) at (0,0) {$X$};
\node (2) at (2,0) {$Y$};
\node (3) at (0,-2) {$\widetilde X$};
\node (4) at (2,-2) {$\widetilde Y$};
\draw [->] (1) to node [above] {$\gamma$} (2) {};
\draw [->] (1) to node [left] {$\Omega_X$} (3) {};
\draw [<-] (2) to node [right] {$\Omega_Y ^\inv$} (4) {};
\draw [->] (3) to node [below] {$\gamma$} (4) {};
\end{tz}
\end{equation}
We use this to replace every pseudomonoid 2\-cell acting on a locally untwisted part of the diagram, with a pseudomonoid 2\-cell acting on a completely untwisted diagram.
\end{proof}

\subsection{Coherence}

We now give the coherence theorems.
\tikzset{box/.style={fill=white, draw=black, line width=1.7pt}}
\tikzset{string/.style={draw=black, line width=1.7pt}}
\begin{proposition}[Coherence for snakeorators]
\label{prop:coherencesnake}
In \free \E, let $P,Q:X \to Y$ be composites of snakeorators, inverse snakeorators and interchangers, such that $X$ is a connected 1\-morphism. Then $P=Q$.
\end{proposition}
\begin{proof}

We begin by reducing the problem to a simpler special case. We define $S:Y \to Y'$ by applying the snake-removal scheme of \autoref{def:srs} to each wire of $Y$, and we define $P' =  S \circ P$ and $Q' = S \circ Q$. Clearly $P'=Q'$ just when $P=Q$. Next, note that $Q'$ is invertible; we define $P'' = P' \circ Q' {}^\inv : Y' \to Y'$, and clearly $P'' = \id_{Y'}$ just when $P'=Q'$. So without loss of generality, we may assume that $X \equiv Y$ and $Q=\id$, and that $X$ has no eliminable cup-cap pairs in the sense of \autoref{def:eliminable}; our task is then to show that $P=\id$.

The composite $P$ is formed from $\sn_1$, $\sn_2$, $\sn_1^\inv$, $\sn_2^\inv$ and interchanger maps. The source of the $\sn_1^\inv$ and $\sn_2^\inv$ maps is the identity, so by naturality, they can all be moved to the beginning of the composite. Similarly, the $\sn_1$ and $\sn_2$ maps can all be moved to the end of the composite. So the rearranged composite is of the following form: a sequence of inverse snakeorators, followed by a sequence of interchangers, followed by a sequence of snakeorators.
\begin{figure}
\tikzset{bluestring/.style={draw=lightblue, line width=2.4pt}}
\tikzset{dottedbluestring/.style={bluestring, dash pattern=on 1pt off 2pt}}
\tikzset{white/.style={draw=white, line width=3.5pt}}
\tikzset{extrascale/.style={scale=0.7}}
\begin{calign}
\nonumber
\begin{tz}[scale=4]
\node (1) [inner sep=0pt] at (0,0) {\begin{tz}[scale=0.89, extrascale]
    \node [inner sep=0pt, anchor=south west] at (0,0) {\tikzpng[scale=3.5, extrascale]{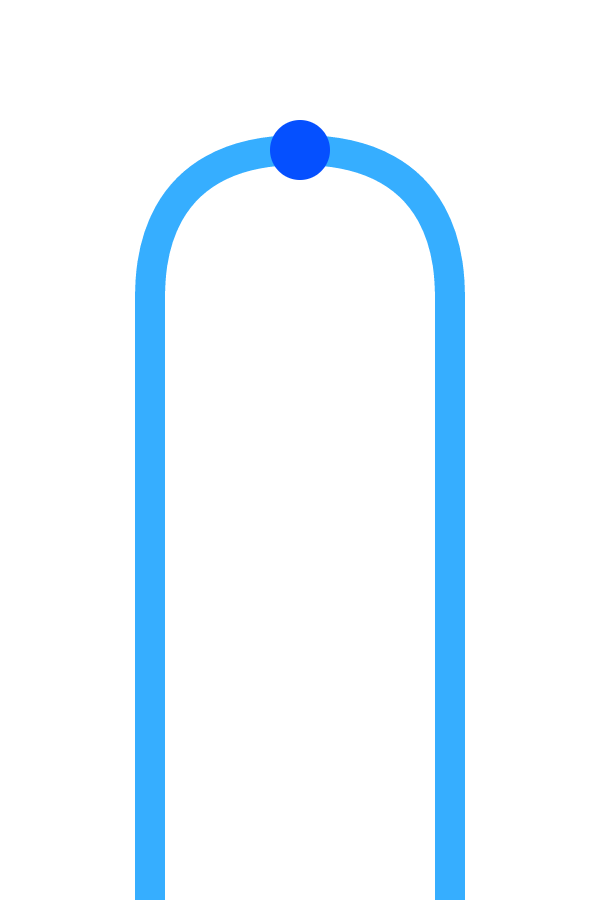}};
    \draw [white] (0.5,1) node (a1) {} to +(0,1) node (a2) {};
    \draw [dottedbluestring] (a1.center) to (a2.center);
\end{tz}};
\node (2) [inner sep=0pt, anchor=west] at ([xshift=0.2cm] 1.east) {\begin{tz}[scale=0.890, extrascale]
    \node [inner sep=0pt, anchor=south west] at (0,0) {\tikzpng[scale=3.5, extrascale]{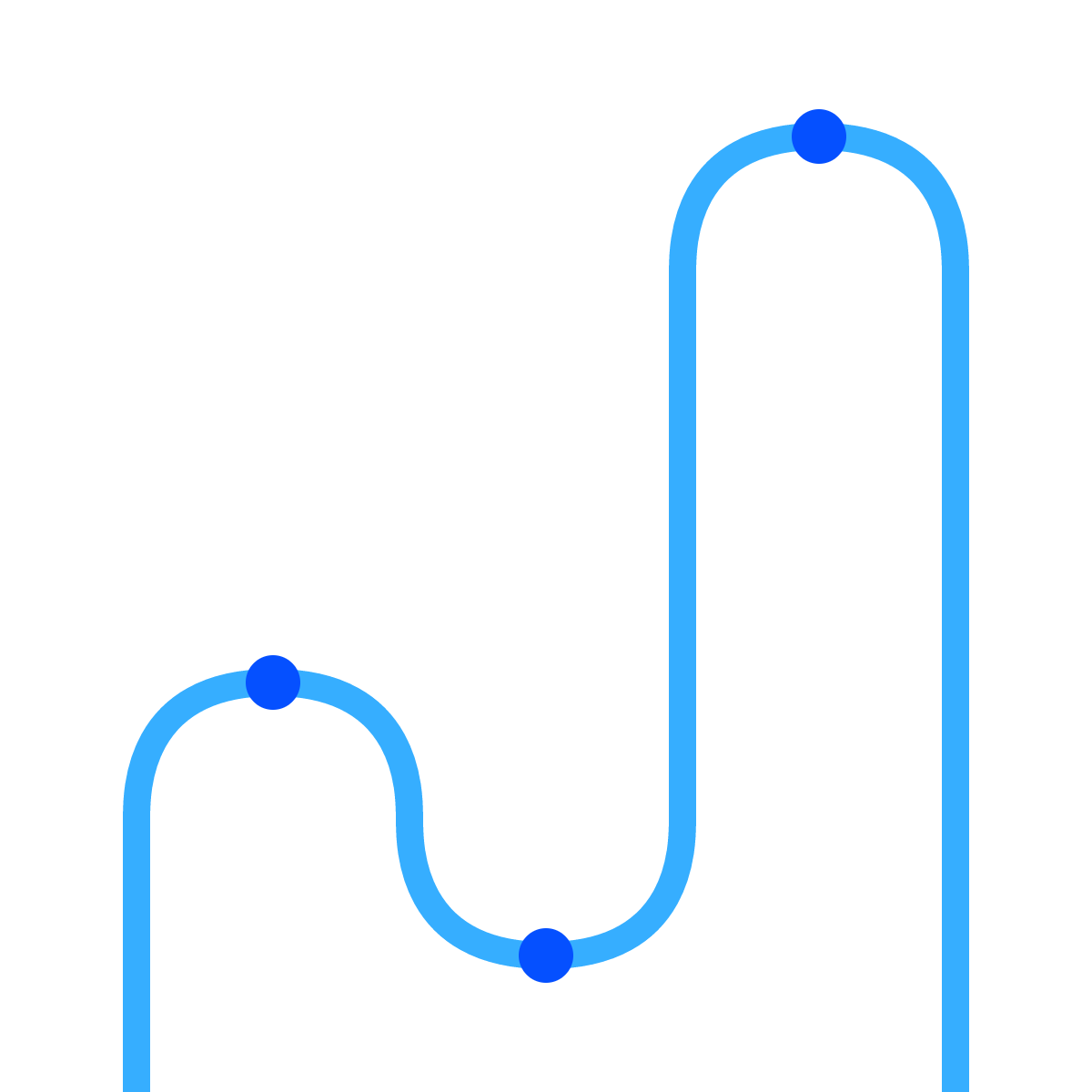}};
    \draw [white] (2.5,1.75) to +(0,1);
    \draw [dottedbluestring] (2.5,1.75) to +(0,1);
\end{tz}};
\node (3) [inner sep=0pt, anchor=west] at ([xshift=0.2cm] 2.east) {\begin{tz}[scale=0.890, xscale=-1, extrascale]
    \node [inner sep=0pt, anchor=south east] at (0,0) {\tikzpng[scale=3.5, xscale=-1, extrascale]{sw2}};
    \draw [white] (2.5,1.75) to +(0,1);
    \draw [dottedbluestring] (2.5,1.75) to +(0,1);
\end{tz}};
\node (5) [inner sep=0pt, anchor=north] at ([yshift=-0.3cm] 3.south) {\begin{tz}[scale=0.890, extrascale]
    \node [inner sep=0pt, anchor=south west] at (0,0) {\tikzpng[scale=3.5, extrascale]{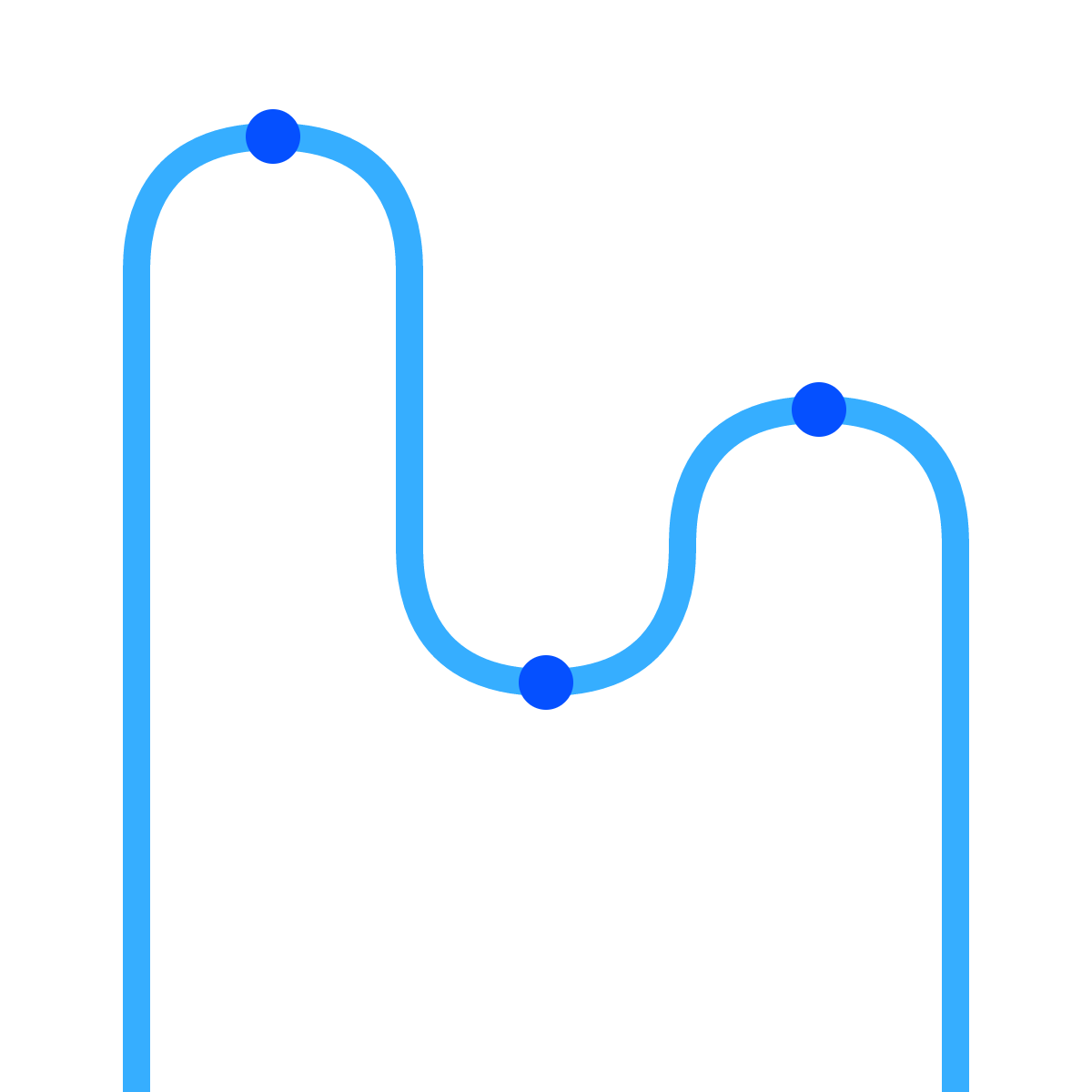}};
    \draw [white] (3.5,0.5) to +(0,1);
    \draw [dottedbluestring] (3.5,0.5) to +(0,1);
\end{tz}};
\node (6) [inner sep=0pt, anchor=west] at ([xshift=0.2cm] 5.east) {\begin{tz}[scale=0.890, extrascale]
    \node [inner sep=0pt, anchor=south west] at (0,0) {\tikzpng[scale=3.5, xscale=-1, extrascale]{sw4}};
    \draw [white] (3.5,0.5) to +(0,1);
    \draw [dottedbluestring] (3.5,0.5) to +(0,1);
\end{tz}};
\node (4) [inner sep=0pt, anchor=west] at ([xshift=0.35cm] 3.east) {\begin{tz}[scale=0.890, extrascale]
    \node [inner sep=0pt, anchor=south west] at (0,0) {\tikzpng[scale=3.5, extrascale]{sw1}};
    \draw [white] (1.5,1) to +(0,1);
    \draw [dottedbluestring] (1.5,1) to +(0,1);
\end{tz}};
\draw [->] (1) to node [above] {$\sigma_1^\inv$} (2);
\draw [->] (2) to node [above] {$\sim$} (3);
\draw [->] (3) to node [above] {$\sigma_2$} (4);
\draw [->, shorten <=8pt, shorten >=1pt] (3) to node [right, pos=0.6] {$\sim$} (5);
\draw [->, shorten <=-15pt] (5) to node [above left, pos=0.4] {$\sigma_2$} (4.south west);
\draw [<-, shorten <=8pt, shorten >=1pt] (4) to node [right, pos=0.6] {$\sigma_1$} (6);
\draw [->] (5) to node [below=3pt] {$\sim$} (6);
\node at (1.95,-0.25) {$(*)$};
\node at (2.15,-0.50) {$(**)$};
\end{tz}
\end{calign}
\caption{Turning $\sigma_2$ into $\sigma_1$ in a general composite\label{fig:BtoA}}
\end{figure}
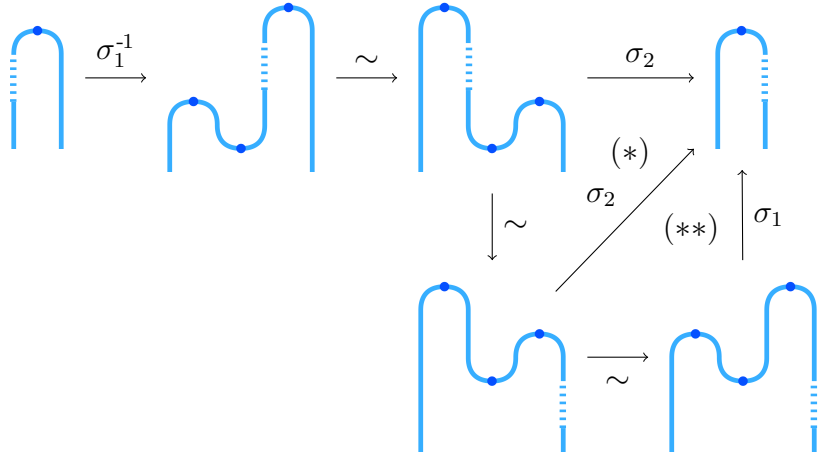

We argue by induction on the number of inverse snakeorator generators. This proof will have two inductive arguments; we call this first one $I_1$. Suppose there are no inverse snakeorators; then we are done. Otherwise, write $F$ for the final inverse snakeorator present in the rearranged composite. Suppose for the rest of the proof that it is a $\sigma_1^\inv$ generator; for a $\sigma_2^\inv$ generator, a similar argument applies.

The final inverse snakeorator $F$ produces a single cup and cap, at least one of which will be annihilated later in the composite, since by assumption $X$ has no eliminable cup-cap pairs. There are two ways that this can happen: either (A) they are annihilated with each other by a $\sigma_1$ generator, or (B) one of them is annihilated with a neighbouring cup or cap by a $\sigma_2$ generator. Suppose case (B) holds; then we will show our composite is equal to one in which (A) in fact holds. The argument is illustrated in \autoref{fig:BtoA}. We begin with some composite that contains a cap. At some point lower down the left-hand wire, a $\sigma_1 ^\inv$ generator acts to create a new cup and cap pair. Later in the composite, following the application of interchangers, and snakeorators to other cups and caps not displayed, a $\sigma_2$ generator is applied to the right-hand cup and cap. (We are supposing that it is the newly-created \textit{cup} which is annihilated first; if the newly-created cap is annihilated first, a reflected argument applies.) In triangle $(*)$, we observe by naturality that the composite is equal to one in which the $\sigma_2$ is applied higher up, adjacent to the leftmost cap. In triangle $(**)$, we apply one of the swallowtail equations to replace the $\sigma_2^\inv$ with a $\sigma_1^\inv$ that annihilates the same cup and cap created in the first arrow of the diagram.

We assume therefore that case (A) holds. Local to the cup and cap produced by $F$, the diagram will look like this:
\begin{equation}
\label{eq:snakeoratorblocked}
\tikzpng[scale=2]{bigidentity}
{\xto{\sigma_1^\inv}}
\tikzpng[scale=2, xscale=-1]{snake}
{\xto[->>]{\sim, \sigma_1, \sigma_2}}
\tikzpng[scale=2, xscale=-1]{snake}
{\xto{\sigma_1}}
\tikzpng[scale=2]{bigidentity}
\end{equation}
The central arrow represents some composite of interchangers, $\sigma_1$ and $\sigma_2$ generators. By assumption, the $\sigma_1$ and $\sigma_2$ maps will act only on the rest of the diagram, not the displayed cups and caps, but the interchangers may affect any part of the diagram. Since this central arrow is not the identity, we cannot immediately apply the invertibility equation to cancel the $\sigma_1^\inv, \sigma_1$ pair from the beginning and end of the composite.

During the course of this central arrow, the chosen cup and cap may become separated vertically, with the cap always remaining above the cup.\footnote{The cap could only go below the cup if an additional inverse snakeorator $\sigma_2^\inv$ was applied to the central connecting wire, but by assumption $\sigma$ is the final inverse snakeorator present in the composite.} Since the initial and final separation is clearly 0, this separation must have some maximum value $n \in \N$ over the course of the central arrow. We now show that our composite is equal to one for which  this maximum value is 0. We give an inductive argument $I_2$, jointly on the value $n$ of this maximum, and its multiplicity $m$ over the course of the central arrow (since the separation may reach its maximum value of $n$ more than once over the composite.)

There are 4 ways that the separation between the cup and cap can increase, by applying interchangers as follows, where the boxes $A$ and $B$ indicate arbitrary composites, the box $C$ is a composite of height 1 \ (i.e. a generator padded on the left and right), and the black wires are arbitrary types:
\begin{calign}
\label{eq:increasesep1}
\begin{tz}[scale=0.63]
\node [inner sep=0pt, anchor=south west] at (0,0) {\tikzpng[scale=2.5]{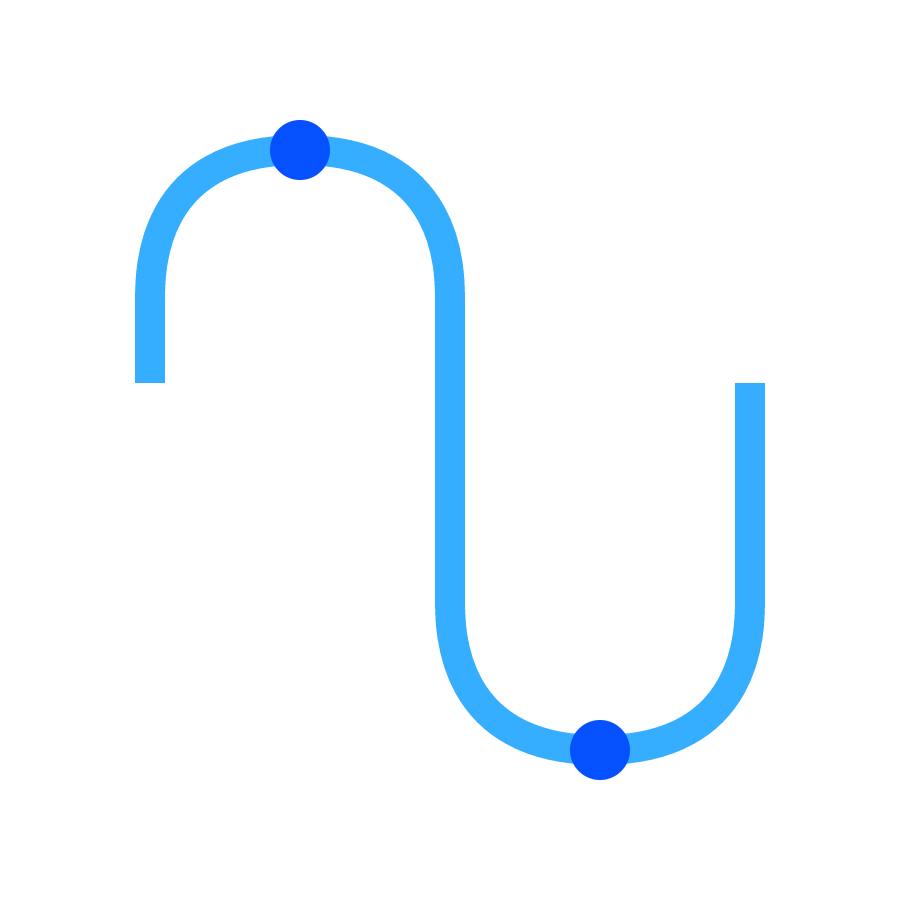}};
\draw [string] (-0.5,2) to +(0,2.25);
\draw [string] (3,2) to +(0,2.25);
\draw [string] (0,0) to +(0,1);
\draw [string] (3.5,0) to +(0,1);
\draw [box] (-1,1) rectangle +(2,1);
\draw [box] (2,1) rectangle +(2,1);
\draw [box] (-1,3) rectangle +(1,0.75);
\node at (0,1.5) {$A$};
\node at (3,1.5) {$B$};
\node at (-0.5,3.375) {$C$};
\end{tz}
\,\,\to\,\,
\begin{tz}[scale=0.63]
\node [inner sep=0pt, anchor=south west] at (0,0) {\tikzpng[scale=2.5]{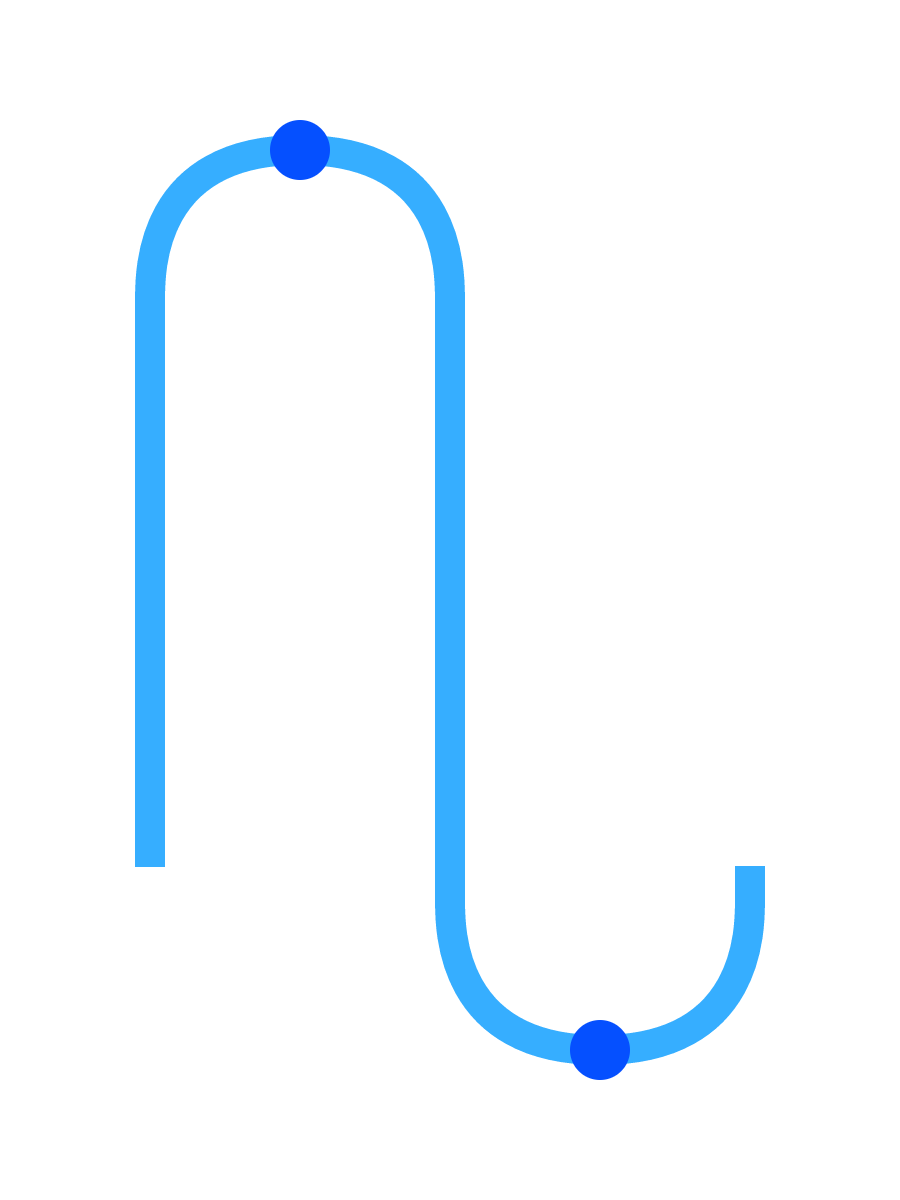}};
\draw [string] (-0.5,2) to +(0,2.25);
\draw [string] (3,2) to +(0,2.25);
\draw [string] (0,0) to +(0,1);
\draw [string] (3.5,0) to +(0,1);
\draw [box] (-1,1) rectangle +(2,1);
\draw [box] (2,1) rectangle +(2,1);
\draw [box] (-1,2.5) rectangle +(1,0.75);
\node at (0,1.5) {$A$};
\node at (3,1.5) {$B$};
\node at (-0.5,2.875) {$C$};
\end{tz}
\\
\begin{tz}[scale=0.63,xscale=-1]
\node [inner sep=0pt, anchor=south east] at (0,0) {\tikzpng[scale=2.5]{snake-s1}};
\draw [string] (0,2) to +(0,2.25);
\draw [string] (3.5,2) to +(0,2.25);
\draw [string] (-.5,0) to +(0,1);
\draw [string] (3.5,0) to +(0,1);
\draw [box] (-1,1) rectangle +(2,1);
\draw [box] (2,1) rectangle +(2,1);
\draw [box] (-0.5,3) rectangle +(1,0.75);
\node at (0,1.5) {$B$};
\node at (3,1.5) {$A$};
\node at (0,3.375) {$C$};
\end{tz}
\,\,\to\,\,
\begin{tz}[scale=0.63, xscale=-1]
\node [inner sep=0pt, anchor=south east] at (0,0) {\tikzpng[scale=2.5]{snake-s2-ls}};
\draw [string] (0,2) to +(0,2.25);
\draw [string] (3.5,2) to +(0,2.25);
\draw [string] (-0.5,0) to +(0,1);
\draw [string] (3,0) to +(0,1);
\draw [box] (-1,1) rectangle +(2,1);
\draw [box] (2,1) rectangle +(2,1);
\draw [box] (-0.5,2.5) rectangle +(1,0.75);
\node at (0,1.5) {$B$};
\node at (3,1.5) {$A$};
\node at (0,2.875) {$C$};
\end{tz}
\\
\begin{tz}[scale=0.63, yscale=-1]
\node [inner sep=0pt, anchor=north west] at (0,0) {\tikzpng[scale=2.5]{snake-s1}};
\draw [string] (0,2) to +(0,2.25);
\draw [string] (3.5,2) to +(0,2.25);
\draw [string] (-0.5,0) to +(0,1);
\draw [string] (3,0) to +(0,1);
\draw [box] (-1,1) rectangle +(2,1);
\draw [box] (2,1) rectangle +(2,1);
\draw [box] (-0.5,3) rectangle +(1,0.75);
\node at (0,1.5) {$A$};
\node at (3,1.5) {$B$};
\node at (0,3.375) {$C$};
\end{tz}
\,\,\to\,\,
\begin{tz}[scale=0.63, yscale=-1]
\node [inner sep=0pt, anchor=north west] at (0,0) {\tikzpng[scale=-2.5]{snake-s2-ls}};
\draw [string] (0,2) to +(0,2.25);
\draw [string] (3.5,2) to +(0,2.25);
\draw [string] (-0.5,0) to +(0,1);
\draw [string] (3,0) to +(0,1);
\draw [box] (-1,1) rectangle +(2,1);
\draw [box] (2,1) rectangle +(2,1);
\draw [box] (-0.5,2.5) rectangle +(1,0.75);
\node at (0,1.5) {$A$};
\node at (3,1.5) {$B$};
\node at (0,2.875) {$C$};
\end{tz}
\\
\label{eq:increasesep4}
\begin{tz}[scale=0.63, scale=-1]
\node [inner sep=0pt, anchor=north east] at (0,0) {\tikzpng[scale=2.5]{snake-s1}};
\draw [string] (-0.5,2) to +(0,2.25);
\draw [string] (3,2) to +(0,2.25);
\draw [string] (0,0) to +(0,1);
\draw [string] (3.5,0) to +(0,1);
\draw [box] (-1,1) rectangle +(2,1);
\draw [box] (2,1) rectangle +(2,1);
\draw [box] (-1,3) rectangle +(1,0.75);
\node at (0,1.5) {$B$};
\node at (3,1.5) {$A$};
\node at (-0.5,3.375) {$C$};
\end{tz}
\,\,\to\,\,
\begin{tz}[scale=0.63, scale=-1]
\node [inner sep=0pt, anchor=north east] at (0,0) {\tikzpng[scale=-2.5]{snake-s2-ls}};
\draw [string] (-0.5,2) to +(0,2.25);
\draw [string] (3,2) to +(0,2.25);
\draw [string] (0,0) to +(0,1);
\draw [string] (3.5,0) to +(0,1);
\draw [box] (-1,1) rectangle +(2,1);
\draw [box] (2,1) rectangle +(2,1);
\draw [box] (-1,2.5) rectangle +(1,0.75);
\node at (0,1.5) {$B$};
\node at (3,1.5) {$A$};
\node at (-0.5,2.875) {$C$};
\end{tz}
\end{calign}
In each case the separation increases by 1 unit. After this increase, arbitrary interchangers can take place in the diagram, not involving the cup and cap; this will leave the separation invariant. Finally, the separation will decrease, giving the end of the local maximum. It could decrease by the reverse of one of the processes (\ref{eq:increasesep1}--\ref{eq:increasesep4}) (4~ways); or by a snakeorator being applied in a region between the cup and cap in height, on the left or the right of the central blue wire (2 ways).

So in total there are 4 ways for the separation to increase, and 6~ways for it to decrease, giving a total of 24 types of  local maximum. For each case, suppose the separation increases from $n$ to $n+1$, and then decreases from $n+1$ to at most $n$. Then we can show directly that this is equal to a composite in which the separation is at most $n$ throughout. While the number of cases is large, they are all handled straightforwardly and in a similar way.

We analyze one case in detail. Suppose that the separation increases from $n$ to $n+1$ by method $\eqref{eq:increasesep1}$; followed by 2\-morphisms $P,Q$ applied between the cup and cap in height, to the left and right of the central wire respectively; followed by the separation decreasing from $n+1$ to at most $n$ by method \eqref{eq:increasesep4}. (2\-morphisms applied above the cap or below the cup may be moved later in the composite by naturality, and we neglect them.) Then by the argument in \autoref{fig:reduceseparation}, we can perform the operations in a different order, under which the separation is at most $n$. Commutativity of this diagram, and others like it for the different cases, follows from naturality of the interchanger in a monoidal bicategory.
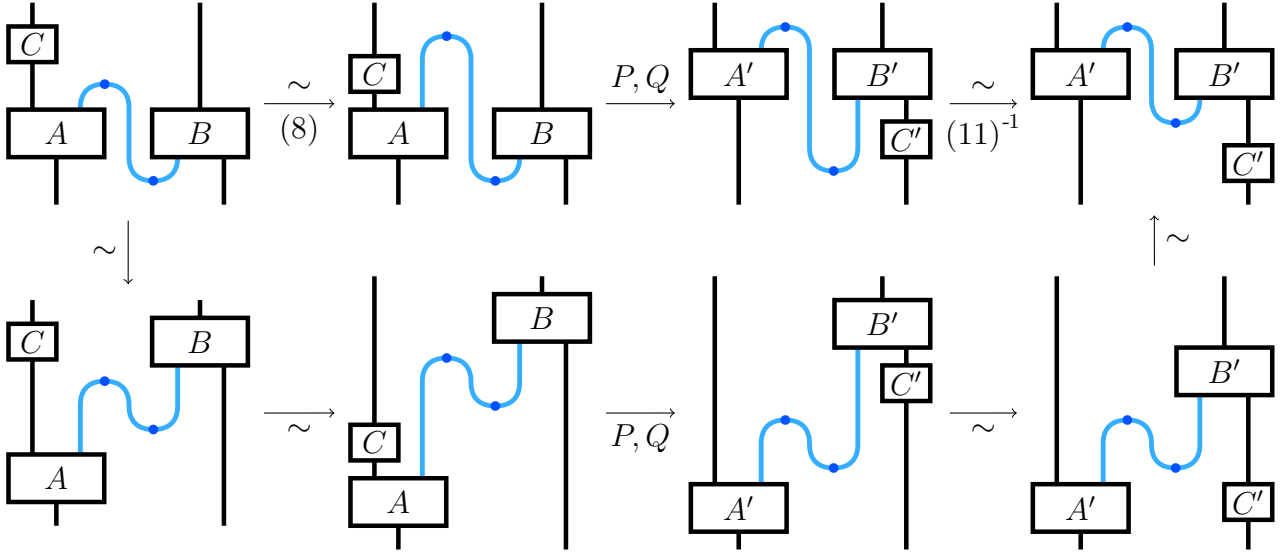
\begin{figure*}
\def\boxscale{0.9}
\[
\hspace{-2cm}
\begin{tz}[xscale=4.3, yscale=4.1, scale=0.95]
\node (1) [inner sep=5pt, scale=\boxscale] at (0,0) {\begin{tz}[scale=0.63]
\node [inner sep=0pt, anchor=south west] at (0,0) {\tikzpng[scale=2.5]{snake-s1}};
\draw [string] (-0.5,2) to +(0,2.25);
\draw [string] (3,2) to +(0,2.25);
\draw [string] (0,0) to +(0,1);
\draw [string] (3.5,0) to +(0,1);
\draw [box] (-1,1) rectangle +(2,1);
\draw [box] (2,1) rectangle +(2,1);
\draw [box] (-1,3) rectangle +(1,0.75);
\node at (0,1.5) {$A$};
\node at (3,1.5) {$B$};
\node at (-0.5,3.375) {$C$};
\end{tz}};
\node (2) [inner sep=5pt, scale=\boxscale] at (1,0) {\begin{tz}[scale=0.63]
\node [inner sep=0pt, anchor=south west] at (0,0) {\tikzpng[scale=2.5]{snake-s2-ls}};
\draw [string] (-0.5,2) to +(0,2.25);
\draw [string] (3,2) to +(0,2.25);
\draw [string] (0,0) to +(0,1);
\draw [string] (3.5,0) to +(0,1);
\draw [box] (-1,1) rectangle +(2,1);
\draw [box] (2,1) rectangle +(2,1);
\draw [box] (-1,2.375) rectangle +(1,0.75);
\node at (0,1.5) {$A$};
\node at (3,1.5) {$B$};
\node at (-0.5,2.75) {$C$};
\end{tz}};
\node (3) [inner sep=5pt, scale=\boxscale] at (2,0) {\begin{tz}[scale=0.63, scale=-1]
\node [inner sep=0pt, anchor=north east] at (0,0) {\tikzpng[scale=-2.5]{snake-s2-ls}};
\draw [string] (-0.5,2) to +(0,2.25);
\draw [string] (3,2) to +(0,2.25);
\draw [string] (0,0) to +(0,1);
\draw [string] (3.5,0) to +(0,1);
\draw [box] (-1,1) rectangle +(2,1);
\draw [box] (2,1) rectangle +(2,1);
\draw [box] (-1,2.5) rectangle +(1,0.75);
\node at (0,1.5) {$B'$};
\node at (3,1.5) {$A'$};
\node at (-0.5,2.875) {$C'$};
\end{tz}
};
\node (4) [inner sep=5pt, scale=\boxscale] at (3,0) {\begin{tz}[scale=0.63, scale=-1]
\node [inner sep=0pt, anchor=north east] at (0,0) {\tikzpng[scale=2.5]{snake-s1}};
\draw [string] (-0.5,2) to +(0,2.25);
\draw [string] (3,2) to +(0,2.25);
\draw [string] (0,0) to +(0,1);
\draw [string] (3.5,0) to +(0,1);
\draw [box] (-1,1) rectangle +(2,1);
\draw [box] (2,1) rectangle +(2,1);
\draw [box] (-1,3) rectangle +(1,0.75);
\node at (0,1.5) {$B'$};
\node at (3,1.5) {$A'$};
\node at (-0.5,3.375) {$C'$};
\end{tz}
};
\node (5) [inner sep=5pt, scale=\boxscale] at (0,-1) {\begin{tz}[scale=0.63]
\node [inner sep=0pt, anchor=south west] at (0,0) {\tikzpng[scale=2.5]{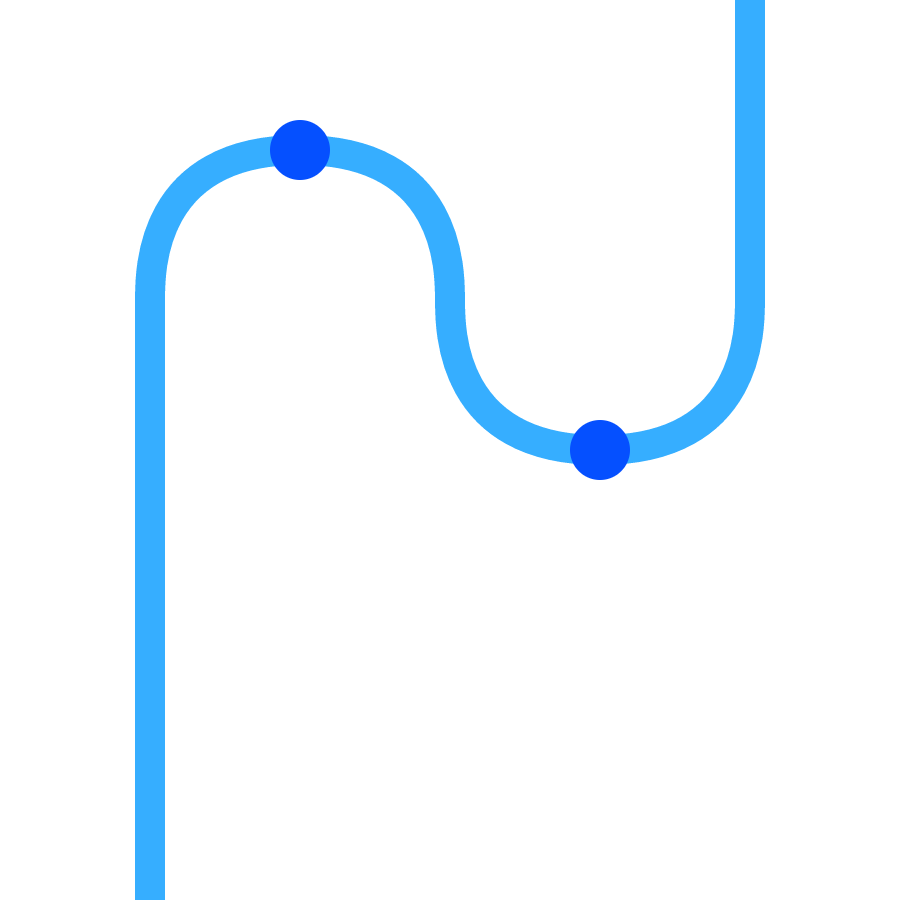}};
\draw [string] (-0.5,1) to +(0,3.25);
\draw [string] (3,3) to +(0,1.25);
\draw [string] (0,-0.5) to +(0,1);
\draw [string] (3.5,-0.5) to +(0,4);
\draw [box] (-1,0) rectangle +(2,1);
\draw [box] (2,2.875) rectangle +(2,1);
\draw [box] (-1,3) rectangle +(1,0.75);
\node at (0,0.5) {$A$};
\node at (3,3.375) {$B$};
\node at (-0.5,3.375) {$C$};
\end{tz}};
\node (6) [inner sep=5pt, scale=\boxscale] at (1,-1) {\begin{tz}[scale=0.63]
\node [inner sep=0pt, anchor=south west] at (0,0) {\tikzpng[scale=2.5]{snake-longl}};
\draw [string] (-0.5,0) to +(0,4.25);
\draw [string] (3,3) to +(0,1.25);
\draw [string] (0,-1.5) to +(0,1);
\draw [string] (3.5,-1.5) to +(0,5);
\draw [box] (-1,-1) rectangle +(2,1);
\draw [box] (2,2.875) rectangle +(2,1);
\draw [box] (-1,0.375) rectangle +(1,0.75);
\node at (0,-0.5) {$A$};
\node at (3,3.375) {$B$};
\node at (-0.5,0.75) {$C$};
\end{tz}};
\node (7) [inner sep=5pt, scale=\boxscale] at (2,-1) {\begin{tz}[scale=0.63, xscale=-1, yscale=-1]
\node [inner sep=0pt, anchor=north east] at (0,0) {\tikzpng[scale=2.5, scale=-1]{snake-longl}};
\draw [string] (-0.5,0) to +(0,4.25);
\draw [string] (3,3) to +(0,1.25);
\draw [string] (0,-1.5) to +(0,1);
\draw [string] (3.5,-1.5) to +(0,5);
\draw [box] (-1,-1) rectangle +(2,1);
\draw [box] (2,2.875) rectangle +(2,1);
\draw [box] (-1,0.375) rectangle +(1,0.75);
\node at (0,-0.5) {$B'$};
\node at (3,3.375) {$A'$};
\node at (-0.5,0.75) {$C'$};
\end{tz}};
\node (8) [inner sep=5pt, scale=\boxscale] at (3,-1) {\begin{tz}[scale=0.63, xscale=-1, yscale=-1]
\node [inner sep=0pt, anchor=north east] at (0,0) {\tikzpng[scale=2.5, scale=-1]{snake-longl}};
\draw [string] (-0.5,0) to +(0,4.25);
\draw [string] (3,3) to +(0,1.25);
\draw [string] (0,-1.5) to +(0,2);
\draw [string] (3.5,-1.5) to +(0,5);
\draw [box] (-1,0) rectangle +(2,1);
\draw [box] (2,2.875) rectangle +(2,1);
\draw [box] (-1,2.875) rectangle +(1,0.75);
\node at (0,0.5) {$B'$};
\node at (3,3.375) {$A'$};
\node at (-0.5,3.25) {$C'$};
\end{tz}};
\draw [->] (1) to node [below] {\eqref{eq:increasesep1}} node [above] {$\sim$} (2);
\draw [->] (2) to node [above] {$P,Q$} (3);
\draw [->] (3) to node [below] {$\eqref{eq:increasesep4} ^\inv$} node [above] {$\sim$} (4);
\draw [->] (1) to node [left] {$\sim$} (5);
\draw [->] (5) to node [below] {$\sim$} (6);
\draw [->] (6) to node [below] {$P,Q$} (7);
\draw [->] (7) to node [below] {$\sim$} (8);
\draw [->, shorten <=-2pt, shorten >=-2pt] (8) to node [right] {$\sim$} (4);
\end{tz}
\hspace{-2cm}
\]
\caption{\label{fig:reduceseparation} Reducing the local maximum separation between a cup and cap}
\end{figure*}

By this argument we have reduced the multiplicity of local maxima of the separation by 1; or, if the multiplicity was already 1, we have reduced the global maximum. By the hypothesis of the inductive argument $I_2$, we are done.

We have demonstrated that the composite in \eqref{eq:snakeoratorblocked} is equal to one in which we apply $\sigma_1 ^\inv$, act only on the cups and caps by interchanging them as a single unit, and then apply $\sigma_1$. By naturality of the interchanger, the $\sigma_1$ and $\sigma_1^\inv$ can now be cancelled. Thus we have shown that our composite is equal to one with strictly fewer inverse snakeorators. By the hypothesis of the inductive argument $I_1$, we are done.
\end{proof}

\begin{proposition}[Coherence for rotational 2\-morphisms]
\label{prop:coherencerotations}
In \free \E, let $X$ be a simple 1\-morphism, and let \mbox{$P,Q : X \stackrel R \twoheadrightarrow Y$} be rotational 2\-morphisms. Then $P=Q$.
\end{proposition}
\begin{proof}
The rotational 2\-morphisms are invertible, so $P=Q$ just when $Q^\inv \circ P = \id_X$; so it is enough to consider the case that $Y=X$ and $Q=\id_X$. Also, by \autoref{lem:straighten}, there is a rotational isomorphism $\Omega_X : X \to \widetilde X$, where $\widetilde X$ is in pseudomonoid form; so without loss of generality we may suppose that $X$ is in pseudomonoid form, and therefore untwisted by \autoref{lem:pseudomonoiduntwisted}.

Let $v$ be any $m$ vertex in $X$. Since $R_m$ and $L_m$ maps act locally on multiplication vertices, and no rotational cells can introduce or eliminate $m$ vertices, we can move all $R_m$ and $L_m$ instances acting on $v$ to the beginning of the composite $P$. Since the source and target of $P$ is untwisted, an equal number of $R_m$ and $L_m$ 2\-morphisms must act on it, since they change the twistedness locally; we cancel adjacent $(R_m,L_m)$ and $(L_m,R_m)$ pairs using \autoref{lem:inversepairs}, at the cost of introducing additional snakeorators. Similarly, each $u$ vertex in  $X$ will be rotated by a succession of $R_u$, $L_u$, $R_f$ and $L_f$ maps; as for the $m$ vertices, we move these to the beginning of the composite, and eliminate them pairwise. Hence we obtain $P=P': X \to X$, where $P'$ is formed purely of interchangers and snake maps. But then by \autoref{prop:coherencesnake} we conclude $P'=\id$, and hence $P=Q=\id$.
\end{proof}

\begin{lemma}
\label{lem:pseudomonoidinterchanger}
In \free \E, suppose $X,Y$ are 1\-morphisms in pseudomonoid form, and \mbox{$X \stackrel R \twoheadrightarrow Y$}. Then $X \stackrel \sim \twoheadrightarrow Y$.
\end{lemma}
\begin{proof}
Consider the following diagram:
\begin{equation}
\begin{tz}[xscale=2]
\node (1) at (0,0) {$X$};
\node (2) at (1,0) {$Y$};
\node (3) at (0.5,-1) {$\widehat X = \widehat Y$};
\draw [->>] (1) to node [above] {$R$} (2);
\draw [->>] (1) to node [below left=-1pt, pos=0.3] {$\Theta_X$} (3.145);
\draw [<<-] (2) to node [below right=-1pt, pos=0.3] {$\Theta_Y ^\inv$} (3.35);
\end{tz}
\end{equation}
The 2\-morphisms $\Theta_X$ and $\Theta_Y$ are defined by \autoref{lem:leftpseudomonoid}, and are composed purely of interchangers. The equality at the bottom follows from \autoref{lem:leftpseudomonoidunique}. By \autoref{prop:coherencerotations}, the diagram commutes. Since the lower path $X \stackrel \sim \twoheadrightarrow \widehat X = \widehat Y \stackrel \sim \twoheadrightarrow Y$ is composed purely of interchangers, the result follows.
\end{proof}

We now prove our main result.

\begin{customthm}{\ref{thm:frobeniuscoherence}}
(Coherence for Frobenius structures.) \em
Let $P, Q: X \to Y$ be 2\-morphisms in \free \F, such that $X$ is connected and acyclic, with at least one boundary wire. Then $P = Q$.
\end{customthm}
\begin{proof}
Let $X$ be a simple 1\-morphism in \free \F, and let $P, Q : X \to Y$ be 2\-morphisms. Since the 2\-morphisms of \free \F are invertible, then $P=Q$ just when $Q ^\inv \circ P = \id$; so it is enough to consider the case that $X \equiv Y$ and $Q = \id$. We also suppose for now that $X$ and $Y$ have exactly one output wire. We then proceed as follows.
\begin{itemize}
\item Define $P_1 : X \to X$ as the image of $P$ under the embedding $\free \F \to \free \E$. It is enough for us to prove $P_1=\id$, since by \autoref{lem:equivalentpresentations}, the presentations \F and \E are equivalent.
\item Define $P_2$ by taking $P_1$ and eliminating all instances of $\mu$ and $\nu$ using \autoref{lem:killmunu}.
\item Define $P_3$ by taking $P_2$ and ensuring all pseudomonoid 2\-cells act on untwisted diagrams only, using \autoref{prop:untwistalgebraic}.
\item Define $P_4 := \Omega_X \circ P_3 \circ \Omega_X ^\inv: \widetilde X \to \widetilde X$.
\end{itemize}
Clearly $P_1 = P_2 = P_3$, and $P_3 = \id$ just when $P_4 = \id$.

We now consider the structure of $P_4$. Its source and target are in pseudomonoid form, and $P_4$ itself is built from rotational 2\-morphisms, along with pseudomonoid 2\-cells acting on diagrams of pseudomonoid form. That is, $P_4$ is of the following form, where $X_i$ and $Y_i$ are of pseudomonoid form, $\gamma_i$ are pseudomonoid 2\-cells, and $R$ indicates a composite rotational 2\-morphism:
\[
P_4 \quad= \quad X_0 \stackrel R \twoheadrightarrow Y_0 \stackrel {\gamma_0} \to X_1 \stackrel R \twoheadrightarrow Y_1 \stackrel {\gamma_1} \to \cdots \stackrel R \twoheadrightarrow Y_n = X_0
\]
By applying \autoref{lem:pseudomonoidinterchanger}, we can build $P_5$ as follows, with ${P_4=P_5}$:
\[
P_5 \quad= \quad X_0 \stackrel \sim \twoheadrightarrow Y_0 \stackrel {\gamma_0} \to X_1 \stackrel \sim \twoheadrightarrow Y_1 \stackrel {\gamma_1} \to \cdots \stackrel \sim \twoheadrightarrow Y_n = X_0
\]
This composite lies purely in the image of the pseudomonoid presentation \P. But then by the coherence theorem for pseudomonoids we conclude $P_5=\id$; see the work of Lack~\cite{Lack_2000}, and the thesis of Houston~\cite[Section~6]{houston-thesis}.

In the statement of the theorem, we required only that $X$ had at least one input or output, but we assumed for our argument above that X had exactly one output. To reduce to the case of a unique output, one can compose with cups and caps appropriately; the swallowtail equations ensure this gives a bijection of hom-sets.
\end{proof}

\subsection{Adjoints}

We now consider the case that the generators $m,u$ have right adjoints. This is required to obtain the correct relationship to linear logic, as we discuss in the next section.

\begin{definition}
The \emph{right-adjoint Frobenius presentation} $\F^*$ is defined to be the Frobenius presentation \F, with the following additional data:
\begin{itemize}
\item Additional 1\-morphisms $m^*$ and $u^*$:
\begin{calign}
\tikzpng[scale=3]{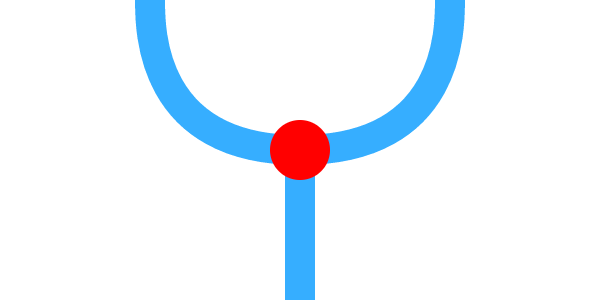}
&
\tikzpng[scale=3]{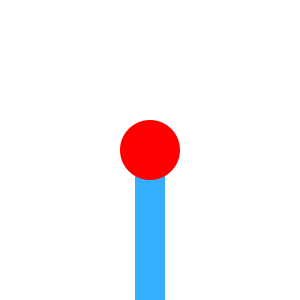}
\end{calign}
\item Additional 2\-morphisms $\eta, \epsilon, \phi, \psi$:
\begin{align}
\tikzpng[scale=2,yscale=2]{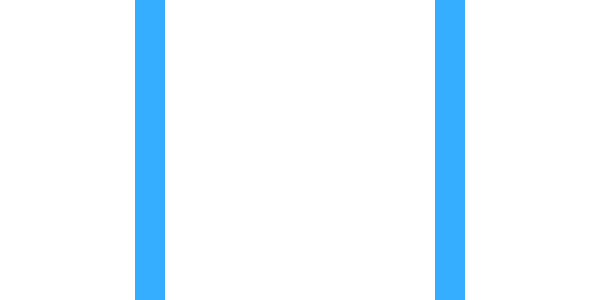}
&{\xto \eta}
\tikzpng[scale=2]{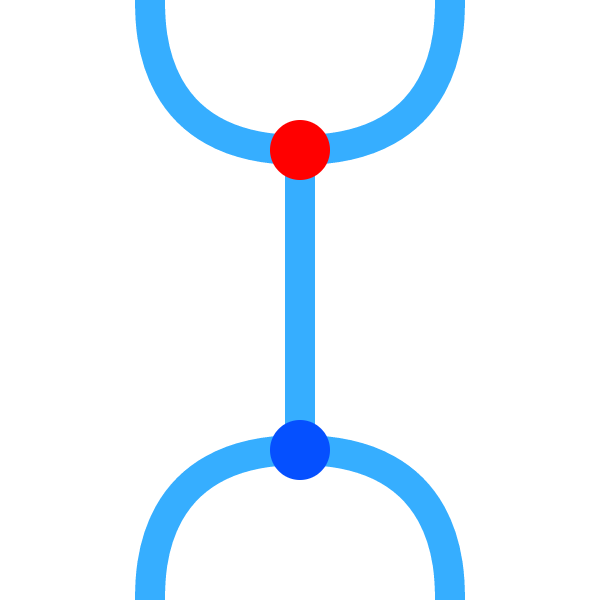}
&
\tikzpng[scale=2]{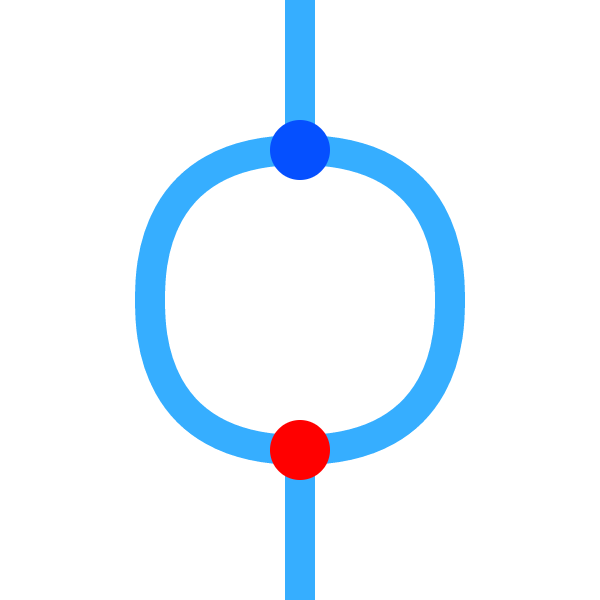}
&{\xto \epsilon}
\tikzpng[scale=2,yscale=2]{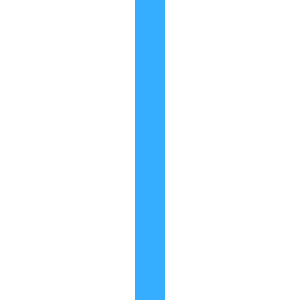}
\\
&{\xto \phi}
\hspace{0.233cm}\tikzpng[scale=2]{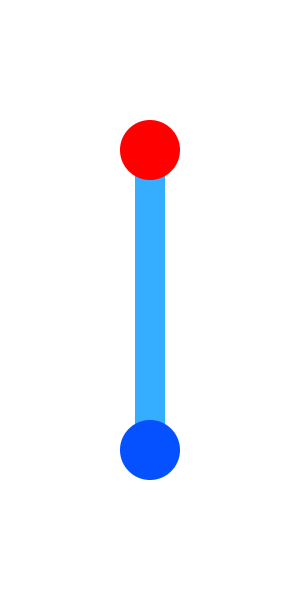}
&
\tikzpng[scale=2]{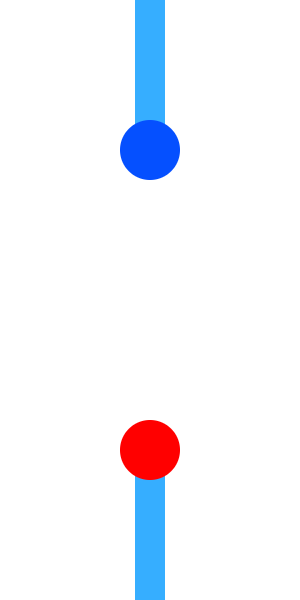}\hspace{0.253cm}
&{\xto \psi}
\hspace{-1pt}\tikzpng[scale=2,yscale=2]{identity}
\end{align}
\item Additional equations, stating that $m \dashv m^*$ and $u \dashv u^*$:
\begin{align}
\id\,\,&=\,\,
\tikzpng[scale=2]{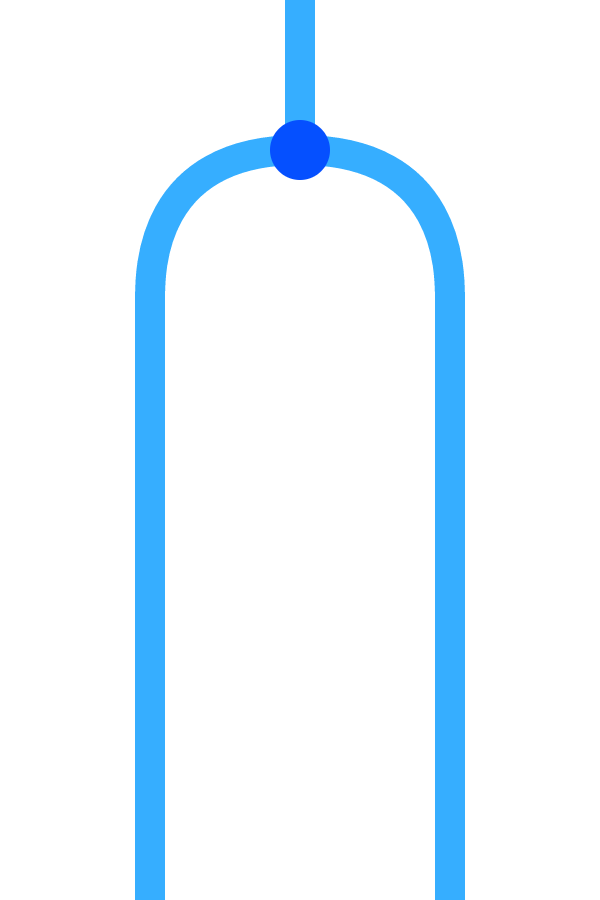}
{\xto {\eta}}
\tikzpng[scale=2]{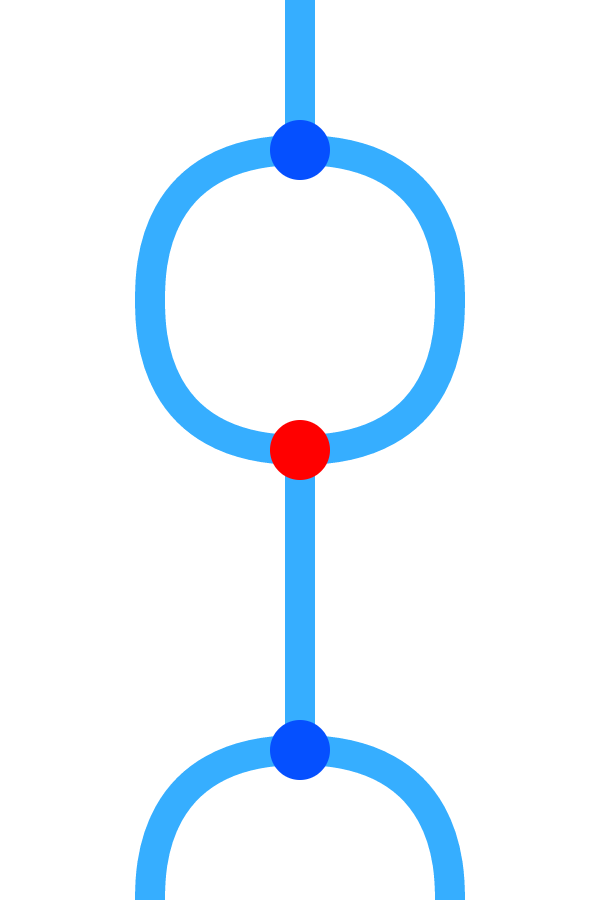}
{\xto \epsilon}
\tikzpng[scale=2]{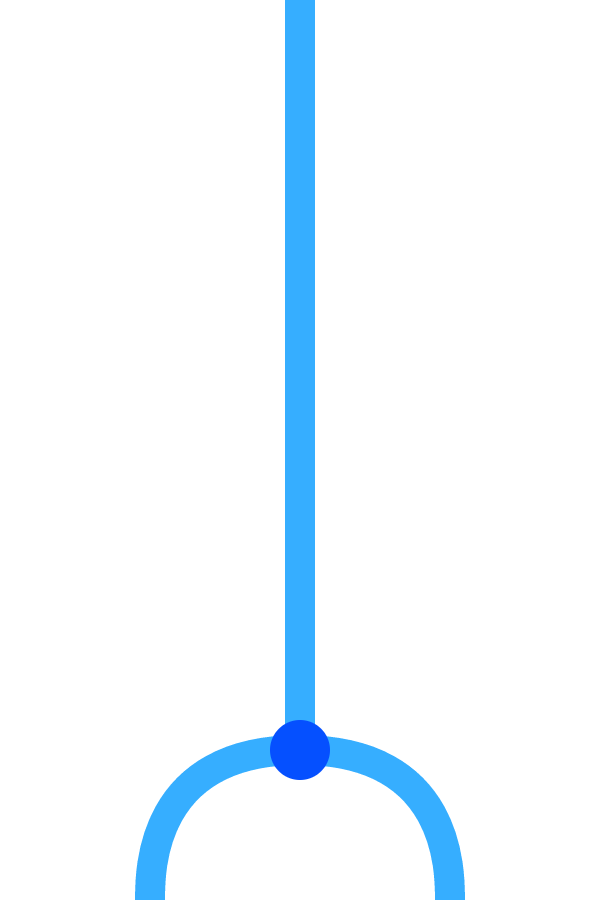}
&
\id\,\,&=\,\,
\tikzpng[scale=2]{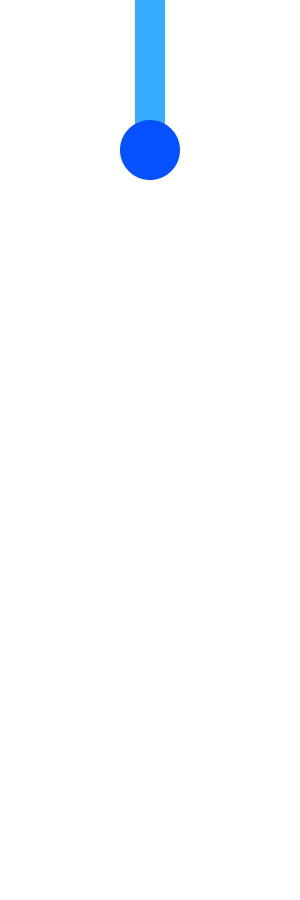}
{\xto {\phi}}
\tikzpng[scale=2]{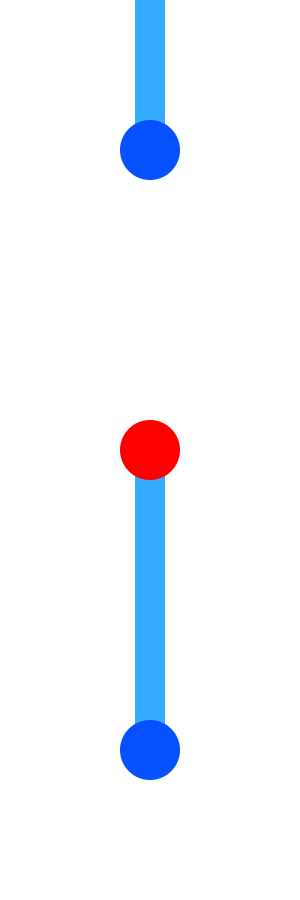}
{\xto \psi}
\tikzpng[scale=2]{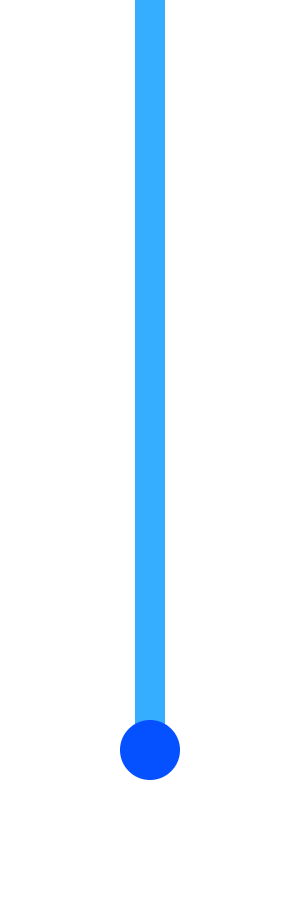}
\\
\id\,\,&=\,\,
\tikzpng[scale=2]{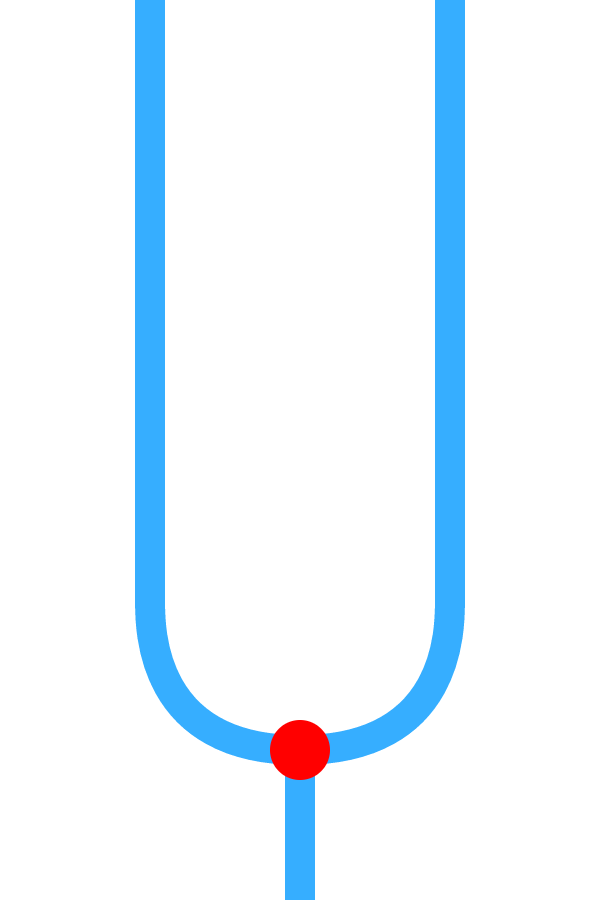}
{\xto {\eta}}
\tikzpng[scale=2]{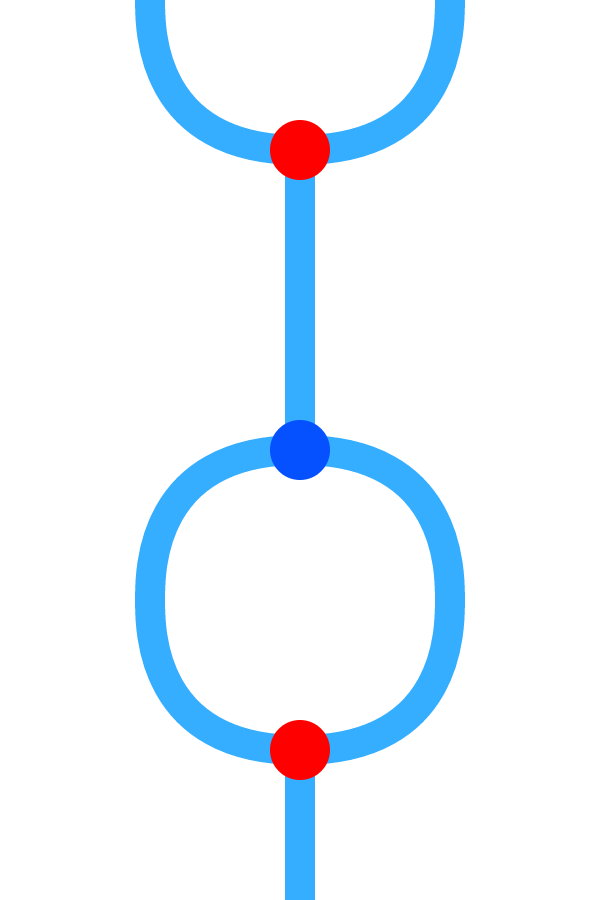}
{\xto \epsilon}
\tikzpng[scale=2]{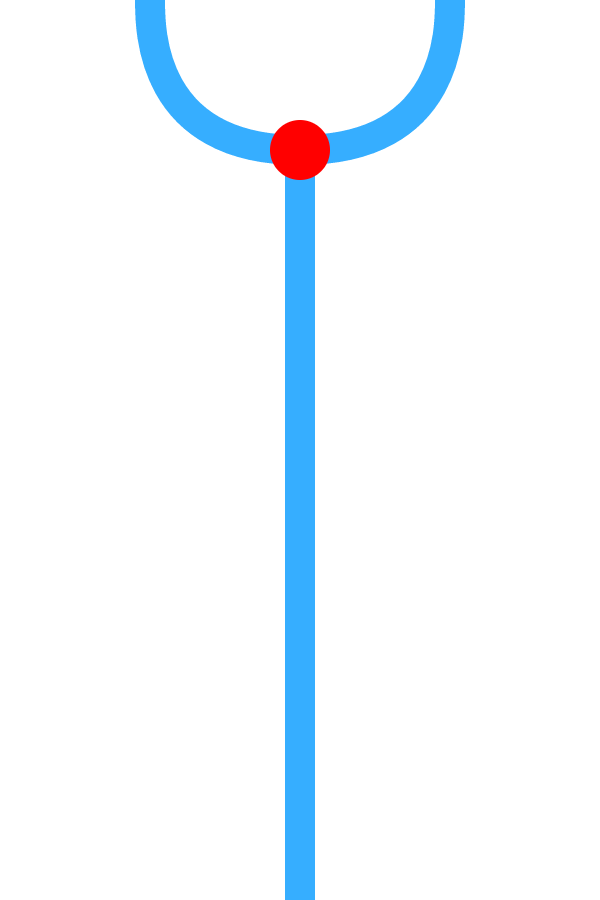}
&
\id\,\,&=\,\,
\tikzpng[scale=2]{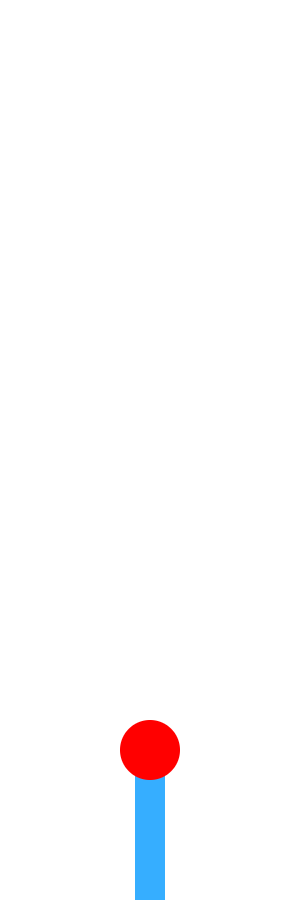}
{\xto {\phi}}
\tikzpng[scale=2]{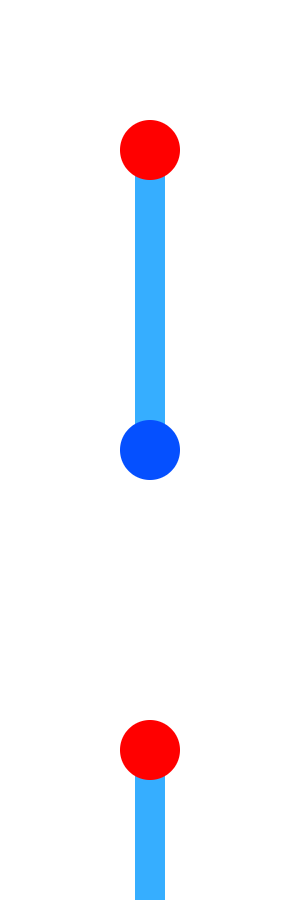}
{\xto \psi}
\tikzpng[scale=2]{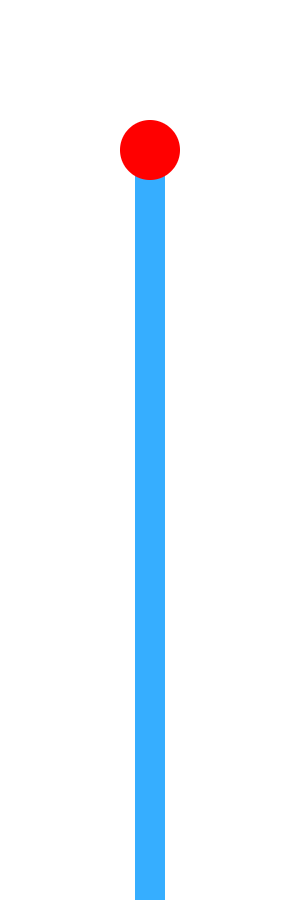}
\end{align}
\end{itemize}
\end{definition}

\begin{definition}
In $\free {\F^*}$, define the following composites:
\begin{align}
\tikzpng[scale=2,yscale=-1]{adjointtensor}
\,\,&:=\,\,
\tikzpng[scale=2,yscale=1, scale=0.8]{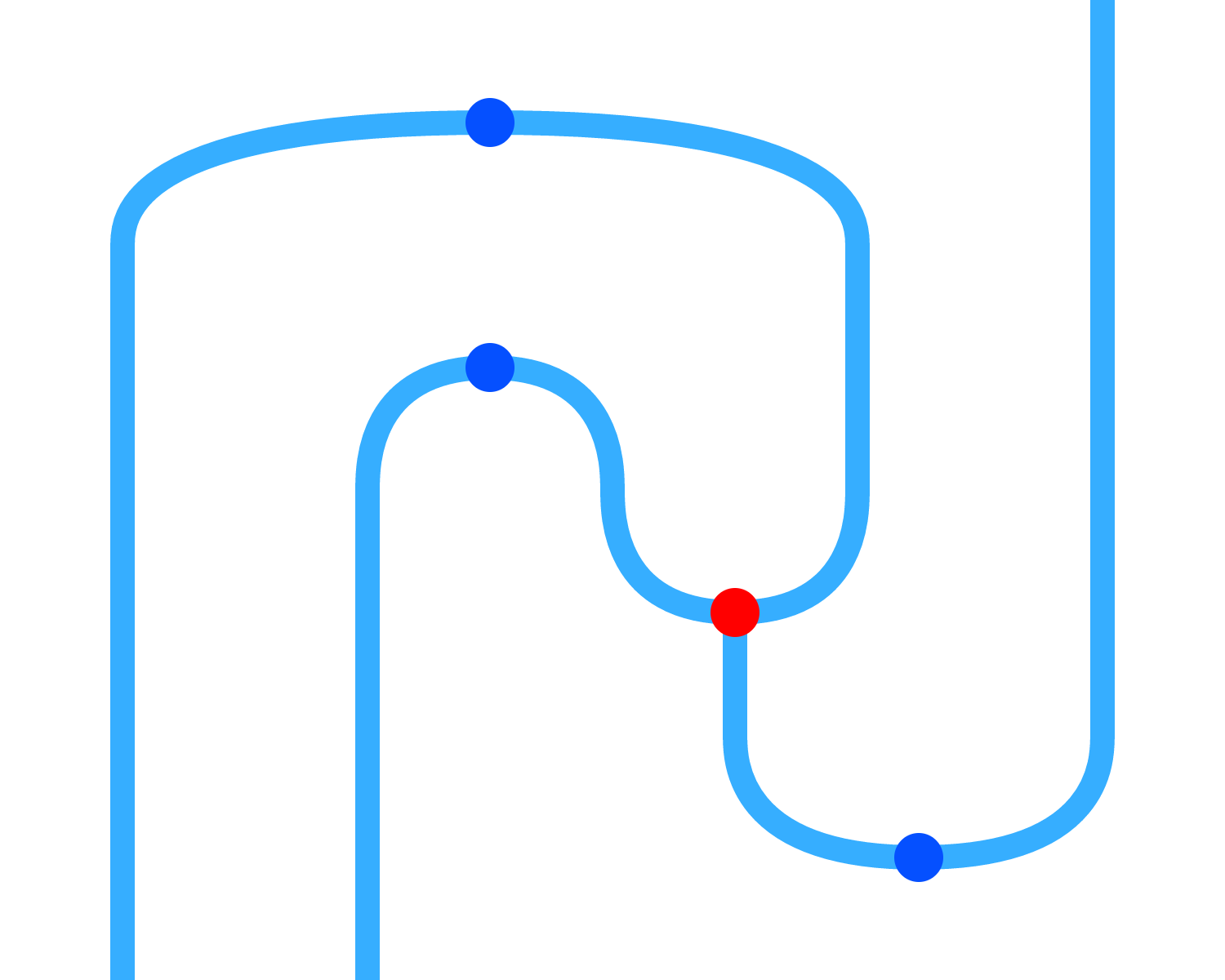}
&
\tikzpng[scale=2,yscale=-1]{adjointunit}
\,\,&:=\,\,
\tikzpng[scale=2,yscale=1]{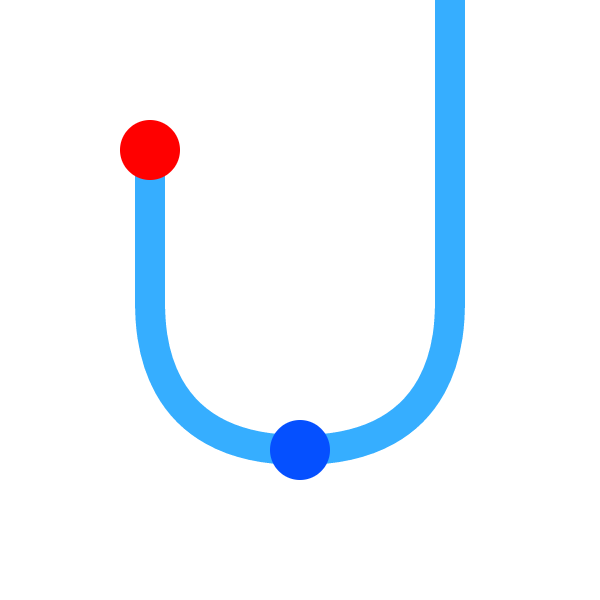}
\end{align}
\end{definition}

\begin{lemma}
\label{lem:additionaladjunctions}
In \free{\F^*}, we have the following adjunctions:
\begin{align*}
\tikzpng[scale=3]{mult}
\,\,&\dashv\,\,
\tikzpng[scale=3]{adjointtensor}
&
\tikzpng[scale=3]{unit}
\,\,&\dashv\,\,
\tikzpng[scale=3]{adjointunit}
\\
\tikzpng[scale=3, yscale=-1]{adjointtensor}
\,\,&\dashv\,\,
\tikzpng[scale=3, yscale=-1]{mult}
&
\tikzpng[scale=3, yscale=-1]{adjointunit}
\,\,&\dashv\,\,
\tikzpng[scale=3, yscale=-1]{unit}
\end{align*}
\end{lemma}
\begin{proof}
The first row of adjunctions are explicit in the definition of $\F^*$. The second row are constructed straightforwardly from the available data.
\end{proof}

\subsection{Graphical calculus for Frobenius pseudomonoids}

As a result of the coherence theorem, we can relax our conventions for the graphical calculus for a Frobenius pseudomonoid. For any number of input and output wires (as long as there is at least one in total), we may draw a generic vertex that connects them:
\[
\tikzpng[scale=3, yscale=-1]{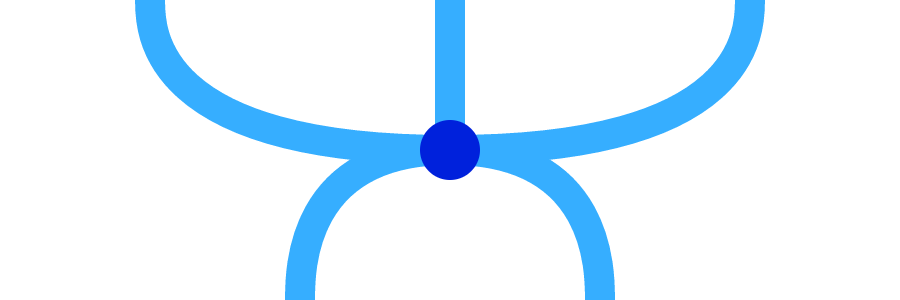}
\]
By the coherence theorem, we are guaranteed that all acyclic, connected ways of forming a composite with this type will be canonically isomorphic, so we may as well draw this simpler vertex to represent the entire isomorphism class. Furthermore, given any two connected acyclic 1\-morphisms with boundary in \free \F, there is a canonical choice of morphism between them, which in the surface calculus we can show as a pointlike operation, such as expression~\eqref{eq:surfaceexample}. In a sense, this is a categorification of the well-known `spider theorem' of Frobenius algebras~\cite{Coecke_2008}.

The graphical calculus of $\free {\F^*}$ involves red nodes, which do not admit a coherence theorem similar to the one we have been studying earlier in this section. Intuitively, while the blue nodes are `fluid' and allow a rich family of deductions, which we can reason about with the coherence theorem, the red nodes are relatively `rigid' and do not share the same topological properties. The coherence theorem can of course still be used for that part of $\free{\F^*}$ which is in the image of the obvious embedding $\free \F \to \free {\F^*}$.

\subsection{$*$-Autonomous categories}

From the work of Street~\cite{Street_2004}, it is known that the following definition of non-symmetric $*$\-autonomous category agrees with that of Barr~\cite{Barr_1995}, up to a Cauchy-completeness assumption. Here $\bicat{Prof}$ is the symmetric monoidal bicategory of categories, profunctors and natural transformations~\cite{borceux}.

\begin{definition}
\label{sa-street}
A \textit{$*$-autonomous category} is a pseudomonoid $(\cat C, m, u, \alpha, \lambda, \rho)$ in \bicat{Prof}, equipped with a morphism $f: \cat C \to \cat 1$ such that $f \circ m : \cat C \times \cat C \to \cat 1$ is a biexact pairing, and such that $m$ and $u$ have right adjoints.
\end{definition}

\begin{proposition}
\label{prop:dataequivalent}
The data of a $*$\-autonomous category is equivalent to that of an $\F^*$\-structure in \bicat{Prof}.
\end{proposition}
\begin{proof}
The reverse direction immediate, since given an $\F^*$\-structure in $\bicat {Prof}$, the composite $f \circ m$ is clearly biexact, thanks to the invertible 2\-morphisms $\mu$ and $\nu$.

For the forward direction, we sketch the main idea. In \autoref{sa-street}, the biexact pairing condition means that if one composes the right leg of $f \circ m:\cat C \times \cat C \to \cat 1$ with the Hom-profunctor of type $\cat 1 \to \cat C \times \cat C ^\op$, the resulting profunctor of type $\cat C \to \cat C ^\op$ is an equivalence. Therefore, this definition is in terms of \textit{two} objects, \cat C and $\cat C ^\op$, and structures defined on them. But since $\cat C$ and $\cat C ^\op$ are equivalent, it stands to reason that we can transport any structures defined on $\cat C^\op$ across to \cat C, by composing with the equivalence.\footnote{This is similar to how, given a monoid on a set $S$ and a bijection $\phi:S \to T$ for some set $T$, we may compose with the bijection and its inverse to obtain a monoid on $T$ in a canonical way.} This yields a definition in terms of just \textit{one} object,  \cat C itself. This is an algorithmic procedure, and the result is an $\F^*$-algebra in \bicat{Prof}.
\end{proof}

\section{3\-dimensional proofs for linear logic}
\label{sec:3dproofs}

\subsection{Overview}

In this section we use the coherence theorem of \autoref{sec:coherence} to develop a theory of surface proofs for multiplicative linear logic with units, and  without the exchange rule. Linear logic is traditionally developed algebraically in terms of \textit{sequent proofs}, which can be represented using geometrical objects called \textit{proof nets}~\cite{blute-coherence, girard}. Our surface proofs share the geometrical character of proof nets, while also sharing some important properties of sequent proofs. We list here some properties of our surface proof calculus.
\begin{enumerate}
\item Correctness is local; any well-typed composite produces a valid proof-theoretic object, with no global property (such as the \textit{long-trip criterion}~\cite{girard}) to  be verified.
\item Equivalence is local, unlike for proof nets where the main dynamical rewiring step involves non-local jumps, and requires the re-validation of a global property~\cite[Section~3.1]{blute-coherence}.
\item Equivalence is broad, establishing some  proof equivalences in fewer steps than for proof nets; sometimes in just one step. (Our coherence result of \autoref{sec:coherence} is critical here.)
\item It works well with the unital fragment of the logic, which is problematic for traditional proof nets, requiring decorations in the form of \textit{thinning links}~\cite{blute-coherence}.
\item It is compositional, in the sense that the surface for a  composite proof is just the union of the surfaces for any partition.
\item It is close to categorical semantics, with a surface giving rise straightforwardly to a morphism in a $*$-autonomous category.
\item The equivalence relation on diagrams is essentially geometrical, and may not be easy to decide in all cases.
\end{enumerate}

We propose that traditional proof nets are essentially 2\-dimensional \textit{projections} of the 3\-dimensional surface geometry. From this perspective, we can make sense of some of the features of proof nets: the long-trip criterion can be interpreted as a non-local check that the 2\-dimensional shadow is consistent with a valid 3\-dimensional geometry, and the thinning link decorations indicate the depth at which a unit is attached in the 3\-dimensional geometry. We illustrate this in \autoref{fig:comparison}, which gives a sequent calculus deduction along with its surface calculus and proof net representations.

We offer the following additional criticisms of our scheme.
\begin{enumerate}
\item When we say in point 4 above that it works well with the unital fragment of the logic, we mean that the notion of proof equivalence is straightforward and local, not that one obtains an efficient method for deciding proof equality~\cite{houston-pspace}.
\item We do not offer a decision procedure for proof equality in terms of our 3\-dimensional notation.
\item We work in a non-symmetric variant of the logic which lacks the exchange rule; adding a symmetry would  mean that the surfaces self-intersect, which would complicate the formalism, although we expect our results could be extended.
\item As complex geometrical objects, our surfaces are not as easy to manipulate with pen-and-paper as traditional proof nets; to become practical, advances in proof assistant technology (such as~\cite{globular}) would be required.
\end{enumerate}
Furthermore, we note that there are many questions which we do not investigate here, such as the decision problem for equivalence of surfaces, and the behaviour of cut elimination, which would be required to for this formalism to yield a fully-fledged logical system. These are questions for future work.

\input{3dgenerators.extra}
\input{3dequations.extra}
\input{3dcompound.extra}
\clearpage


\subsection{The 2\-dimensional calculus}
\label{sec:2dcalculus}
\tikzset{extrascale/.style={scale=0.8}}

In this section we develop the 2\-dimensional string diagram calculus for sequents. This calculus is the Joyal-Street calculus for monoidal categories~\cite{Joyal_1991}, directed from left to right. We use the standard 2-sided sequent calculus for nonsymmetric multiplicative linear logic with units~\cite{blute-coherence}: our sequents are pairs $\Gamma \vdash \Delta$, where $\Gamma$ and $\Delta$ are ordered lists (separated with ``,'') of expressions in the following grammar, where $V=\{A,B,C,\ldots\}$ is a set of atomic variables:
\[
S::= I \,\,|\,\, \bot \,\,|\,\, V \,\,|\,\, S \otimes S \,\,|\,\, S \parr S \,\,|\,\, S^* \,\,|\,\, \ls S
\]
This syntax includes left and right negation; we assume isomorphisms $\ls(S^*) \simeq S \simeq (\ls S)^*$ which, for simplicity, we generally suppress at the syntactic level. Atomic variables are represented as black dots, pointing in different directions depending on their side of the sequent:
\begin{calign}
\nonumber
\begin{tz}
\node (1) [inner sep=0pt, draw=none] at (0,0) {\tikzpng[scale=3, rotate=-90, extrascale]{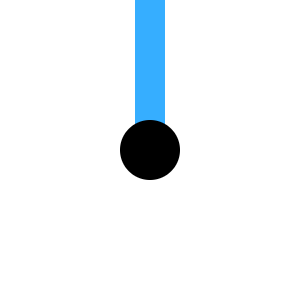}};
\draw [lightblue, dotted, line width=1.5pt] (1.east) to +(0.3,0);
\node [left] at (0,0) {$A$};
\end{tz}
&
\begin{tz}
\node (1) [inner sep=0pt] at (0,0) {\tikzpng[scale=3, rotate=90, extrascale]{blackdot}};
\draw [lightblue, dotted, line width=1.5pt] (1.west) to +(-0.3,0);
\node [right] at (0,0) {$A$};
\end{tz}
\\\nonumber
A \vdash \cdots & \cdots \vdash A
\end{calign}
The two sides of a sequent are represented graphically by trees, which are drawn connected together at their roots. The basic connective ``,'' is denoted as a blue vertex with zero or more branches to the left or right, as follows:
\tikzset{dot/.style={circle, draw=none, fill=black, inner sep=1.0pt}}
\begin{calign}
\nonumber
\begin{tz}[scale=0.76, extrascale]
\node [anchor=south west, inner sep=0pt] at (0,0) {\tikzpng[scale=3, rotate=90, extrascale]{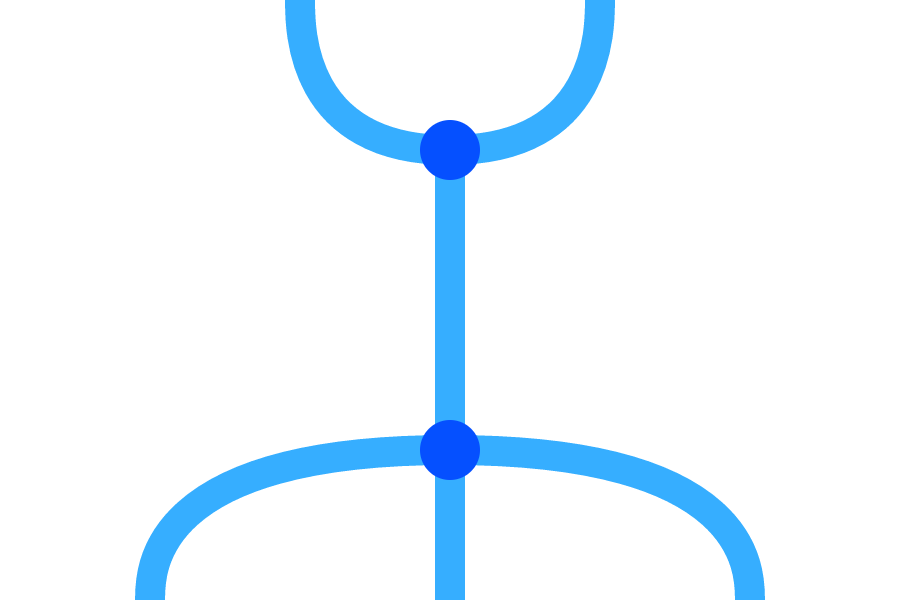}};
\node [dot] at (0,1) {};
\node [left] at (0,1) {$B$};
\node [dot] at (0,2) {};
\node [left] at (0,2) {$A$};
\node [dot] at (2,2.5) {};
\node [right] at (2,2.5) {$C$};
\node [dot] at (2,1.5) {};
\node [right] at (2,1.5) {$D$};
\node [dot] at (2,0.5) {};
\node [right] at (2,0.5) {$E$};
\end{tz}
&
\begin{tz}[scale=0.76, extrascale]
\node [anchor=south west, inner sep=0pt] at (0,0) {\tikzpng[scale=3, rotate=90, extrascale]{multform}};
\node [dot] at (2,1.5) {};
\node [right] at (2,1.5) {$A$};
\node [dot] at (2,0.5) {};
\node [right] at (2,0.5) {$B$};
\end{tz}
&
\begin{tz}[scale=0.76, extrascale]
\node [anchor=south west, inner sep=0pt] at (0,0) {\tikzpng[scale=3, rotate=-90, extrascale]{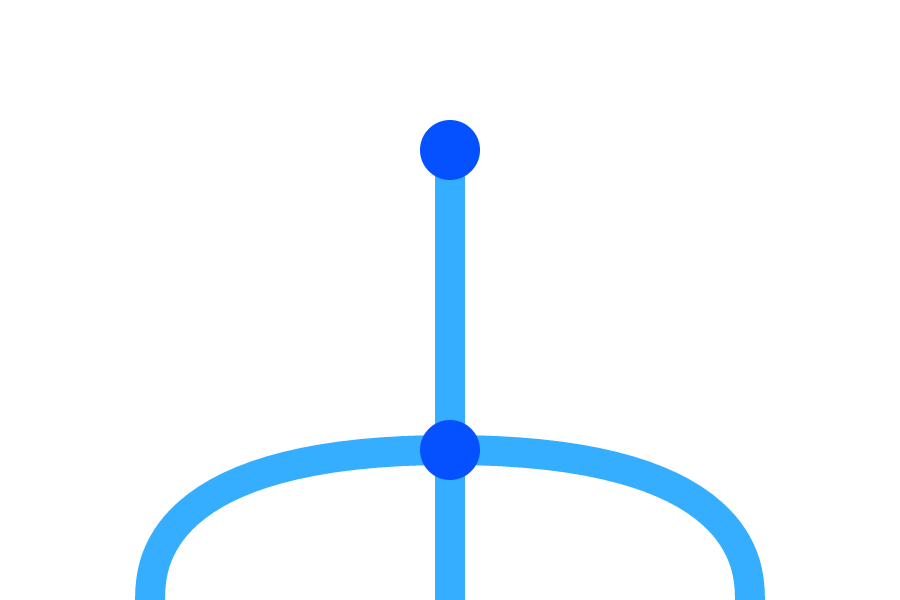}};
\node [dot] at (0,2.5) {};
\node [left] at (0,2.5) {$A$};
\node [dot] at (0,1.5) {};
\node [left] at (0,1.5) {$B$};
\node [dot] at (0,0.5) {};
\node [left] at (0,0.5) {$C$};
\end{tz}
\\*
\nonumber
A,B \vdash C, D, E
&
\vdash A,B
&
A, B, C \vdash
\end{calign}
The connectives $\otimes$ and $\parr$, which are always binary, are drawn in blue on their natural side (left for $\otimes$, right for $\parr$), and in red on the other side, as we show with the following examples:
\begin{calign}
\nonumber
\begin{tz}[scale=0.76, extrascale]
\node [anchor=south west, inner sep=0pt] at (0,0) {\tikzpng[scale=3, rotate=-90, extrascale]{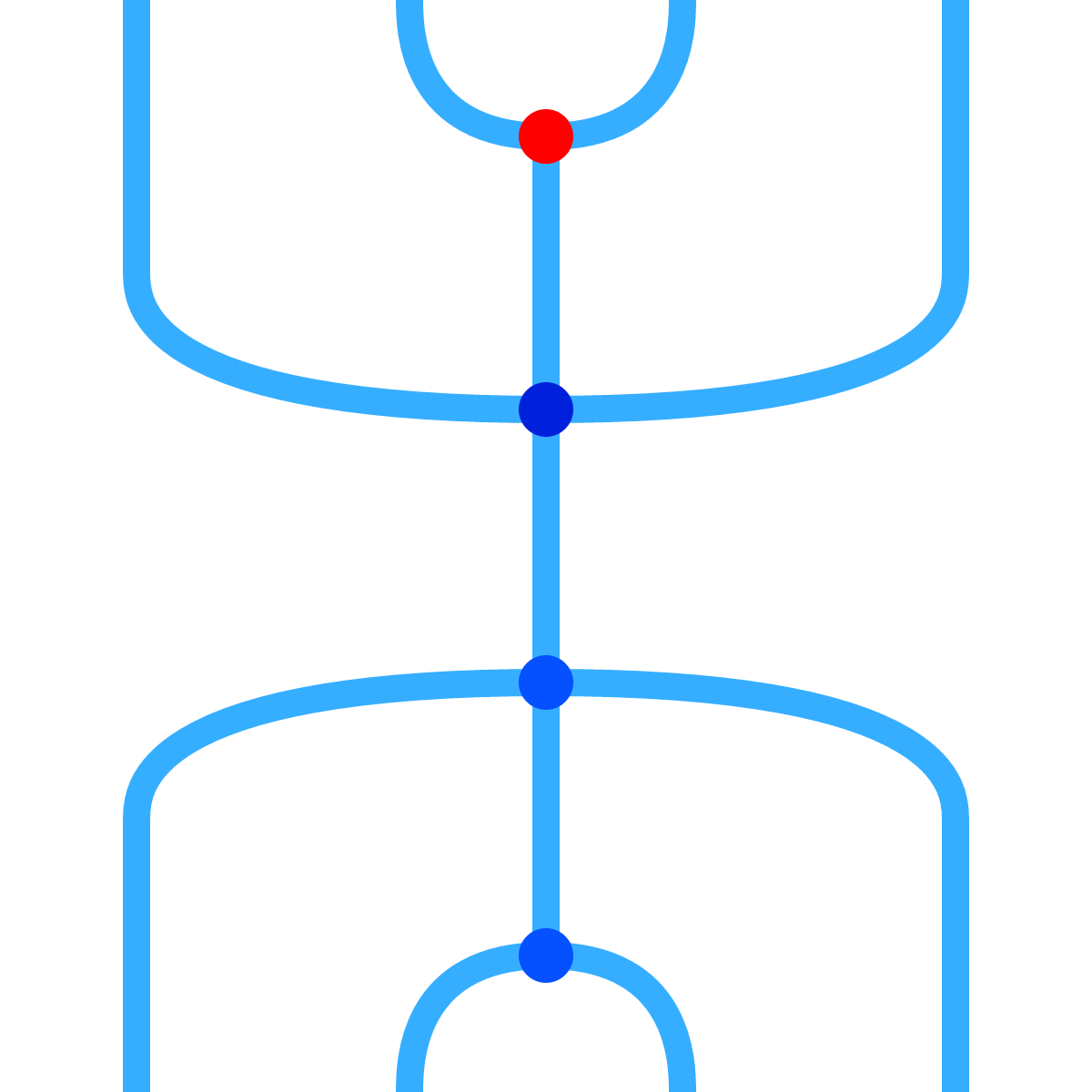}};
\node [dot] at (0,3.5) {};
\node [left] at (0,3.5) {$A$};
\node [dot] at (0,2.5) {};
\node [left] at (0,2.5) {$B$};
\node [dot] at (0,1.5) {};
\node [left] at (0,1.5) {$C$};
\node [dot] at (0,0.5) {};
\node [left] at (0,0.5) {$D$};
\node [dot] at (4,3.5) {};
\node [right] at (4,3.5) {$E$};
\node [dot] at (4,2.5) {};
\node [right] at (4,2.5) {$F$};
\node [dot] at (4,1.5) {};
\node [right] at (4,1.5) {$G$};
\node [dot] at (4,0.5) {};
\node [right] at (4,0.5) {$H$};
\end{tz}
&
\begin{tz}[scale=0.76, extrascale]
\node [anchor=south west, inner sep=0pt] at (0,0) {\tikzpng[scale=3, rotate=-90, extrascale]{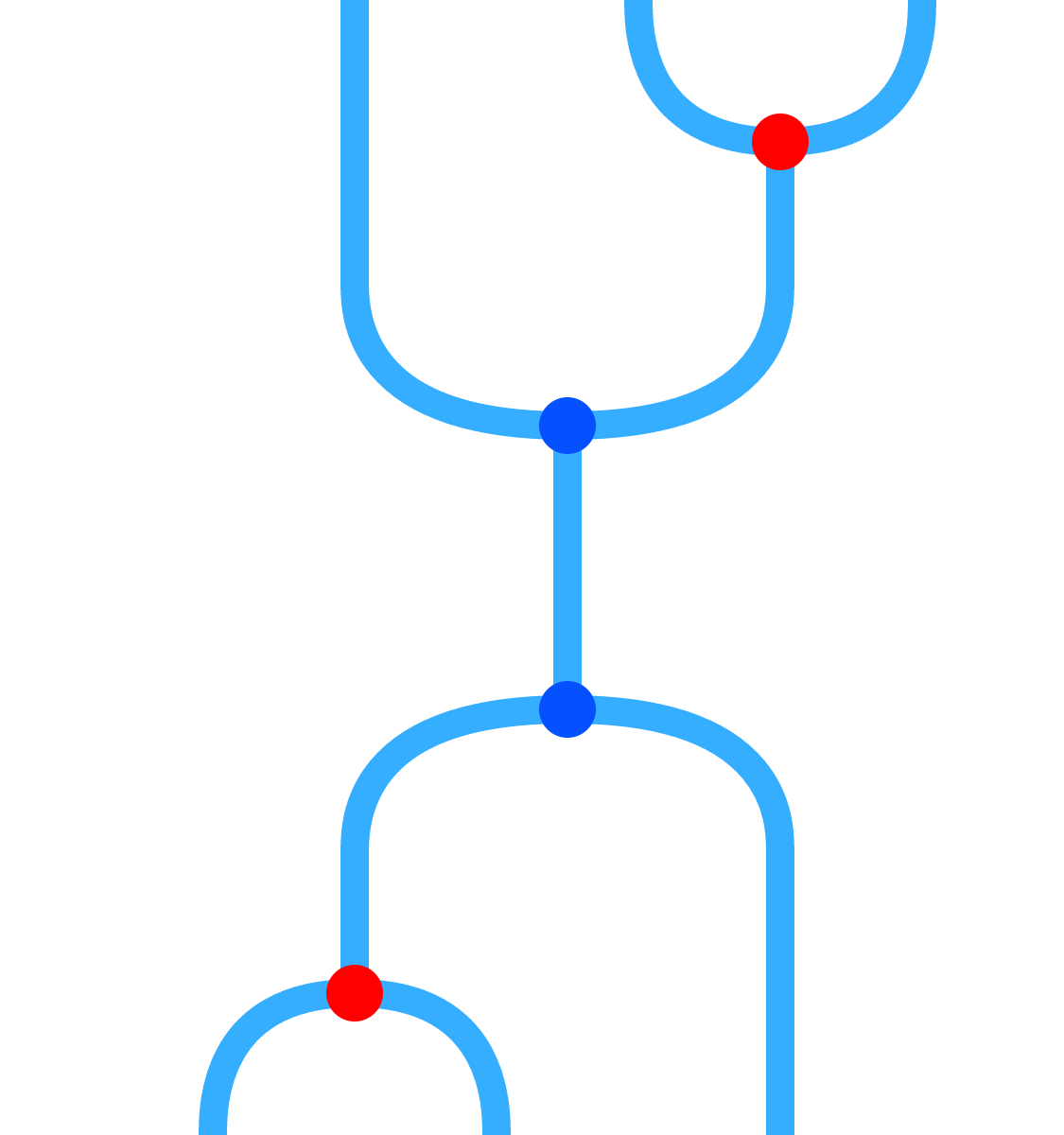}};
\node [dot] at (0,3) {};
\node [left] at (0,3) {$A$};
\node [dot] at (0,2) {};
\node [left] at (0,2) {$B$};
\node [dot] at (0,1) {};
\node [left] at (0,1) {$C$};
\node [dot] at (4,2.5) {};
\node [right] at (4,2.5) {$D$};
\node [dot] at (4,1.5) {};
\node [right] at (4,1.5) {$E$};
\node [dot] at (4,0.5) {};
\node [right] at (4,0.5) {$F$};
\end{tz}
\\*\nonumber
A, (B \otimes C), D \vdash E, (F \otimes G), H
&
(A \parr B) , C \vdash D \parr (E \otimes F)
\end{calign}
Note that a blue dot with a binary branching is therefore an overloaded notation; this is a deliberate feature.

The units $I$ and $\bot$ are represented by blue dots on their natural side (left for $I$, right for $\bot$), and red dots on the other side, as shown:
\begin{calign}
\nonumber
\begin{tz}[scale=0.76, extrascale]
\node [anchor=south west, inner sep=0pt] at (0,0) {\tikzpng[scale=3, rotate=-90, extrascale]{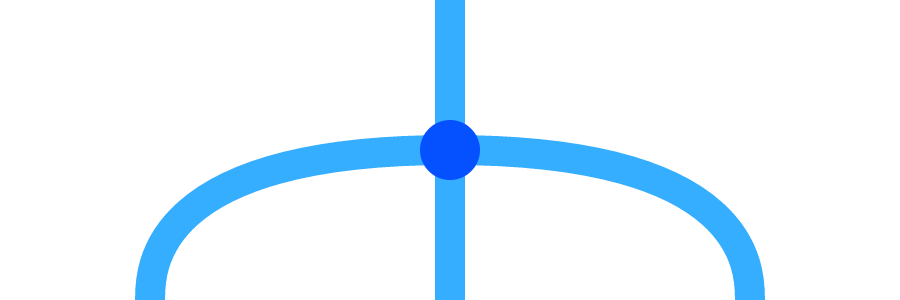}};
\node [dot] at (0,2.5) {};
\node [left] at (0,2.5) {$A$};
\node [dot, fill=darkblue] at (0,1.5) {};
\node [left] at (0,1.5) {};
\node [dot] at (0,0.5) {};
\node [left] at (0,0.5) {$B$};
\node [dot] at (1,1.5) {};
\node [right] at (1,1.5) {$C$};
\end{tz}
&
\begin{tz}[scale=0.76, extrascale]
\node [anchor=south west, inner sep=0pt] at (0,0) {\tikzpng[scale=3, rotate=90, extrascale]{comma}};
\node [dot, fill=red] at (0,1) {};
\node [left] at (0,1) {};
\node [dot] at (0,2) {};
\node [left] at (0,2) {$A$};
\node [dot] at (2,2.5) {};
\node [right] at (2,2.5) {$B$};
\node [dot, fill=red] at (2,1.5) {};
\node [right] at (2,1.5) {};
\node [dot, fill=darkblue] at (2,0.5) {};
\node [right] at (2,0.5) {};
\end{tz}
\\
\nonumber
A,I,B \vdash C
&
A, \bot \vdash B, I, \bot
\end{calign}
We represent $(-)^*$ as turning right by a half-turn, and ${}^* \hspace{-0.5pt} (-)$ as turning left by a half-turn, as shown:
\begin{calign}
\nonumber
\begin{tz}[scale=0.76, extrascale]
\node [anchor=south west, inner sep=0pt] at (0,0) {\tikzpng[scale=3, rotate=90, extrascale]{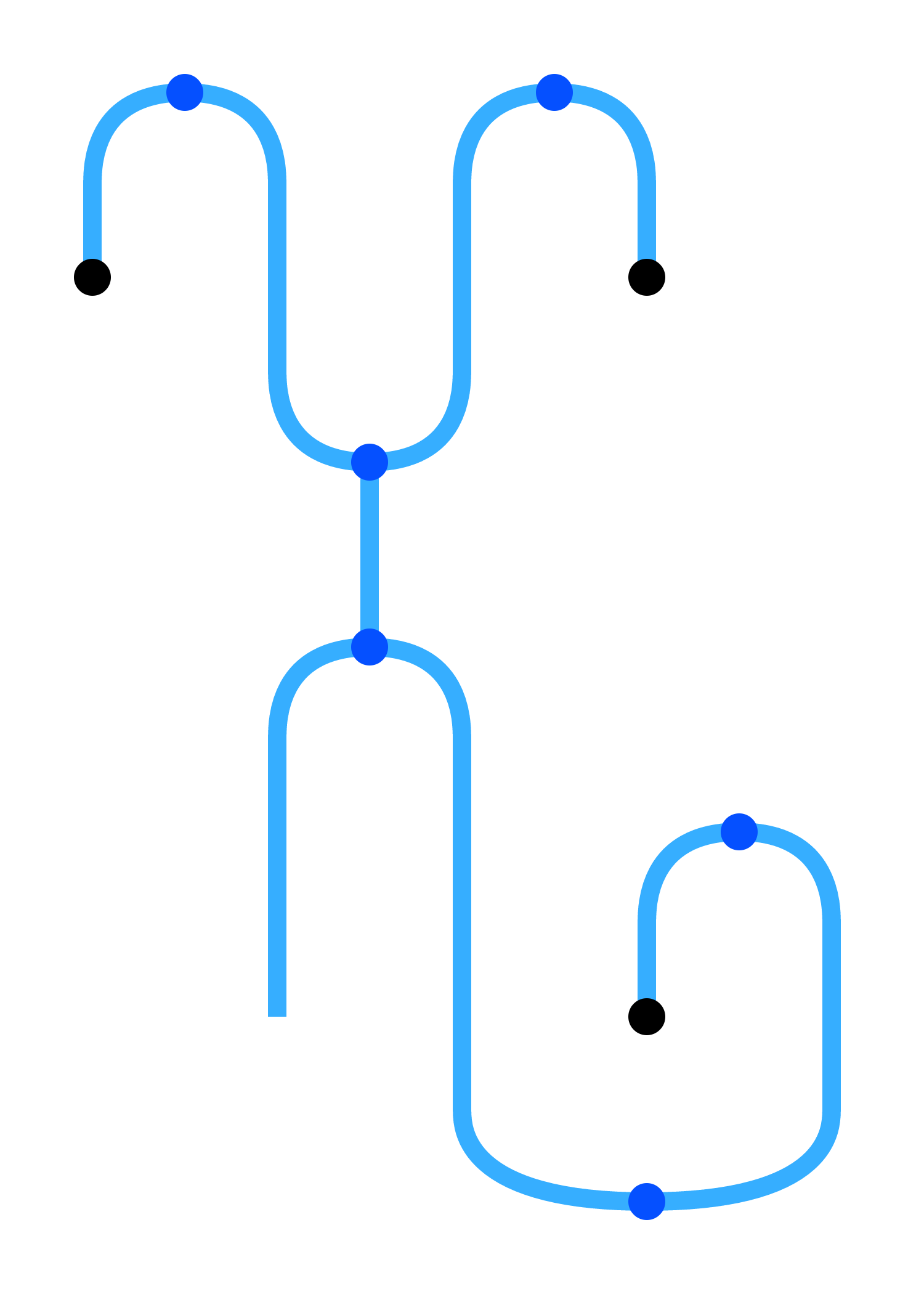}};
\node [right] at (1.5,3.5) {$A$};
\node [right] at (1.5,0.5) {$B$};
\node [dot] at (5.5,1.5) {};
\node [right] at (5.5,3.5) {$C$};
\node [right] at (5.5,1.5) {$D$};
\end{tz}
\\
\nonumber
A^* , \ls B \vdash \ls \,\ls \,C, D
\end{calign}
Diagrams built from  sequents in this way are of a simple kind; as graphs, they are all acyclic and connected. In general we can allow arbitrary well-typed composites of these components; such diagrams represent 1\-morphisms in the monoidal bicategory \free {\F^*}, described in \autoref{sec:coherence}.

\subsection{The 3\-dimensional calculus}

Diagrams in the 3\-dimensional calculus are volumes, surfaces,  lines and points embedded in $\R^3$. Formally they are expressions in the graphical calculus for Gray categories, which is by now well-developed~\cite{barrett-graydiagrams, bartlett-wire, hummon-thesis, CSPthesis}.\footnote{For us, these 3-dimensional diagrams play the role of a formal syntax. While this is nontraditional, it is adequate for our purposes.} This body of work establishes the soundness and completeness of the graphical calculus for reasoning about Gray categories~\cite[Theorem 2.26]{barrett-graydiagrams}, and in particular for semistrict monoidal 2\-categories. However, the technicalities are quite substantial, and the class of piecewise-linear diagrams about which the appropriate completeness theorem is proved is rather technical in flavour; for convenience, and since the work in this section serves primarily to illustrate an application of the main  results of \autoref{sec:coherence}, we work with a looser class of diagrams, for which a completeness result is not yet established, but can reasonably be expected to exist.

We now describe the 3\-dimensional calculus in more detail. Diagrams consist of \textit{sheets}, bounded on the left and right by \textit{edges}, which are bounded above and below by \emph{vertices}. (Sheets can also be bounded by the sides of the diagram, and edges can also be bounded by the top or bottom of the diagram.) Diagrams are immersed in 3\-dimensional space, meaning that sheets can exist in front or behind other sheets, and wires on sheets of different depths can cross; however, components never intersect. Here is an example:
\begin{equation*}
\tikzset{blob/.style={draw, circle, fill=black, inner sep=2pt}}
\tikzset{every picture/.style={scale=0.7}}
\begin{tz}[scale=1.5]
\draw [surface] (\xoff,\yoff) node (1) {} to (3+\xoff,\yoff) node (2) {} to (3+\xoff,3+\yoff) node (3) {} to (\xoff,3+\yoff) node (4) {} to (1.center);
\draw [black] (2,\yoff) to [out=up, in=down] (1,1+\yoff) node (x) {} to (1,3+\yoff);
\node [blob] at (x.center) {};
\draw [surface] (-\xoff,-\yoff) node (5) {} to (3-\xoff,-\yoff) node (6) {} to (3-\xoff,3-\yoff) node (7) {} to (-\xoff,3-\yoff) node (8) {} to (5.center);
\draw [black] (1,-\yoff) to (1,\yoff) to [out=up, in=down] (2,1+\yoff) to (2,2) to [out=up, in=down] (2,3-\yoff);
\node [blob] at (2,2) {};
\end{tz}
\end{equation*}
Here we have front and back sheets, each containing an edge, which contains a vertex. Towards the bottom of the picture, the wires appear to cross, although they do not intersect in 3\-dimensional space: this is called an \emph{interchanger}.

For our application to linear logic, we use two types of vertex: \emph{coherent vertices}, arising from invocations of the coherence result of \autoref{sec:coherence}; and \emph{adjunction vertices}, arising from the adjunctions listed in \autoref{lem:additionaladjunctions}.
\begin{itemize}
\item
\textbf{Coherent vertices.} Say that a 2\-dimensional calculus diagram is \textit{simple} when it is connected and acyclic with nonempty boundary, and in the blue fragment of the calculus, not involving red nodes or black atomic variable nodes. Then any two simple diagrams can be connected by a coherent vertex, denoted as follows:
\begin{calign}
\label{eq:surfaceexample}
\begin{tz}[scale=1.5]
\draw [surface] (-\xoff,2) node (1) {} to (0.8,2) node (2) {} to [out=45, in=left, looseness=0.4] (3-\xoff,2+\yoff) node (3) {} to (3-\xoff,\yoff) node (4) {} to [out=left, in=45, in looseness=0.6] (2,0) node (5) {} to (1,0) node (6) {} to [out=145, in=right] (-\xoff, \yoff) node (7) {} to (1.center) to (2.center);
\draw [surface] (2.center) to (1.5,0.8) node (8) {} to (2.2,2-\yoff) node (9) {} to (2.center);
\draw [surface] (6.center) to (8.center) to (9.center) to [out=-165, in=right] (\xoff,2-2*\yoff) node (14) {} to (\xoff,-\yoff) node (13) {} to [out=right, in=-155] (6.center);
\draw [surface] (9.center) to (3+\xoff,2-\yoff) node (15) {} to (3+\xoff,-\yoff) node (16) {} to [out=left, in=-35] (5.center) to (8.center) to (9.center);
\end{tz}
&
\begin{tz}
\node [inner sep=0pt] (1) at (0,1.5) {\tikzpng[scale=2, rotate=90]{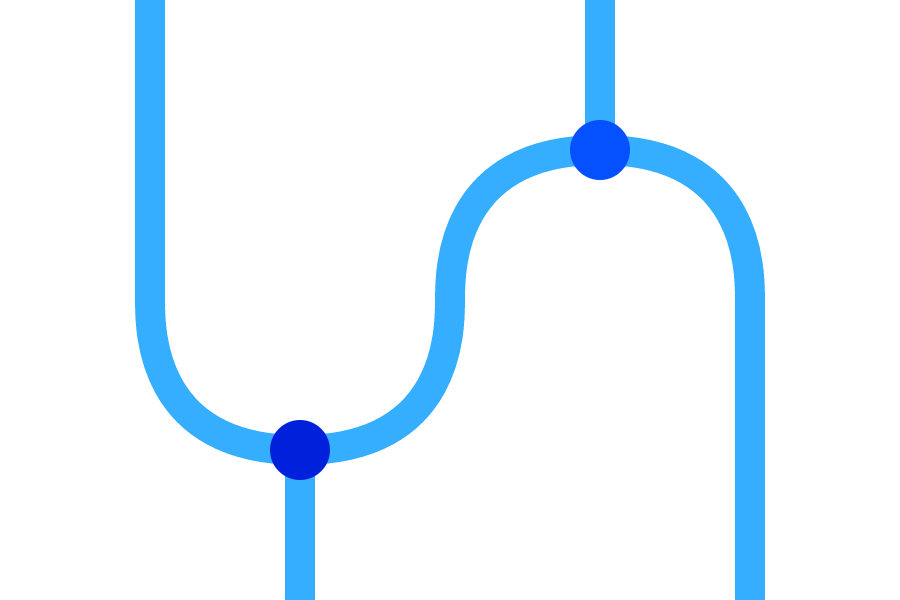}};
\node [inner sep=0pt] (2) at (0,0) {\tikzpng[scale=2, rotate=90]{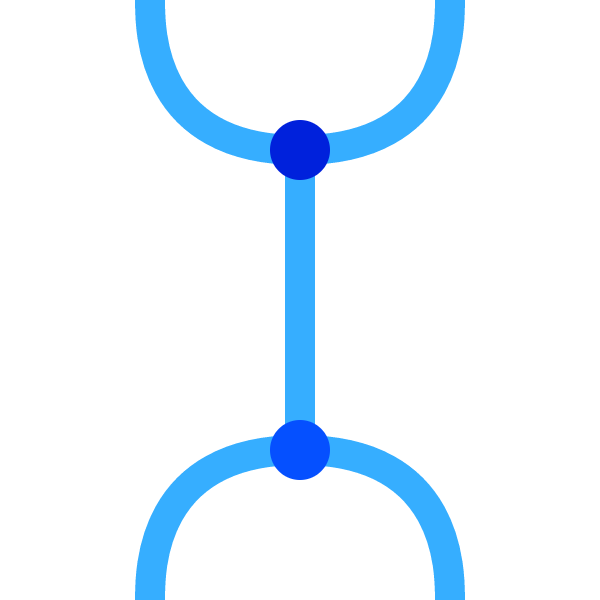}};
\draw [->, shorten <=-2pt, shorten >=-2pt] (1) to (2);
\end{tz}
\end{calign}
On the left we give the surface representation, and on the right we give the 2\-dimensional calculus representation of the upper and lower boundaries. The coherent vertex is the point in the middle of the surface diagram where 4 edges meet.
\item
\textbf{Adjunction vertices.} Listed in \autoref{fig:3dgenerators}, these introduce and eliminate red and black edges in the surface calculus.
\end{itemize}

We now define equivalence in the graphical language, giving {intuitive} interpretations of each generating relation in italics. 
\begin{definition}
\label{def:surfaceequivalence}
Two surface diagrams are \textit{equivalent} when they are related by the least equivalence relation generated by the following:
\begin{itemize}
\item \textbf{Coherence.} Let $P,Q$ be surface diagrams built from interchangers and coherent vertices, whose upper and lower boundaries are simple 2\-dimensional diagrams; then $P=Q$. \emph{(All acyclic equations of coherent vertices hold.)}
\item \textbf{Adjunction.} The equations listed in \autoref{fig:3dequations} hold. \emph{(Bent wires can be pulled straight.)}
\item \textbf{Isotopy.} The equations of a monoidal bicategory hold. \emph{(If two diagrams are ambient isotopic, they are equivalent.)}
\item \textbf{Locality.} Suppose surface diagrams $P,Q$ differ only with respect to subdiagrams $P',Q'$, with $P'=Q'$. Then $P=Q$. \emph{(Equivalence applies locally in the interior of a diagram.)}
\end{itemize}
\end{definition}

\noindent
It is a fair summary of this definition to say that two diagrams are equivalent just when one can be \textit{deformed} into the other.

Our presentation here is informal, but we emphasize that our definition of the surface calculus and its equivalence relation can be made completely precise in terms of the formal development of \autoref{sec:coherence}: the diagrams represent composite 2\-morphisms in $\free {\F^*}$, and two diagrams are equivalent just when their corresponding 2\-morphisms are equal.

\subsection{Interpreting the sequent calculus}

We saw in \autoref{sec:2dcalculus} how individual sequents in multiplicative linear logic can be interpreted as 2\-dimensional diagrams. We now see how proofs can be interpreted as 3\-dimensional surface diagrams. We view these surfaces as directed from top to bottom, just like traditional sequent calculus proofs; so for a particular surface, its \textit{hypothesis} is the upper boundary, and its \textit{conclusion} is the lower boundary.

We use a basis for the sequent calculus with a symmetry between introduction and elimination for $\otimes$, $\parr$, $I$ and $\bot$; the rules $\otimes$\-R, $\parr$\-L, $I$\-R and $\bot$\-L are derivable (see \autoref{ex:additionalrules}.) Furthermore, we include only CUT rules with minimal overlapping contexts; the more general CUT rules are derivable using negation. (These two features account for the differences between our presentations and others in the literature~\cite{Abrusci:1991}.) The interpretation of AXIOM and CUT rules are given recursively in \autoref{fig:3dcompound}, with black wires standing for atomic variables and green wires standing for general variables; the interpretation of the remaining rules, which we call the \emph{core fragment} of the logic, is given in \autoref{fig:sequentscore}.

We now establish the main theorem that controls proof equivalence for this 3\-dimensional notation.
We now prove our main theorem regarding the application to linear logic.
\begin{theorem}
\label{thm:maintheorem}
Two sequent proofs in multiplicative linear logic have equal interpretations in the free $*$-autonomous category just when their surfaces diagrams are equivalent.
\end{theorem}
\begin{proof}
The generators in \autoref{fig:3dgenerators}, and the equations in \autoref{fig:3dequations}, are precisely the definition of an $\F^*$\-structure rendered  in the surface calculus. (In fact they are redundant, since they include additional adjunction equations which \autoref{lem:additionaladjunctions} shows are already derivable.) Under the equivalence of \autoref{prop:dataequivalent}, it is essentially immediate that this reduces to the standard interpretation of the sequent calculus in a $*$-autonomous category~\cite{seely-linearlogic}: CUT is by composition and AXIOM is the identity morphism, mediated by the adjunctions $V_* \dashv V^*$ of profunctors induced by a variable defined by a functor $V:\cat 1 \to \cat C$; $(\parr,\bot)$ is the monoidal product of the  $*$\-autonomous category; the negations come from the biclosed structure; and $(\otimes,I)$ is the dual monoidal structure induced by negation.
\end{proof}

It is interesting to analyze the different contributions to proof equivalence made by each part of \autoref{def:surfaceequivalence} of surface equivalence. \textbf{Coherence} tells us that any two proofs built in the virtual part of the logic given in \autoref{fig:sequentscore} are equal. \textbf{Adjunction} tells us that AXIOM and CUT cancel each other out, both for atomic and compound variables. \textbf{Isotopy} tells us that that `commutative conversion' is possible, where we exchange the order of multiple independent hypotheses. \textbf{Locality} tells us that we can apply our equations in the context of a larger proof, in the manner of deep inference~\cite{guglielmi-deepinference}.

We give a formal statement of coherence for the virtual fragment of the logic, since it is a result of independent interest.
\begin{corollary}
If two sequent proofs in the virtual fragment of the logic given in \autoref{fig:sequentscore} have the same hypotheses and conclusion, then they are equal in the free $*$\-autonomous category.
\end{corollary}

\noindent
While well-known to experts, and derivable from existing results~\cite[Proposition~2.1.9]{Lamarche_2006}, we cannot find this statement explicitly elsewhere in the literature. Furthermore, our proof method is certainly novel.


We comment on some interesting features of the translation between the sequent calculus and the surface calculus. The fundamental simplicity of the surface calculus is clear, from the minimality of the data in \autoref{fig:3dgenerators}, as compared to Figures~\ref{fig:3dcompound} and~\ref{fig:sequentscore}. Partly this is achieved by the greater degree of locality: for example, the cut rules for $\ls V$ and $V^*$ are both interpreted using the same surface generators, composed in different ways. But more significantly, the entire virtual fragment of the sequent calculus is interpreted in the \emph{trivial} part of the surface calculus, which could be regarded as significantly reducing the bureaucracy of proof analysis, to use Girard's phrasing~\cite{girard}. To make the most of these advantages, we suggest that the surface calculus can  serve {directly} as a toolkit for logic, not just as a way to visualize sequent calculus proofs.

\input{sequents-core.extra}
\input{timesRproof.extra}

\subsection{Examples}
\label{sec:examples}

We now look in detail at a number of examples: we derive the surface form of the missing $\otimes$\-R rule; we analyze equivalence of a proof involving units; and we investigate the classic triple-dual problem.

\begin{example}[Additional rules]
\label{ex:additionalrules}
Presentations of multiplicative linear logic usually include the rules $\otimes$-R, $\parr$-L, $I$-R, $\bot$-L, which are missing from \autoref{fig:3dcompound} and \autoref{fig:sequentscore}; however, they are derivable. We analyze $\otimes$-R in detail in \autoref{figure:tensorright}. On the left-hand side, we derive the rule in our chosen basis for the the sequent calculus. In the middle image, we interpret it in the surface calculus, using the rules we have described. In the third image, we simplify the surface calculus interpretation using the rules in \autoref{fig:3dequations}. From this simplified diagram, we see that it does not in fact involve the variables, the nontrivial generators being applied in the central part of diagram only. Simple interpretations of the other 3 rules can be derived similarly.
\input{unitproof.extra}
\end{example} 

\begin{example}
[Equivalence of two proofs with units]
The example is given in \autoref{figure:proofswithunits}. We present two distinct sequent proofs of the tautology $A, B \vdash \bot \parr (A \otimes B)$, along with their corresponding surface proofs. The heights are aligned to help understand how the surface proofs have been constructed. We make use of the $\otimes$-R rule derived in the previous example.

It can be seen by inspection that the surface proofs are equivalent, as follows. Starting with the surface on the left, we allow the $\bot$-introduction vertex to move up and to the left; this is an application of \textbf{Coherence} and \textbf{Locality}. We also allow the $B$-introduction vertex at the top of the diagram to move down, behind both the $A$-introduction and $\bot$-introduction vertices; this is an application of \textbf{Isotopy}.
\end{example}

\begin{example}[Triple-dual problem]
\label{ex:tripleunit}
Starting with the identity \mbox{$A \multimap X \to A \multimap X$}, we can uncurry on the left to obtain a morphism \mbox{$A \otimes (A \multimap X) \to X$}, and curry on the right to obtain a morphism $p_A : A \to X \multimapinv (A \multimap X)$; in a similar way, we can also define a morphism $q_A : A \to (X \multimapinv A) \multimap X$. Then the \textit{triple-dual problem}, originally due to Kelly and Mac Lane~\cite{Kelly_1979} and generalized here to the non-symmetric setting, is to determine whether the following equation holds:
\begin{equation}
\label{eq:tripleunit}
\begin{aligned}
\begin{tikzpicture}[xscale=1.5,yscale=2]
\node (1) at (0,0) {$X \multimapinv ((X \multimapinv A) \multimap X)$};
\node (2) at (4,0) {$X \multimapinv A$};
\node (3) at (4,-1) {$X \multimapinv ((X \multimapinv A) \multimap X)$};
\draw [->] (1) to node [above] {$X \multimapinv q_A$} (2);
\draw [->] (2) to node [right] {$p_{X \multimapinv A}$} (3);
\draw [->] (1) to node [below] {$\id$} (3);
\end{tikzpicture}
\end{aligned}
\end{equation}
We analyze this equation in the case that $X=\bot$; then $q_A$ and $p_A$ are the isomorphisms $\ls(A^*) \simeq A \simeq (\ls A)^*$. We give the surfaces for the clockwise and anticlockwise paths of \eqref{eq:tripleunit}:
\input{tripledual.extra}

\noindent
We conclude in this case that the proofs are equivalent by the \textbf{Coherence} rule. This deduction can be readily identified by eye, and so is a single-step deduction in a natural sense. We contrast this with the treatments of Blute et al~\cite[Section~4.2]{blute-coherence} and Hughes~\cite[Example~2]{hughes-simplestar} in terms of proof nets, where the proofs require several deductive steps, and the overall deduction is far from immediate.
\end{example}


\bibliographystyle{abbrv}
\bibliography{references}

\ignore{

}

\end{document}